\documentclass[twocolumn,comsoc]{IEEEtran}
\usepackage{url}
\usepackage[T1]{fontenc}% optional T1 font encoding
\usepackage{graphicx}
\usepackage{bm}
\usepackage{cite}
\usepackage{multirow}
\usepackage{enumerate}
\usepackage{algorithm}
\usepackage{algorithmic}
\usepackage{setspace}
\usepackage{mathrsfs}
\usepackage{amsthm}
\usepackage{mathdots}
\usepackage{booktabs}
\ifCLASSINFOpdf

\else

\fi
\usepackage{subfigure}
\usepackage{amsmath}

\usepackage[cmintegrals]{newtxmath}

\newtheorem{definition}{Definition}
\newtheorem{theorem}{Theorem}
\newtheorem{lemma}{Lemma}
\newtheorem{corollary}{Corollary}
\theoremstyle{remark} 
\newtheorem{remark}{remark}
\newtheorem{example}{Example}
\begin{document}

\title{
%Analysis and Enhancement for Geometric Detector via Dual Manifold of Toeplitz Hermitian Positive Definite Matrix Manifold 
Dual Power Spectrum Manifold and Toeplitz HPD Manifold: Enhancement and Analysis for Matrix CFAR Detection
%Dual Power Spectrum Manifold Based Analysis and Enhancement for Geometric Detector on Toeplitz Hermitian Positive Definite Matrix Manifold
%Affine Invariant Geometric Measure on Toeplitz Hermitian Positive Definite Matrix Manifold and Dual Power Spectrum Manifold for Target Detection
}

\author{Hao Wu, Yongqiang Cheng, Xixi Chen, Zheng Yang, Xiang Li and Hongqiang Wang
\thanks{This work was supported by the NSFC of China under grant No. 61871472. (Corresponding author: Yongqiang Cheng.)\par
Hao Wu, Yongqiang Cheng, Xixi Chen, Zheng Yang, Xiang Li, Hongqiang Wang are with the College of Electronic Science and Technology, National University of Defense Technology, Changsha 410000, China (Email: wu\_hao\_95@163.com; nudtyqcheng@gmail.com; yangzheng18@nudt.edu.cn; xixichen99@163.com; lixiang01@vip.sina.com; oliverwhq@tom.com).}
}

\markboth
{IEEE,~Vol.~x, No.~x, \today}
{Hao Wu\headeretal: The Dual Manifold of Toeplitz Hermitian Positive Definite Matrix Manifold for Matrix CFAR Detection}

\maketitle

\begin{abstract}
Recently, an innovative matrix CFAR detection scheme based on information geometry, also referred to as the geometric detector, has been developed speedily and exhibits distinct advantages in several practical applications. 
These advantages benefit from the geometry of the Toeplitz Hermitian positive definite (HPD) manifold $\mathcal{M}_{\mathcal{T}H_{++}}$, but the sophisticated geometry also results in some challenges for geometric detectors, such as the implementation of the enhanced detector to improve the SCR (signal-to-clutter ratio) and the analysis of the detection performance.
%, which impede the further research of this detection scheme.
To meet these challenges, this paper develops the dual power spectrum manifold $\mathcal{M}_{\text{P}}$ as the dual space of $\mathcal{M}_{\mathcal{T}H_{++}}$.
%This paper develops the dual power spectrum manifold to circumvent these problems. Compared with the Toeplitz HPD manifold, the dual power spectrum manifold has a lower dimension and the points on this manifold can directly relate to the power spectrum of the observed data.
%possesses strong relation with target echoes. 
For each affine invariant geometric measure on $\mathcal{M}_{\mathcal{T}H_{++}}$, we show that there exists an equivalent function named induced potential function on $\mathcal{M}_{\text{P}}$. 
By the induced potential function, the measurements of the dissimilarity between two matrices can be implemented on $\mathcal{M}_{\text{P}}$, and the geometric detectors can be reformulated as the form related to the power spectrum.
The induced potential function leads to two contributions: 1) The enhancement of the geometric detector, which is formulated as an optimization problem concerning $\mathcal{M}_{\mathcal{T}H_{++}}$, is transformed to an equivalent and simpler optimization on $\mathcal{M}_{\text{P}}$.
%By the induced potential function, the enhancement of the matrix CFAR detection, which is formulated as an optimization problem with respect to the Toeplitz HPD matrix, is transformed to an equivalent and simpler optimization on the dual power spectrum manifold. 
In the presented example of the enhancement, the closed-form solution, instead of the gradient descent method, is provided through the equivalent optimization.
%In the presented examples for the $m\times m$ dimensional matrix, the computational complexity for solving the equivalent optimization is $\mathcal{O}(m)$, which is $\mathcal{O}(km^3)$ for the original optimization.
2) The detection performance is analyzed based on $\mathcal{M}_{\text{P}}$, and the advantageous characteristics, which benefit the detection performance, can be deduced by analyzing the corresponding power spectrum to the maximal point of the induced potential function. 
\end{abstract}


\IEEEpeerreviewmaketitle

\section{Introduction}
%\IEEEPARstart{R}{ecently}, an innovative matrix CFAR detection scheme based on information geometry has been developed and applied in the target detection issues in radar, sonar, navigation and communication fields.

\IEEEPARstart{D}{etection} of targets, especially in the clutter background, is a common challenge in radar, sonar, navigation and communication fields, and thus has been studied and gained much interest during the last decades\cite{Zeitouni1992,Stoica2004,Chen2009,Moustakides2012,Pascal2008,Pascal2010,Pascal2011,Rong2021}. Recently, an innovative matrix CFAR detection scheme based on information geometry\cite{Amari2016} is developed speedily and exhibits distinct advantages in the low SCR (signal-to-clutter ratio) regime, short pulse sequences, heterogeneous clutter backgrounds, etc\cite{Arnaudon2013Riemannian,zhao2019,Zhang2019,HUA2021}. 
In this detection scheme, the detection issues are transformed to the geometric problems, that the observed data are modeled as covariance matrices and treated as points on the matrix manifold.
Then, the decision is made based on the geometric measurement of the dissimilarity between the corresponding points of primary data and secondary data on the matrix manifold.
%, because the value of the geometric measure implies whether the presence or absence of the target echo in the primary data.
%Then, the decision of these geometry-based detection methods is made by the geometric measure between the primary data and secondary data.
The Riemannian distance (RD)\cite{Said2017,Said2018} is the initial geometric measure in the matrix CFAR detection scheme, and it is
%Based on the Riemannian geometry, the Riemannian distance (RD)\cite{Said2017,Said2018} is employed as the geometric measure and 
applied in target detection with X-band and high frequency surface wave radars\cite{Lapuyade2008Radar,Arnaudon2013Riemannian}, burg estimation of the scatter matrix\cite{Burgestimation} and the monitoring of wake vortex turbulences\cite{Barbaresco2010,Liu2013}. 
%Inspired by the information theory, matrix theory and other significant theories, 
To improve the detection performance and circumvent the expensive computation, the matrix CFAR detection scheme has been extended by replacing the RD with other different geometric measures, such as Kullback-Leibler divergence (KLD)\cite{hua2017matrix}, Jensen-Shanon divergence (JSD)\cite{hua2018information2}, log-determinant divergence (LDD)\cite{Cherian2013} and their modified extension\cite{hua2018geometric,hua2018information,hua2017geometric}, etc\cite{zhao2019,Zhao2018}. By these pioneered works, the geometry-based detection method has been widely investigated from both theoretical and practical perspective, and is developing to a critical and effective approach in detection issues because of the extraordinary detection performance and little requirement of the statistical characteristics of clutter environment\cite{Hua_TBD,HUA2021,Wong2017,Barbaresco2019,Barbaresco2008}.
%Nowadays, it is developing to a critical and effective approach in detection issues, and has been applied widely because of the well detection performance and little requirement of the statistical characteristics of clutter environment\cite{HUA2021,Wong2017,Barbaresco2019,Barbaresco2008}. 
Therefore, the matrix CFAR detection scheme has a wide application prospect and is deserved to be further studied.\par
The geometry-based detection methods benefit from the high-dimensional and non-flat geometry of the matrix manifold, which can more accurately describe the intrinsic dissimilarity among the matrices. However, the sophisticated geometry also results in the complex principle and expressions of these detectors. Moreover, there are twofold challenges in these methods.
%Actually, the matrix manifold is a high-dimensional and non-linear space, which results in the sophisticated principle and structure of the geometric detector. So that, there are the twofold challenges in the geometry-based detection methods.
%However, there are many problems in the geometry-based detection method.
The first challenge concerns the low SCR regime and the enhancement of the geometry-based detection method. Under the low SCR regime, the target echo and the clutter are hard to discriminative. An effective way is to map the manifold into a more discriminative space. But, the sophisticated geometry leads to the difficulty of the determination of the effective enhanced mapping.
%The enhanced mapping is deduced from an optimization problems and often solved by the gradient descent methods with numerous iterations. Thus, the calculation of the enhanced mapping is complex, and the properties of such mapping are hard to analyze. 
Another is about the analysis of the detection performance. The complex structure of these detectors lead to the absence of the analytic and closed PDF of the test statistic. Thus, the detection performance of the geometry-based detection method is hard to deduce. Moreover, because of the unavailable analysis of the detection performance, the relation between the performance and characteristics of the target echo cannot be studied, which would impede the further research about the geometry-based detection methods.\par
%mainly focuses on the benefits from the geometric structure of the matrix manifold. The analysis based on the aspect of the manifold structure is significant and intrinsic in terms of geometry, but there is still a question concerning the advantaged scenario of the selected geometric measure.\par
% The first one is that the analysis approach of the geometric detector is missing because of the absent analytic and closed PDF of the test statistics. In previous works, the detection performance mainly focuses on the benefits from the geometric structure of the matrix manifold. The analysis based on the aspect of the manifold structure is significant and intrinsic in terms of geometry, but there is still a question concerning the advantaged scenario of the selected geometric measure. The second problem concerns about the enhancement of the geometric measures. The enhanced mapping is deduced from an optimization problems and often solved by the gradient descent methods with numerous iterations. Thus, the calculation of the enhanced mapping is sophisticated, and then the properties of such mapping are hard to analyze.\par
This paper mainly focuses on these two challenges and seeks for some new theoretical tools to deal with them. In the traditional signal theory, the duality of two signal spaces often play a pivotal role in the signal processing issues. For example, by the Fourier transformation, the time domain and frequency domain can be mutually transformed from each other, and the analysis and filtering operations of the signal on the time domain can be easily implemented on the frequency domain. Inspired by the duality between the time domain and the frequency domain, the dual manifold of the matrix manifold is desired to be investigated for the new theoretical tools.\par
%As for the traditional signal theory, the signal can be modeled on both the time domain and frequency domain. By the Fourier transformation, the time domain and frequency domain can be mutually transformed from each other. In addition, the analysis and filter operations of the signal on the time domain can be implemented on the frequency domain. Deriving from the strong relation between the time domain and frequency domain, the matrix manifold can also be related to other formed manifold.\par
This paper is mainly concentrated on the matrix CFAR detection scheme, and the covariance matrix is considered as a Toeplitz Hermitian positive definite (HPD) matrix and located on the Toeplitz HPD manifold $\mathcal{M}_{\mathcal{T}H_{++}}$.
% because the empirical means is the usually used method for covariance matrix estimation and the output of it is definitely Toeplitz-Toeplitz structure formed. 
% Moreover, the geometric measures are supposed to be endowed with the affine invariant property on the Toeplitz HPD manifold, e.g. RD, KLD, JSD and LDD. 
As the dual manifold of $\mathcal{M}_{\mathcal{T}H_{++}}$, the dual power spectrum manifold $\mathcal{M}_{\text{P}}$ is established by using the Wiener-Khinchin theorem.
By the duality between these two manifolds, every geometric detectors with affine invariant geometric measures can be equivalently transformed to a relatively straightforward form with the induced potential function on the dual power spectrum manifold.
% The dual power spectrum manifold has a lower dimension than that of the Toeplitz HPD manifold, and there exists an equivalent induced potential function on the dual power spectrum manifold for each affine invariant geometric measures on the Toeplitz HPD manifold. 
Moreover, the duality between these two manifolds can transform the mentioned twofold challenges to relatively simple problems on the dual power spectrum manifold. Specifically, the main contributions of this paper are listed as follows:
\begin{enumerate}
	\item \textbf{Toeplitz HPD manifold}: For the geometry-based detection methods, the covariance matrix of observed data is shown to be located on a Toeplitz HPD matrix manifold. On the Toeplitz HPD manifold, the geometric measures are established to quantify the intrinsic dissimilarity between two Toeplitz HPD matrix. Moreover, the definition of the affine invariant measure is provided, and we show that the commonly used RD, KLD, JSD and LDD belong to the affine invariant measures. 
	%Because of the high-dimensional and non-linear structure, the matrix manifold is sophisticated which results in the difficulty of the further research for geometric detectors. 
	\item \textbf{Dual power spectrum manifold}: As the lower-dimensional dual manifold of the Toeplitz HPD manifold, the dual power spectrum manifold is established. The relation between the Toeplitz HPD manifold and dual power spectrum manifold is studied, that the power spectrum of observed data equals to the Fourier transform of the first row of the corresponding Toeplitz HPD matrix. Moreover, another significant property between these two manifold is shown, that the eigenvalues of the Toeplitz HPD matrix asymptotically tend to the corresponding power spectrum for each observed data.
	\item \textbf{Induced potential function}: By the duality between the two manifolds, we have shown: for every affine invariant geometric measure on $\mathcal{M}_{\mathcal{T}H_{++}}$, there exists an equivalent function on $\mathcal{M}_{\text{P}}$, which is defined as the induced potential function for $\mathcal{M}_{\text{P}}$. Therefore, by the induced potential function, the measurements of the dissimilarity between two matrices can be implemented on $\mathcal{M}_{\text{P}}$, and the geometric detectors can be reformulated as the form related to the power spectrum. The form with the induced potential function is more straightforward than it with the geometric measure, so the induced potential function is of great potentials to address the challenges concerning the geometric detector.
%	From the Toeplitz HPD manifold to the dual power spectrum manifold, the so-called whitening spectrum mapping based on the whitening processing, which maps a pair of matrices to a power spectrum, is established to represent the dissimilarity between the two matrices. Moreover, a theorem have been shown: for every affine invariant geometric measure, there exists an induced potential function on the dual power spectrum manifold that this geometric measure is equivalent to the composition of the induced potential function and the whitening spectrum mapping.
%	Moreover, by this mapping, a theorem is showed: for each affine invariant measure, there exists a induced potential function on the dual power spectrum manifold, that the affine invariant measure between two matrices is equivalent to the induced potential function acting on the resulting power spectrum of the whitening spectrum mapping. 
%	Therefore, by the induced potential function and whitening spectrum mapping, the measurement between two Toeplitz HPD matrix can be implemented on the dual power spectrum manifold, and the geometric detectors can be reformulated as the form related to the power spectrum.
	\item \textbf{Equivalent enhancement on dual power spectrum manifold}: To improve the detection performance in the low SCR regime, the enhancement of the matrix CFAR detection is proposed and formulated as an optimization problem with respect to the Toeplitz HPD matrix. By the duality between the Toeplitz HPD manifold and the dual power spectrum manifold, we transform this optimization problem to an equivalent optimization on the dual power spectrum manifold. This equivalent optimization problem is simpler and easier to solve than the original optimization because the decision variable has lower dimension and the constraint of each component is independent. In the presented example, the closed-form solution of the equivalent optimization is figured out, that it was used to be solved by the gradient descent method from the perspective of the original optimization.
%	In the presented examples for $m\times m$ dimensional matrix, the computational complexity for solving the equivalent optimization is $\mathcal{O}(m)$, which is $\mathcal{O}(km^3)$ for the original optimization ($k$ denotes the number of iteration for the gradient descent method).
	\item \textbf{Performance analysis based on dual power spectrum manifold}: The analytic method for detection performance is proposed based on the induced potential function. Because the geometric detector using the affine invariant geometric measure can be reformulated as the form based on the induced potential function, the detection performance is related to the value of the induced potential function. Therefore, the advantageous characteristics, which benefit the detection performance, can be deduced by analyzing the corresponding power spectrum to the maximal point of the induced potential function. 
%	For discussing the relation between the detection performance and characteristics of the target echo, the test statistics of geometric detectors are reformulated as the form of induced potential functions. In this form, the change of the detection performance with different target echoes is related to the varying of the induced potential function with respect to the power spectrum. Because the partial derivative can reflect the change trend of functions, the analytic method based on the partial derivative of the induced potential function is proposed for the detection performance analysis. 
	As the examples, the detection performance of RD, KLD and LDD is analyzed by this method. Moreover, by the induced potential function, the optimal dimension of the enhanced mapping in the enhanced geometric detector is deduced.
\end{enumerate}\par
The remainder of this paper is organized as follows. Section~\ref{sec:form} reviews the covariance matrix model of the observed data for target detection issues, and the matrix formed binary hypothesis model is introduced. In section~\ref{sec:THPD_SP}, the Toeplitz HPD manifold, the dual power spectrum manifold, the affine invariant measure and the induced potential function are introduced. Section~\ref{sec:enhancement} reports on the equivalent enhancement on the dual power spectrum manifold. Section~\ref{sec:perform} presents the analytic method based on the dual power spectrum manifold for detection performance. Finally, section~\ref{sec:con} provides the brief conclusions of the paper.\par
The following notations are adopted in this paper: the math italic $x$, lowercase bold italic $\bm{x}$ and uppercase bold $\mathbf{A}$ denote the scalars, vectors and matrices, respectively. Symbol $(\cdot)^H$ indicates the conjugate transpose operator. $\text{rk}(\mathbf{A})$, $\text{tr}(\mathbf{A})$ and $|\mathbf{A}|$ mean the rank, trace and determinant of matrix $\mathbf{A}$, respectively. And, constant matrix $\mathbf{I}$ indicates the identity matrix.

\section{Toeplitz Covariance Matrix Model for Target Detections}\label{sec:form}
%In radar target detection issues, radar transmits a coherent pulse train with length $m$ over a coherent processing interval (CPI), and the observed data can be expressed as a series of $m$-dimension vectors $\bm{x}=(x_1,x_2,\dots,x_m)$ which represent the corresponding range cells, respectively. 
In detection issues, the observed data consist of a series of data which can be classified into two sorts. The first one is primary data that may contains the target echo, and the second one is secondary data which provides the reference of background clutters. By the primary data and the secondary data, the detection problem can be formulated as following binary hypothesis test model,
\begin{equation}
	\begin{split}
		&\mathcal{H}_0:
		\begin{cases}
			\bm{x}=\bm{w}\\
			\bm{x}_k=\bm{w}_k\quad k=1,\dots,K,\\
		\end{cases}\\
		&\mathcal{H}_1:
		\begin{cases}
			\bm{x}=\bm{s}+\bm{w}\\
			\bm{x}_k=\bm{w}_k\quad k=1,\dots,K,\\
		\end{cases}
	\end{split}
\end{equation}
where $\bm{x}$ is the primary data, $\bm{s}$ is the target echo, $\bm{w}$ is the noise (clutter), $\bm{x}_k$ is the $k$-th secondary data and $K$ is the number of secondary data. Suppose the noise $\bm{w}$ and $\bm{w}_k\;(k=1,\dots,K)$ are the independently and identically distributed stochastic variable, thus the power and other statistical characteristic of $\bm{w}$ can be estimated by $\bm{w}_k\;(k=1,\dots,K)$.\par
To deal with the detection problem, the covariance matrix as the second order statistics is often employed to represent the statistical characteristics of observed data, which is shown as the following Toeplitz structure
\begin{equation}
	\mathbf{C}_{\bm{x}}=
	\begin{bmatrix}
		c_0 & c_1 & \dots & c_{m-1}\\
		c_{-1} & c_0 &\dots &c_{m-2}\\
		\vdots & \vdots & \ddots & \vdots \\
		c_{1-m} & c_{2-m} &\dots &c_0\\
	\end{bmatrix},
	\label{eq:covariance}
\end{equation}
where
\begin{equation}
	c_k=c_{-k}^*=\sum_{i=k+1}^m x_ix_{i-k}^*\quad(k=0,\dots,m-1).
	\label{eq:correlation}
\end{equation}\par
Suppose that the noise $\bm{w}$ and the target echo $\bm{s}$ are independent, then the correlation coefficient, under hypothesis $\mathcal{H}_1$, can be decomposed as
\begin{equation}
	c_k|_{\mathcal{H}_1}=\sum_{i=k+1}^m x_ix_{i-k}^*=\sum_{i=k+1}^m  s_is_{i-k}^*+\sum_{i=k+1}^m w_iw_{i-k}^*.
\end{equation}
%Especially, the diagonal elements indicate the power of noise plus the radar echo signal, i.e., 
%\begin{equation}
%	c_0|_{\mathcal{H}_1}=\parallel\bm{s}\parallel^2+\parallel\bm{w}\parallel^2.
%	\label{eq:power_c0}
%\end{equation}
Therefore, the covariance matrix $\mathbf{C}_{\bm{x}}$ under $\mathcal{H}_1$ can be decomposed to
\begin{equation}
	\mathbf{C}_{\bm{x}}|_{\mathcal{H}_1}=\mathbf{C}_{\bm{s}}+\mathbf{C}_{\bm{w}}.
\end{equation}
\par
By the covariance matrix $\mathbf{C}_{\bm{x}}$, the binary hypothesis test model can be reformulated as
\begin{equation}
	\begin{split}
		&\mathcal{H}_0:
		\begin{cases}
			\mathbf{C}_{\bm{x}}=\mathbf{C}_{\bm{w}}\\
			\mathbf{C}_{\bm{x}_k}=\mathbf{C}_{\bm{w}_k}\quad k=1,\dots,K,\\
		\end{cases}\\
		&\mathcal{H}_1:
		\begin{cases}
			\mathbf{C}_{\bm{x}}=\mathbf{C}_{\bm{s}}+\mathbf{C}_{\bm{w}}\\
			\mathbf{C}_{\bm{x}_k}=\mathbf{C}_{\bm{w}_k}\quad k=1,\dots,K.\\
		\end{cases}
	\end{split}
\end{equation}
The covariance matrix is used to represent the statistical characteristic of the observed data, but the operation performed to the matrix is expensive because of the high dimension.
%\begin{remark}
%	Because the $\bm{w}$ and $\bm{w}_k\;(k=1,\dots,K)$ are independently and identically distributed, the matrix $\mathbf{C}_{\bm{x}}|_{\mathcal{H}_0}=\mathbf{C}_{\bm{w}}$ is closer to the $\mathbf{C}_{\bm{w}_k}\; (k=1,\dots,K)$ than $\mathbf{C}_{\bm{x}}|_{\mathcal{H}_1}=\mathbf{C}_{\bm{w}}+\mathbf{C}_{\bm{s}}$. 
%\end{remark}
\section{Toeplitz HPD Manifold and Dual Power Spectrum Manifold}\label{sec:THPD_SP}
This section would introduce the Toeplitz HPD manifold and its dual manifold.
Briefly, the main contents of this section are presented as follows. Firstly, the covariance matrix $\mathbf{C}_{\bm{x}}$ in (\ref{eq:covariance}) is shown to be Toeplitz HPD matrix, and the geometry of the Toeplitz HPD manifold $\mathcal{M}_{\mathcal{T}H_{++}}$ is studied. Then, the dual manifold of $\mathcal{M}_{\mathcal{T}H_{++}}$, named dual power spectrum manifold, is builded, and two ways of connecting these two manifold are deduced. Finally, we show that there exists an induced potential function on dual power spectrum manifold for every affine invariant geometric measure, which is equivalent to the corresponding geometric measure.
\subsection{Toeplitz HPD Manifold}\label{sec:THPD}
For the covariance matrix $\mathbf{C}_{\bm{x}}$, the following theorem state that $\mathbf{C}_{\bm{x}}$ is a Toeplitz HPD matrix.
%Because the covariance matrix $\mathbf{C}_{\bm{x}}$ is calculated using the signal vector $\bm{x}$ by (\ref{eq:covariance}) and (\ref{eq:correlation}), it is endowed some special properties.
\begin{theorem}
	Given $\bm{x}=(x_1,x_2,\dots,x_m)\in\mathbb{C}^m\backslash\{\bm{0}\}$, the corresponding matrix $\mathbf{C}_{\bm{x}}$ is a Toeplitz Hermitian positive definite matrix.
	\label{the:THPD}
\end{theorem}
\begin{proof}
	According to the (\ref{eq:covariance}) and (\ref{eq:correlation}), the matrix $\mathbf{C}_{\bm{x}}$ is obviously Hermitian and Toeplitz. Then, the positive definiteness is shown below.\par
	For any $\bm{y}=(y_1,y_2,\dots,y_m)\in\mathbb{C}^m\backslash\{\bm{0}\}$, the quadratic form is
	\begin{equation}
		\bm{y}^H\mathbf{C}_{\bm{x}}\bm{y}=\sum^m_{i=1}\sum^m_{j=1} c_{i-j}y_iy_j^*.
	\end{equation}
	By (\ref{eq:correlation}), it can be expressed as
	\begin{equation}
		\begin{split}
			\bm{y}^H\mathbf{C}_{\bm{x}}\bm{y}=&\sum^m_{i=1}\sum^m_{j=1}\sum_{1\le i+k\le m\atop 1\le j+k\le m} x_{i+k}x_{j+k}^*y_iy_j^*\\
			=&\sum^{m-1}_{k=1-m}\sum_{1\le i+k\le m\atop1\le i\le m}\sum_{1\le j+k\le m\atop1\le j\le m} x_{i+k}x_{j+k}^*y_iy_j^*\\
			=&\sum^{m-1}_{k=1-m}\left(\sum_{1\le i+k\le m\atop1\le i\le m}x_{i+k}y_i\right)\left(\sum_{1\le j+k\le m\atop1\le j\le m} x_{j+k}y_j\right)^*\\
			=&\sum^{m-1}_{k=1-m}\left\Arrowvert\sum_{1\le i+k\le m\atop1\le i\le m}x_{i+k}y_i\right\Arrowvert^2\ge 0.\\
		\end{split}
	\end{equation}
	And, the quadratic form is zero, if and only if
	\begin{equation}
		\sum_{1\le i+k\le m\atop1\le i\le m}x_{i+k}y_i=0,\quad1-m\le k\le m-1,
	\end{equation}
	i.e.,
	\begin{equation}
		\tilde{\mathbf{X}}\bm{y}=
		\begin{bmatrix}
			& & & x_1\\
			& & x_1 & x_2\\
			& \iddots & \vdots & \vdots\\
			x_1 & \dots & x_{m-1} & x_m\\
			x_2 & \dots & x_m & \\
			\vdots & \iddots & & \\
			x_m & & & \\
		\end{bmatrix}
		\begin{bmatrix}
			y_1\\
			y_2\\
			\vdots\\
			y_m
		\end{bmatrix}=\bm{0}.
	\end{equation}
	Actually, as $\bm{y}\ne \bm{0}$, $\tilde{\mathbf{X}}\bm{y}=\bm{0}$ means $\text{rk}(\tilde{\mathbf{X}})<m$. However, the $m$ columns of $\tilde{\mathbf{X}}$ are linearly independent as $\bm{x}\ne\bm{0}$, i.e., $\text{rk}(\tilde{\mathbf{X}})=m$. 
	Therefore, for any $\bm{y}\in\mathbb{C}^m\backslash\{\bm{0}\}$ the quadratic form $\bm{y}^H\mathbf{C}_{\bm{x}}\bm{y}>0$, i.e., the matrix $\mathbf{C}_{\bm{x}}$ is positive definite.\par
	Totally, $\mathbf{C}_{\bm{x}}$ is a Toeplitz Hermitian positive definite matrix.
\end{proof}
Therefore, the matrices $\mathbf{C}_{\bm{x}}$ and $\mathbf{C}_{\bm{w}_k}\;(k=1,\dots,K)$ are located on the Toeplitz HPD manifold, which is defined as follows.
\begin{definition}[Toeplitz HPD manifold]
The Toeplitz HPD manifold is
\begin{equation}
	\mathcal{M}_{\mathcal{T}H_{++}}=\{\mathbf{C}|\mathbf{C}\in\mathcal{T}(m,\mathbb{C}),\mathbf{C}^H=\mathbf{C},\mathbf{C}\succ \mathbf{0}\},
\end{equation}
where $\mathcal{T}(m,\mathbb{C})$ is the set of $m\times m$ dimensional complex Toeplitz matrices.
\label{def:THPD}
\end{definition}
On the Toeplitz HPD matrix manifold, the geometric measure, which is used for quantifying the difference between two matrices $\mathbf{C}_1,\mathbf{C}_2$, is the function $\mathcal{D}:\;\mathcal{M}_{\mathcal{T}H_{++}}\times \mathcal{M}_{\mathcal{T}H_{++}}\to \mathbb{R}$ satisfying\cite{Amari2016}:
\begin{itemize}
	\item $\forall \mathbf{C}_1,\mathbf{C}_2\in\mathcal{M}_{\mathcal{T}H_{++}},\;\mathcal{D}(\mathbf{C}_1,\mathbf{C}_2)\ge 0$;
	\item $\mathcal{D}(\mathbf{C}_1,\mathbf{C}_2)=0$ if and only if $\mathbf{C}_1=\mathbf{C}_2$.
\end{itemize}
The geometric measure $\mathcal{D}$ can be regarded as the extended distance which is not limited by the symmetry and triangular inequality.\par
There are some typical geometric measures on the matrix manifold:
\begin{itemize}
	\item Riemannian Distance:
	\begin{equation}
		\mathcal{D}_{\text{RD}}(\mathbf{C}_1,\mathbf{C}_2)=\left\Arrowvert \log\left(\mathbf{C}_1^{-\frac{1}{2}}\mathbf{C}_2\mathbf{C}_1^{-\frac{1}{2}}\right)\right\Arrowvert_F^2;
	\end{equation}
	\item Kullback-Leibler Divergence:
	\begin{equation}
		\mathcal{D}_{\text{KL}}(\mathbf{C}_1,\mathbf{C}_2)=\text{tr}(\mathbf{C}_1\mathbf{C}_2^{-1}-\mathbf{I})-\log\left|\mathbf{C}_1\mathbf{C}_2^{-1}\right|;
	\end{equation}
	\item Jensen-Shannon Divergence:
	\begin{equation}
		\begin{split}
			&\mathcal{D}_{\text{JS}}(\mathbf{C}_1,\mathbf{C}_2)\\
			=&\frac{1}{2}\left[\mathcal{D}_{\text{KL}}\left(\mathbf{C}_1,\frac{\mathbf{C}_1+\mathbf{C}_2}{2}\right)+\mathcal{D}_{\text{KL}}\left(\mathbf{C}_2,\frac{\mathbf{C}_1+\mathbf{C}_2}{2}\right)\right];
		\end{split}
	\end{equation}
	\item log-determinant Divergence:
	\begin{equation}
		\mathcal{D}_{\text{LD}}(\mathbf{C}_1,\mathbf{C}_2)=\log\left|\frac{\mathbf{C}_1+\mathbf{C}_2}{2}\right|-\log\sqrt{\left|\mathbf{C}_1\mathbf{C}_2\right|}.
	\end{equation}
\end{itemize}
In fact, these geometric measures are affine invariant measures\cite{Pennec2006}, which remain invariant under affine transformations.
\begin{definition}[Affine Invariant Measure]
	The geometric measure $\mathcal{D}(\mathbf{C}_1,\mathbf{C}_2)$ is called affine invariant measure, if and only if $\forall \mathbf{W}\in \text{GL}(m,\mathbb{C})$ (general linear group) the following equation is established,
	\begin{equation}
		\mathcal{D}(\mathbf{W}^H\mathbf{C}_1\mathbf{W},\mathbf{W}^H\mathbf{C}_2\mathbf{W})=\mathcal{D}(\mathbf{C}_1,\mathbf{C}_2).
	\end{equation}
	\label{def:gm}
\end{definition}
\begin{lemma}
	The geometric measure $\mathcal{D}_{\text{RD}},\mathcal{D}_{\text{KL}},\mathcal{D}_{\text{JS}},\mathcal{D}_{\text{LD}}$ are affine invariant measures.
	\label{lem:afigm}
\end{lemma}
\begin{proof}
	See Appendix~\ref{app:aff_invar_mea}.
\end{proof}
\begin{figure*}[htp]
	\centering
	\includegraphics[width=0.75\textwidth]{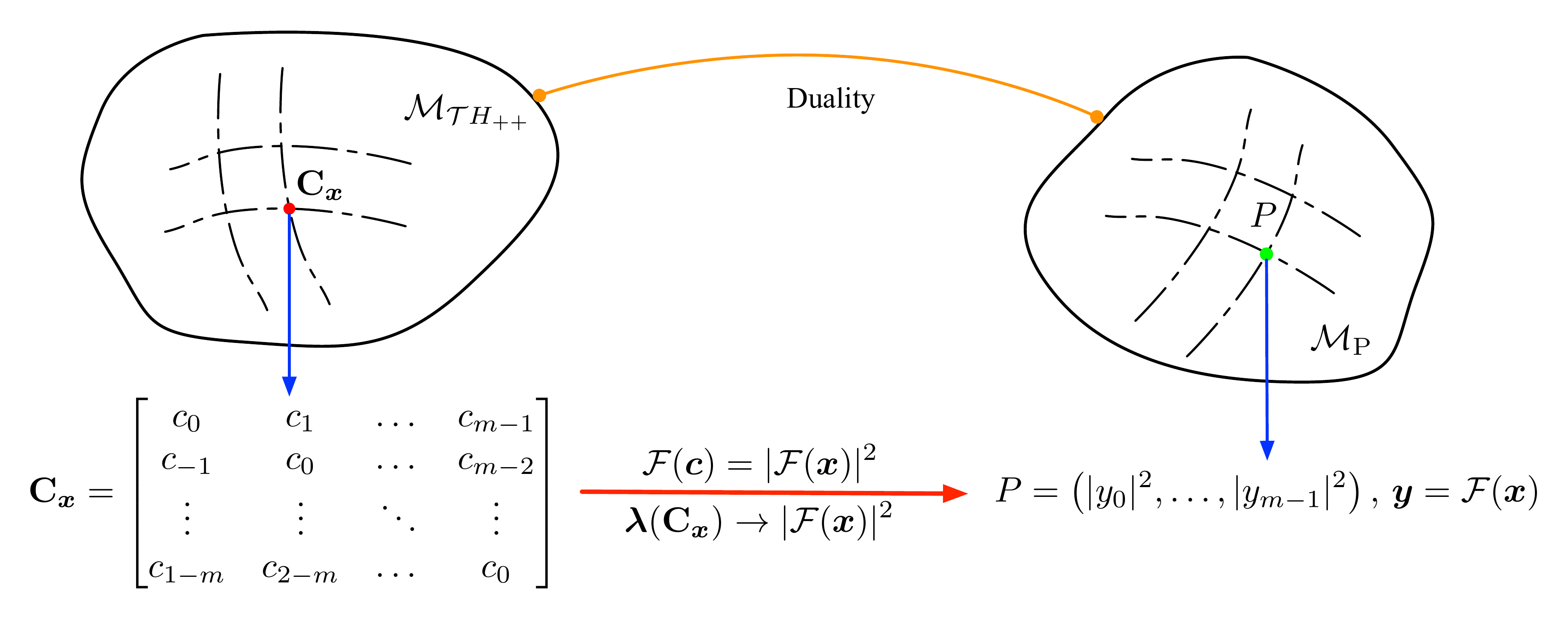}
	\caption{The relation between Toeplitz HPD matrix manifold and dual power spectrum manifold. There are two ways which can connect the two manifolds: the Fourier transforms of the elements of $\mathbf{C}_{\bm{x}}$ equals to the power spectrum of $\bm{x}$, and the eigenvalues of $\mathbf{C}_{\bm{x}}$ tend to the power spectrum of $\bm{x}$.}
	\label{fig:duality}
\end{figure*}
\subsection{Dual Power Spectrum Manifold}
By the Wiener-Khinchin theorem, the power spectrum is the Fourier transform of the autocorrelation. Moreover, the manifold composed by the power spectrum is the dual manifold of $\mathcal{M}_{\mathcal{T}H_{++}}$. Specifically, the dual power spectrum manifold $\mathcal{M}_{\text{P}}$ is defined as follows.
\begin{definition}[Dual Power Spectrum Manifold]
The dual power spectrum manifold is
	\begin{equation}
		\begin{split}
			\mathcal{M}_{\text{P}}=\left\{P=\left(|y_0|^2,\dots,|y_{m-1}|^2\right)\,\Big|\;\bm{y}=\mathcal{F}(\bm{x}),\bm{x}\in\mathbb{C}^m\backslash\bm{0}\right\},
		\end{split}
	\end{equation}
	where $\mathcal{F}(\bm{x})=(y_0,\dots,y_{m-1})$ that
	\begin{equation}
		y_k=\sum^m_{i=1}x_ie^{-j\frac{2\pi ki}{m}}.
	\end{equation}
	\label{def:dual}
\end{definition}
In fact, the points on Toeplitz HPD manifold and dual power spectrum manifold are both generated by the received signals $\bm{x}$, and they can be connected by the following lemma. 
\begin{lemma}
	Given $\bm{x}=(x_1,x_2,\dots,x_m)\in\mathbb{C}^m\backslash\bm{0}$, the power spectrum $P=\left(|y_0|^2,|y_1|^2,\dots,|y_{m-1}|^2\right)$ satisfies
	\begin{equation}
		|y_k|^2=\sum^{m-1}_{i=1-m}c_ie^{-j\frac{2\pi ki}{m}},
	\end{equation}
	\label{lem:c2p}
	where $\bm{c}=(c_0,c_1,\dots,c_{m-1})$ is the first row of $\mathbf{C}_{\bm{x}}$.
\end{lemma}
\begin{proof}
	\begin{equation}
		\begin{split}
			|y_k|^2=&\sum^m_{i=1}x_ie^{-j\frac{2\pi ki}{m}}\left(\sum^m_{l=1}x_le^{-j\frac{2\pi kl}{m}}\right)^*\\
			=&\sum^m_{i=1}\sum^m_{l=1}x_ix_l^*e^{-j\frac{2\pi k(i-l)}{m}}\\
			=&\sum^{m-1}_{i=1-m}\sum_{1\le t\le m\atop1\le t-i\le m}x_tx_{t-i}^*e^{-j\frac{2\pi ki}{m}}\\
			=&\sum^{m-1}_{i=1-m}c_ie^{-j\frac{2\pi ki}{m}}\\
		\end{split}
		\label{eq:lm1_1}	
	\end{equation}
\end{proof}
Actually, according to the properties of autocorrelation sequences, the $c_k$ tends to zero, i.e., $\lim_{k\to\infty} c_k=0$. Then, the eigenvalue of $\mathbf{C}_{\bm{x}}$ satisfies the following theorem.
\begin{lemma}
	While $m\to\infty$, if $\sum_k c_k$ is absolute convergence, the eigenvalue $\lambda_k$ of $\mathbf{C}_{\bm{x}}$ satisfies
	\begin{equation}
		\lambda_k=|y_k|^2=\sum^{m-1}_{i=1-m}c_ie^{-j\frac{2\pi ki}{m}}\quad(k=0,\dots,m-1),
	\end{equation}
	i.e., $\Lambda=(\lambda_0,\lambda_1,\dots,\lambda_{m-1})=P=\left(|y_0|^2,|y_1|^2,\dots,|y_{m-1}|^2\right)$.
	\label{lem:lam2p}
\end{lemma}
\begin{proof}
	Let the circulant matrix 
	\begin{equation}
		\mathbf{C}^*=
		\begin{bmatrix}
			c_0 &  \cdots & c_l & c_{1-m+l} & \cdots & c_{-1}\\
			c_{-1} & \cdots & c_{l-1} & c_l & \cdots & c_{-2}\\
			\vdots  &  & \vdots & \vdots &  & \vdots\\
			c_1 &  \cdots & c_{1-m+l} & c_{2-m+l} & \cdots  & c_{0}
		\end{bmatrix},
	\end{equation}
	where $l=[m/2]$ is the nearest integer of $m/2$.\par
	As $\lim_{k\to\infty} c_k=0$ and the properties of banded Toeplitz matrix and circulant matrix\cite{gray2005toeplitz}, while $m\to\infty$, $\mathbf{C}_{\bm{x}}\to \mathbf{C}^*$ and the eigenvalues $\lambda_k\;(k=0,\dots,m-1)$ of $\mathbf{C}_{\bm{x}}$ satisfies
	\begin{equation}
		\lambda_k=\sum^{l}_{i=1-m+l}c_ie^{-j\frac{2\pi ki}{m}}.
	\end{equation}
	Because $\sum_k c_k$ is absolute convergence, we can get
	\begin{equation}
		\begin{split}
			&\lim_{m\to\infty}\left|\sum^{m-1}_{i=l+1}c_ie^{-j\frac{2\pi ki}{m}}\right|\le\lim_{m\to\infty}\sum^{m-1}_{i=l+1}|c_i|=0,\\
			&\lim_{m\to\infty}\left|\sum^{-m+l}_{i=1-m}c_ie^{-j\frac{2\pi ki}{m}}\right|\le\lim_{m\to\infty}\sum^{-m+l}_{i=1-m}|c_i|=0.
		\end{split}
	\end{equation}
	Then, the following equation is established under $m\to\infty$,
	\begin{equation}
		\lambda_k=\sum^{l}_{i=1-m+l}c_ie^{-j\frac{2\pi ki}{m}}=\sum^{m-1}_{i=1-m}c_ie^{-j\frac{2\pi ki}{m}}.
	\end{equation}
	And, as Lemma~\ref{lem:c2p}, we can get $\lambda_k=|y_k|^2$ and $\Lambda=(\lambda_0,\lambda_1,\dots,\lambda_{m-1})=P=\left(|y_0|^2,|y_1|^2,\dots,|y_{m-1}|^2\right)$.
\end{proof}
\begin{remark}
	Lemma~\ref{lem:lam2p} reveals that the eigenvalues of $\mathbf{C}_{\bm{x}}$, $\lambda_0,\dots,\lambda_{m-1}$, asymptotically equal to the power spectrum of $\bm{x}$. That means the point $\mathbf{C}_{\bm{x}}$ on $\mathcal{M}_{\mathcal{T}H_{++}}$ can be related to the $(\lambda_0,\dots,\lambda_{m-1})$ on $\mathcal{M}_{\text{P}}$. Moreover, as the properties of banded Toeplitz matrix and circulant matrix, while $m\to \infty$\cite{gray2005toeplitz},
	\begin{equation}
		\mathbf{C}_{\bm{x}}\to\mathbf{F}\,\text{diag}(\lambda_0,\lambda_1,\dots,\lambda_{m-1})\,\mathbf{F}^H,
	\end{equation}
	\begin{equation}
		\mathbf{C}_{\bm{x}}\to\mathbf{F}\,\text{diag}\left(|y_0|^2,|y_1|^2,\dots,|y_{m-1}|^2\right)\,\mathbf{F}^H,
	\end{equation}
	where $\mathbf{F}$ is a constant matrix that does not depend on $\mathbf{C}_{\bm{x}}$.
\end{remark}
 Totally, the duality between $\mathcal{M}_{\mathcal{T}H_{++}}$ and $\mathcal{M}_{\text{P}}$ are illustrated as Fig.\ref{fig:duality}. There are two crucial ways which can connect the two manifolds, and the asymptotical relationship ($\bm{\lambda}(\mathbf{C}_{\bm{x}})\to|\mathcal{F}(\bm{x})|^2$) will play a more important role in the next contents. 
\subsection{Induced Potential Function on Dual Power Spectrum Manifold}
%The dual power spectrum manifold $\mathcal{M}_{\text{P}}$ and Toeplitz HPD manifold $\mathcal{M}_{\mathcal{T}H_{++}}$ are both generated by the received signal, and the matrix $\mathbf{C}_{\bm{x}}$ can be related to the power spectrum $\{\lambda_0,\dots,\lambda_{m-1}\}$ on $\mathcal{M}_{\text{P}}$. Moreover, the difference between two Toeplitz HPD matrices also can be modeled by a power spectrum on $\mathcal{M}_{\text{P}}$.\par

In signal processing theory, the whitening processing is an effective method to eliminate the statistical characteristics of the noise in the whitened signal. Consider the matrices $\mathbf{C}_1,\mathbf{C}_2$ are generated by the signal $\bm{s}+\bm{w}$ and $\bm{w}$, respectively. The result of whitening processing can be regarded as $\mathbf{C}_1\mathbf{C}_2^{-1}$. This is because 
\begin{equation}
	\begin{split}
		&\mathbf{C}_1\mathbf{C}_2^{-1}\to \mathbf{F}\,\text{diag}(\lambda_0/\lambda'_0,\dots,\lambda_{m-1}/\lambda'_{m-1})\,\mathbf{F}^H\,(m\to\infty)\\
		&\begin{cases}
			\lambda_k=|y_k|^2,\quad\bm{y}=\mathcal{F}(\bm{s}+\bm{w}),\\
			\lambda'_k=|y'_k|^2,\quad\bm{y}'=\mathcal{F}(\bm{w}).\\
		\end{cases}
	\end{split}
\end{equation}
\par
We formulate the result of whitening processing to represent the difference between $\mathbf{C}_1,\mathbf{C}_2$, then the following mapping from $\mathcal{M}_{\mathcal{T}H_{++}}$ to $\mathcal{M}_{\text{P}}$ can be defined.
%the mapping from the product of the Toeplitz HPD manifold $\mathcal{M}_{\mathcal{T}H_{++}}$ to dual power spectrum manifold $\mathcal{M}_{\text{P}}$ can be defined.
\begin{definition}[Whitening Spectrum Mapping]
	Define the whitening spectrum mapping $\mathcal{P}:\mathcal{M}_{\mathcal{T}H_{++}}\times\mathcal{M}_{\mathcal{T}H_{++}}\to \mathcal{M}_{\text{P}}$, that
	\begin{equation}
		\mathcal{P}(\mathbf{C}_1,\mathbf{C}_2)=\Lambda_{\mathbf{C}_1\mathbf{C}_2^{-1}}=(\lambda_0,\lambda_1,\dots,\lambda_{m-1}),
	\end{equation}
	where $\Lambda_{\mathbf{C}}=(\lambda_0,\lambda_1,\dots,\lambda_{m-1})$ are the eigenvalues of $\mathbf{C}$.
	\label{def:map}
\end{definition}
\begin{remark}
	As a well-known property, the eigenvalues of $\mathbf{A}\mathbf{B}$ equal to the eigenvalues of $\mathbf{B}\mathbf{A}$ when $\mathbf{A},\mathbf{B}\in \mathbb{C}^{m\times m}$. Therefore, the eigenvalues of $\mathbf{C}_1\mathbf{C}_2^{-1}=(\mathbf{C}_1\mathbf{C}_2^{-\frac{1}{2}})\mathbf{C}_2^{-\frac{1}{2}}$ equals to that of $\mathbf{C}_2^{-\frac{1}{2}}\mathbf{C}_1\mathbf{C}_2^{-\frac{1}{2}}$. That means $\mathcal{P}(\mathbf{C}_1,\mathbf{C}_2)=\Lambda_{\mathbf{C}'}$ is also established,
%	\begin{equation}
%		\mathcal{P}(\mathbf{C}_1,\mathbf{C}_2)=\Lambda_{\mathbf{C}'}
%	\end{equation}
	where $\mathbf{C}'=\mathbf{C}_2^{-\frac{1}{2}}\mathbf{C}_1\mathbf{C}_2^{-\frac{1}{2}}$.
\end{remark}
Through such mapping, the difference between two matrices can also be quantified by a function on $\mathcal{M}_{\text{P}}$.
\begin{theorem}
	For each affine invariant measure $\mathcal{D}(\mathbf{C}_1,\mathbf{C}_2)$, there exists a function $\varphi_{\mathcal{D}}\,:\,\mathcal{M}_{\text{P}}\to\mathbb{R}$ satisfying 
	\begin{equation}
		\varphi_{\mathcal{D}}(\mathcal{P}(\mathbf{C}_1,\mathbf{C}_2))=\mathcal{D}(\mathbf{C}_1,\mathbf{C}_2),\quad \forall\mathbf{C}_1,\mathbf{C}_2\in\mathcal{M}_{\mathcal{T}H_{++}}.
	\end{equation}
	\label{the:afi2ipf}
\end{theorem}
\begin{proof}
	As affine invariant property, $\forall\mathbf{C}_1,\mathbf{C}_2\in\mathcal{M}_{\mathcal{T}H_{++}}$,
	\begin{equation}
		\mathcal{D}(\mathbf{C}_1,\mathbf{C}_2)=\mathcal{D}(\mathbf{C}_2^{-\frac{1}{2}}\mathbf{C}_1\mathbf{C}_2^{-\frac{1}{2}},\mathbf{I}).
	\end{equation}
	Because $\mathbf{C}_2^{-\frac{1}{2}}\mathbf{C}_1\mathbf{C}_2^{-\frac{1}{2}}$ is a symmetric matrix, it can be orthogonally decomposed to
	\begin{equation}
		\mathbf{C}_2^{-\frac{1}{2}}\mathbf{C}_1\mathbf{C}_2^{-\frac{1}{2}}=\mathbf{Q}^H\,\text{diag}\,(\lambda_0,\dots,\lambda_{m-1})\mathbf{Q},
	\end{equation}
	where $\mathbf{Q}^H\mathbf{Q}=\mathbf{Q}\mathbf{Q}^H=\mathbf{I}$ and $\lambda_0,\dots,\lambda_{m-1}$ are the eigenvalues of $\mathbf{C}_2^{-\frac{1}{2}}\mathbf{C}_1\mathbf{C}_2^{-\frac{1}{2}}$. Therefore,
	\begin{equation}
		\begin{split}
			\mathcal{D}(\mathbf{C}_1,\mathbf{C}_2)=&\mathcal{D}(\mathbf{Q}\mathbf{C}_2^{-\frac{1}{2}}\mathbf{C}_1\mathbf{C}_2^{-\frac{1}{2}}\mathbf{Q}^H,\mathbf{Q}\mathbf{Q}^H)\\
			=&\mathcal{D}(\text{diag}\,(\lambda_0,\dots,\lambda_{m-1}),\mathbf{I}).
		\end{split}
	\end{equation}
	That means $\mathcal{D}(\mathbf{C}_1,\mathbf{C}_2)$ only depending on the eigenvalues of $\mathbf{C}_2^{-\frac{1}{2}}\mathbf{C}_1\mathbf{C}_2^{-\frac{1}{2}}$, i.e., $\Lambda_{\mathbf{C}_1\mathbf{C}_2^{-1}}$.\par
	Then, let $\varphi_{\mathcal{D}}(\lambda_0,\dots,\lambda_{m-1})=\mathcal{D}(\text{diag}\,(\lambda_0,\dots,\lambda_{m-1}),\mathbf{I})$, we can get $\varphi_{\mathcal{D}}(\mathcal{P}(\mathbf{C}_1,\mathbf{C}_2))=\mathcal{D}(\mathbf{C}_1,\mathbf{C}_2)$.
\end{proof}
\begin{remark}
	The function $\varphi_{\mathcal{D}}$ can be constructed by $\varphi_{\mathcal{D}}(\mathcal{P}(\mathbf{C}_1,\mathbf{C}_2))=\mathcal{D}(\text{diag}\,(\lambda_0,\dots,\lambda_{m-1}),\mathbf{I})\;(\mathcal{P}(\mathbf{C}_1,\mathbf{C}_2)=\Lambda_{\mathbf{C}_1\mathbf{C}_2^{-1}}=(\lambda_0,\dots,\lambda_{m-1}))$.
\end{remark}
\begin{figure}[htp]
	\centering
	\includegraphics[width=0.37\textwidth]{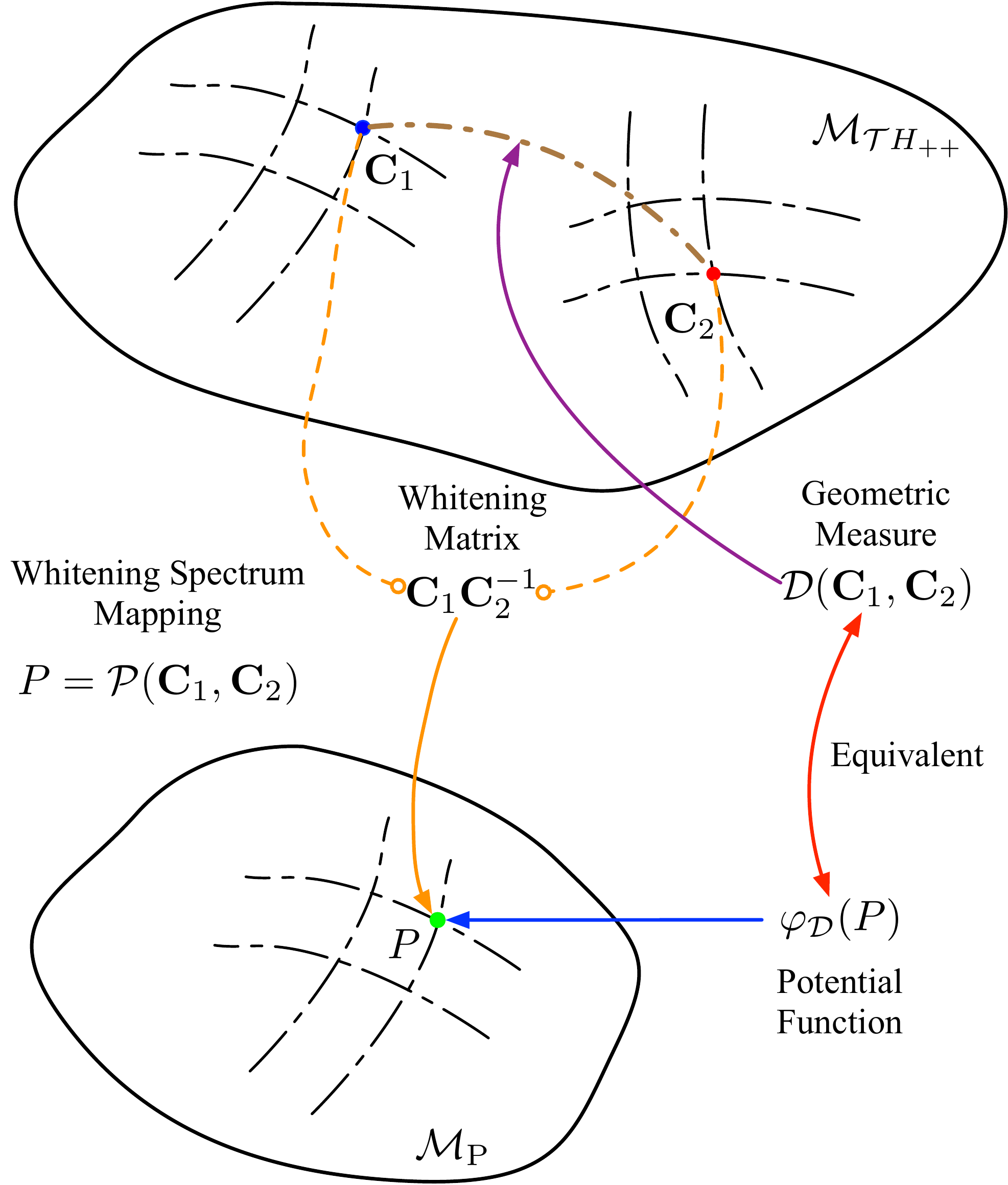}
	\caption{The relation between the geometric measure and induced potential function.}
	\label{fig:gm2pf}
\end{figure}
\begin{definition}[Induced Potential Function]
	For the affine invariant measure $\mathcal{D}$, the function $\varphi_{\mathcal{D}}$ is called the $\mathcal{D}$ induced potential function.
	\label{def:pf}
\end{definition}
Totally, the relation between the geometric measure and induced potential function is shown in Fig.\ref{fig:gm2pf}.\par
In addition, the induced potential functions of the typical geometric measures are as follows:
\begin{itemize}
	\item Riemannian Distance:
	\begin{equation}
		\varphi_{\mathcal{D}_{\text{RD}}}(\lambda_0,\dots,\lambda_{m-1})=\sum^{m-1}_{k=0} \log^2\lambda_k;
	\end{equation}
	\item Kullback-Leibler Divergence:
	\begin{equation}
		\varphi_{\mathcal{D}_{\text{KL}}}(\lambda_0,\dots,\lambda_{m-1})=\sum^{m-1}_{k=0} \lambda_k-1-\log\lambda_k;
	\end{equation}
	\item Jensen-Shannon Divergence:
	\begin{equation}
		\begin{split}
			\varphi_{\mathcal{D}_{\text{JS}}}(\lambda_0,\dots,\lambda_{m-1})=\sum^{m-1}_{k=0} \log\frac{\lambda_k+1}{2\sqrt{\lambda_k}};
		\end{split}
	\end{equation}
	\item log-determinant Divergence:
	\begin{equation}
		\varphi_{\mathcal{D}_{\text{LD}}}(\lambda_0,\dots,\lambda_{m-1})=\sum^{m-1}_{k=0} \log\frac{\lambda_k+1}{2\sqrt{\lambda_k}}.
	\end{equation}
\end{itemize}
The derivation of the induced potential functions is presented in Appendix~\ref{app:potential_function}. According to the induced potential function and Theorem~\ref{the:afi2ipf}, JSD equals to LDD for each pair of $\mathbf{C}_1,\mathbf{C}_2\in\mathcal{M}_{\mathcal{T}H_{++}}$. Therefore, only the LDD is considered in the next discussion.\par
In summary, the induced potential function is equivalent to the corresponding geometric measure on the Toeplitz HPD manifold. Moreover, compared with the RD, KLD and LDD, the corresponding induced potential functions have simpler forms, that means the operations and analysis of the induced potential functions are easier to implement than geometric measures.
%In summary, the induced potential function is $m$-dimensional vector, which is lower than the a pair of $m\times m$-dimensional matrices.
%\subsection{Summary}
\section{Enhanced Matrix CFAR Detection Based on Dual Power Spectrum Manifold}\label{sec:enhancement}
%Matrix manifold enhancement is an effective way to improve the detection performance by mapping the covariance matrix to a more discriminative and lower-dimensional manifold\cite{Yang2020}. The enhanced mapping is always obtained based on an optimization problem, which is often solved by the gradient descend with expensive computation. Fortunately, based on the relationship between Toeplitz HPD matrix manifold and dual power spectrum manifold, the optimization problem on $\mathcal{M}_{\mathcal{T}H_{++}}$ can be transformed to an equivalent optimization problem on $\mathcal{M}_{\text{P}}$. 
This section formulate the enhancement of the matrix CFAR detection as an optimization problem and then transform it to an equivalent optimization problem on the dual power spectrum manifold, which has a simpler form and is easier to solve than the original optimization. 
Finally, some examples of the enhancement, which are solved by the equivalent transformation, are provided. 
%Briefly, the main contents of this section are presented as follows. Firstly, we introduce the matrix CFAR detection methods. Then, the enhancement of the matrix CFAR detection is formulated. Finally, the equivalent optimization on $\mathcal{M}_{\text{P}}$ is showed. 
%In this section, the enhancement of the detection method based on Toeplitz HPD matrix manifold and the equivalent enhancement based on dual power spectrum manifold would be discussed.\par
%The brief contents are presented as follows. In the first part, the enhancement of geometric detector is introduced. In the second part, we transform the enhancement of geometric detector to an equivalent optimization problem on the dual power spectrum manifold (Theorem~\ref{the:cauthy}, Lemma~\ref{lem:eig_vary}, Lemma~\ref{lem:L2Q} and Theorem~\ref{the:op_eq}). Moreover, the construction of enhanced mapping is provided. The last part presents a number of numeric experiments to show the superiority of the enhancement by solving the equivalent optimization problem on the dual power spectrum manifold. In addition, a series of conclusion about the performance of enhancement are also provided.
%The detection performance of geometric measure based detector can be analyzed on the dual power spectrum manifold $\mathcal{M}_{\text{P}}$.
\subsection{Matrix CFAR Detection}
In the geometric detector, the difference between the Toeplitz covariance matrix of the primary data and the secondary data is quantified by the selected geometric measure, then the geometric measure is compared with the threshold to make the decision. Totally, the matrix CFAR detection scheme is depicted in Fig.\ref{fig:CFAR}.\par
\begin{figure}[htp]
	\centering
	\includegraphics[width=0.45\textwidth]{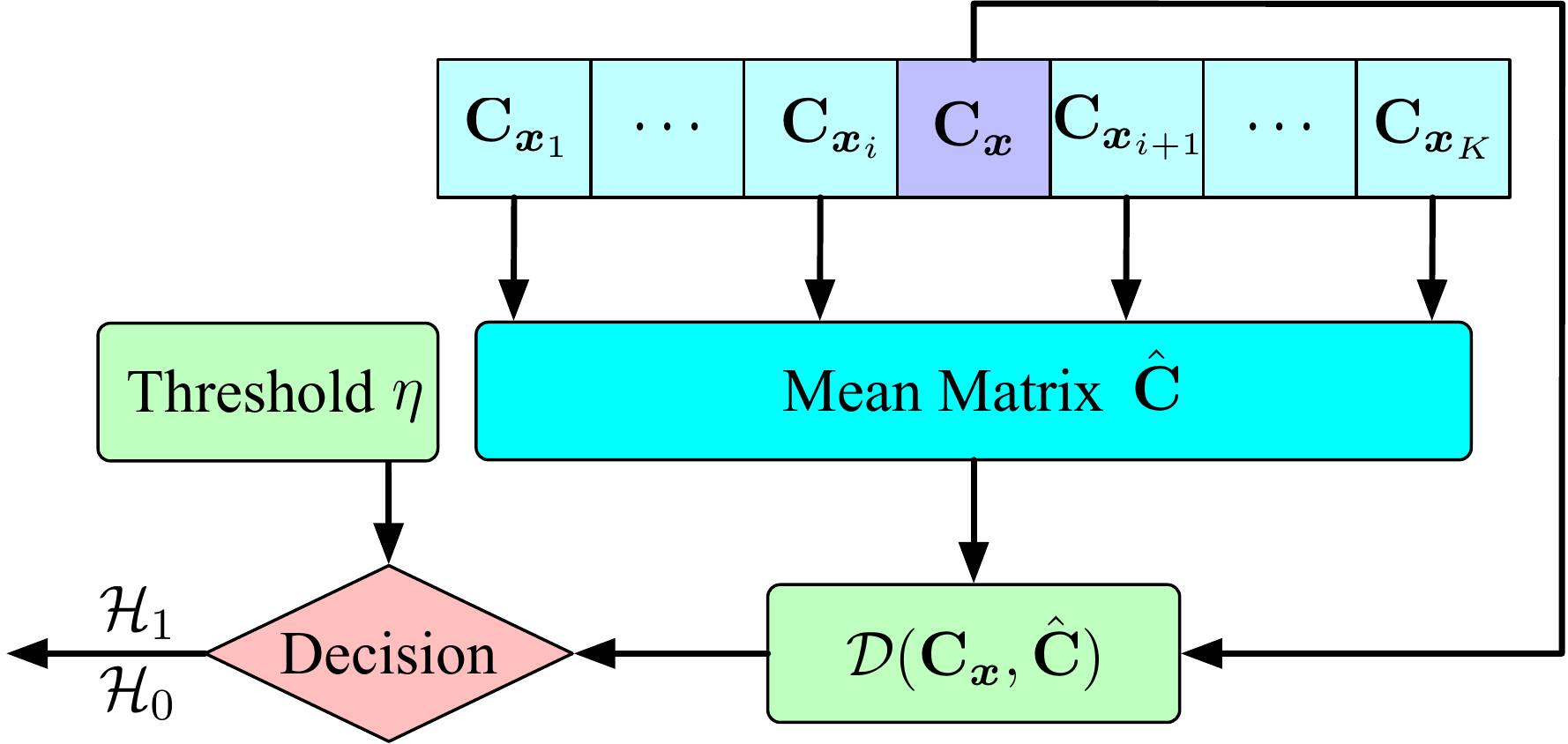}
	\caption{The diagram of the matrix CFAR detection scheme, where $\mathbf{C}_{\bm{x}}$ denotes the covariance matrix of the primary data, and $\mathbf{C}_{\bm{x}_k}$ is the covariance matrix of the secondary data.}
	\label{fig:CFAR}
\end{figure}
The detection is performed by a sliding window, which divides the data into the primary data and secondary data. The detailed steps are summarized as follows:
\begin{enumerate}
	\item The Toeplitz covariance matrices are calculated using the primary data and the secondary data, respectively. 
%	And, the data from CUT is the primary data, others are secondary data.
	\item The mean matrix $\hat{\mathbf{C}}$ of secondary data is estimated using $\mathbf{C}_{\bm{x}_1},\dots,\mathbf{C}_{\bm{x}_K}$.
	\item The selected geometric measure is employed to quantify the difference between $\mathbf{C}_{\bm{x}}$ and $\hat{\mathbf{C}}$.
	\item The decision is made by the comparison between the test statistics and the threshold, i.e.,
	\begin{equation}
		\mathcal{D}(\mathbf{C}_{\bm{x}},\hat{\mathbf{C}})\underset{\mathcal{H}_0}{\overset{\mathcal{H}_1}{\gtrless}}\eta.
		\label{eq:dec_rul}
	\end{equation}
	where $\eta$ is the threshold under false alarm probability $P_f$ 
	and often obtained by $100P_f^{-1}$ independent Monte Carlo runs\cite{zhao2019}.
\end{enumerate}
\begin{table*}[hbt]
  \caption{Solution or iteration formula of estimate $\hat{\mathbf{C}}$ based on different geometric measures.}
  \label{tab:mean}
  \centering
  \begin{tabular}{lc}
  	\toprule
   	Geometric measure & Formula\\
   	\midrule
   	Riemannian Distance & $\hat{\mathbf{C}}_{t+1}=\hat{\mathbf{C}}_t^{\frac{1}{2}}\exp\left(-\epsilon_t\sum\limits^K_{k=1}\log\left(\hat{\mathbf{C}}_t^{-\frac{1}{2}}\mathbf{C}_{\bm{x}_k}\hat{\mathbf{C}}_t^{-\frac{1}{2}}\right)\right)\hat{\mathbf{C}}_t^{\frac{1}{2}}$ \\
    KL Divergence & $\hat{\mathbf{C}}=\left(\frac{1}{K}\sum\limits^K_{k=1}\mathbf{C}^{-1}_{\bm{x}_k}\right)^{-1}$ \\
%    Jensen-Shannon Divergence & $\hat{\mathbf{C}}_{t+1}=\left(\frac{1}{K}\sum\limits^K_{k=1}\left(\frac{\mathbf{C}_{\bm{x}_k}+\hat{\mathbf{C}}_t}{2}\right)^{-1}\right)^{-1}$ \\
    log-determinant Divergence & $\hat{\mathbf{C}}_{t+1}=\left(\frac{1}{K}\sum\limits^K_{k=1}\left(\frac{\mathbf{C}_{\bm{x}_k}+\hat{\mathbf{C}}_t}{2}\right)^{-1}\right)^{-1}$ \\
    \bottomrule
  \end{tabular}
\end{table*}
The estimate $\hat{\mathbf{C}}$ in step 2 is often deduced from the following optimization problem
\begin{equation}
	\hat{\mathbf{C}}=\arg\min_{\mathbf{C}}\sum^K_{k=1}\mathcal{D}(\mathbf{C}_{\bm{x}_k},\mathbf{C})
\end{equation}
the solutions or iteration formulas with different geometric measures are listed in the Table~\ref{tab:mean}.\par
%The closed analytic formula of threshold $\eta$ is often unavailable for the geometric measure based detection methods, and it is often obtained by Monte Carlo runs\cite{zhao2019}. Suppose the preset false alarm probability is $P_F$, approximate $100P_F^{-1}$ samples under $\mathcal{H}_0$ are used to determine the $\eta$.\par
\subsection{Enhancement of Matrix CFAR Detection}\label{sec:enhance_detector}
The dimensionality reduction on HPD matrix is an effective way to make the $\mathcal{H}_0$ and $\mathcal{H}_1$ data more discriminative by the enhanced mapping\cite{Harandi2018}, as shown in Fig.\ref{fig:dim_reduc}.\par
\begin{figure}[htp]
	\centering
	\includegraphics[width=0.49\textwidth]{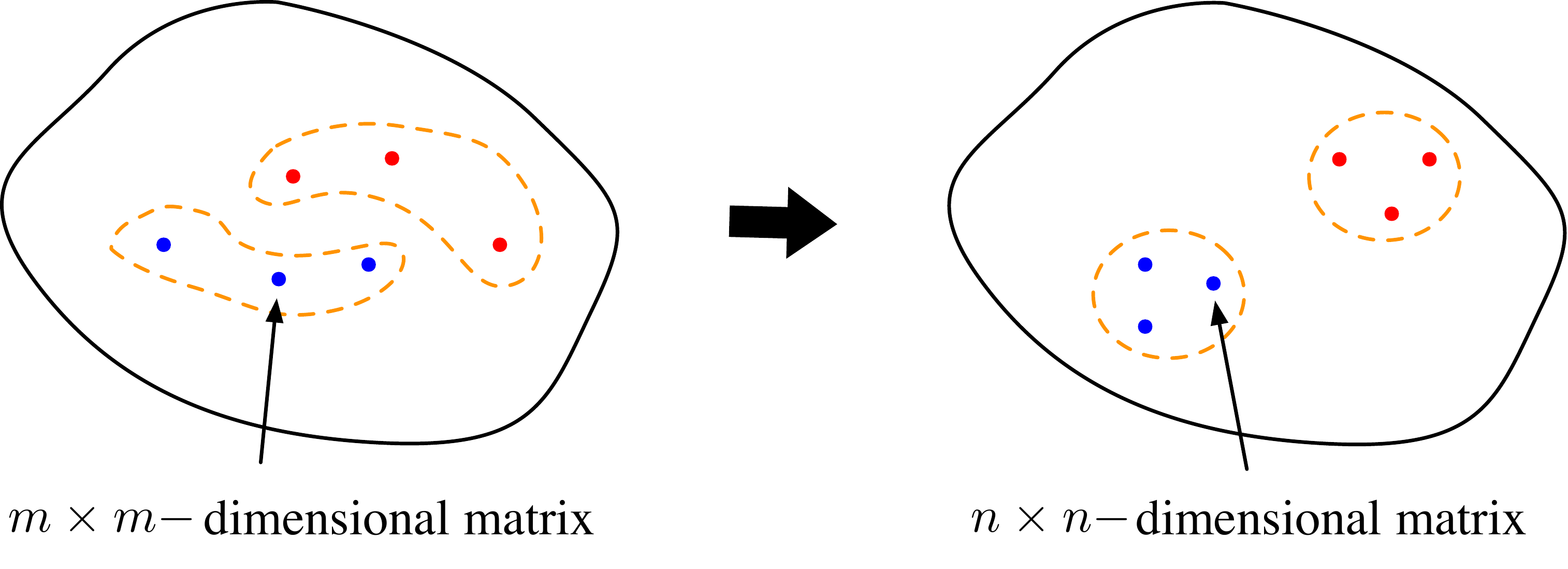}
	\caption{Dimensionality reduction on Toeplitz HPD manifold: The data on the high-dimensional Toeplitz HPD manifold can be mapped to a lower-dimensional manifold, where the $\mathcal{H}_0$ data and $\mathcal{H}_1$ data are more discriminative.}
	\label{fig:dim_reduc}
\end{figure}
Define the enhanced mapping as
\begin{equation}
	\Pi_{\mathbf{W}}(\mathbf{C})=\mathbf{W}^H\mathbf{C}\mathbf{W},
\end{equation}
where $\mathbf{W}\in\mathbb{C}^{m\times n}$ is a full column rank matrix, i.e., $\text{rk}(\mathbf{W})=n$. In this paper, we only consider $n\le m/2$. By the enhanced mapping, the modified geometric measure between $\mathbf{C}_1$ and $\mathbf{C}_2$ is\cite{Yang2020}
\begin{equation}
	\mathcal{D}_m(\Pi_{\mathbf{W}}(\mathbf{C}_1),\Pi_{\mathbf{W}}(\mathbf{C}_2))=\mathcal{D}_n(\mathbf{W}^H\mathbf{C}_1\mathbf{W},\mathbf{W}^H\mathbf{C}_2\mathbf{W}),
\end{equation}
where $\mathcal{D}$ is an affine invariant geometric measure, and $\mathcal{D}_m,\mathcal{D}_n$ is the implementation of $\mathcal{D}$ in $m\times m$ and $n\times n$ dimensional matrix, respectively.\par
Without loss of generality, for any fixed $n\le m/2$, the crucial point of the enhanced geometric measure is the full column rank mapping matrix $\mathbf{W}$, which is determined by the following optimization problem
\begin{equation}
	\begin{split}
		\max_{\mathbf{W}}&\;f(\mathcal{D}_n(\mathbf{W}^H\mathbf{C}_1\mathbf{W},\mathbf{W}^H\mathbf{C}_2\mathbf{W}))\\
		\text{s.t.}&\quad\mathbf{W}\in\mathbb{C}^{m\times n},\,\text{rk}(\mathbf{W})=n,
	\end{split}
	\label{eq:op_THPD}
\end{equation}
where $f$ is a real function in $\mathbb{R}$. In most cases, $f$ is the trivial function $f(x)=x$.\par
Let $\mathbf{C}_1=\mathbf{C}_{\bm{x}}$ and $\mathbf{C}_2=\hat{\mathbf{C}}$, the detection decision of the enhanced detection method is given by
\begin{equation}
	\mathcal{D}_n(\Pi_{\mathbf{W}}(\mathbf{C}_{\bm{x}}),\Pi_{\mathbf{W}}(\hat{\mathbf{C}}))\underset{\mathcal{H}_0}{\overset{\mathcal{H}_1}{\gtrless}}\eta',
\end{equation}
where $\eta'$ is the threshold of the enhanced detector and still obtained by independent Monte Carlo runs.\par
\subsection{Equivalent Enhancement Based on Dual Power Spectrum Manifold}
The optimization (\ref{eq:op_THPD}) is an optimization problem on the Grassmann manifold and often solved by the Riemannian conjugate gradient method\cite{Harandi2018} which requires many iterations and results in the expensive computational burden. 
This part aims to obtain a simpler solution by equivalently transforming the optimization (\ref{eq:op_THPD}) onto $\mathcal{M}_{\text{P}}$.\par
%This part desire to equivalently transform the optimization (\ref{eq:op_THPD}) onto $\mathcal{M}_{\text{P}}$ for the simpler solution.\par
According to Theorem~\ref{the:afi2ipf}, there is a induced potential function $\varphi_{\mathcal{D}_m}$ for the geometric measure $\mathcal{D}_m$ satisfying $\varphi_{\mathcal{D}_m}(\mathcal{P}(\mathbf{C}_1,\mathbf{C}_2))=\mathcal{D}_m(\mathbf{C}_1,\mathbf{C}_2)$. Moreover, considering the enhanced mapping, the following equation is also established
\begin{equation}
	\varphi_{\mathcal{D}_n}(\mathcal{P}(\Pi_{\mathbf{W}}(\mathbf{C}_1),\Pi_{\mathbf{W}}(\mathbf{C}_2)))=\mathcal{D}_n(\Pi_{\mathbf{W}}(\mathbf{C}_1),\Pi_{\mathbf{W}}(\mathbf{C}_2)).
	\label{eq:ps_w_c1c2}
\end{equation}
By the mapping $\mathbf{W}$, the power spectrum $\mathcal{P}(\mathbf{C}_1,\mathbf{C}_2)$ is transformed to $\mathcal{P}(\Pi_{\mathbf{W}}(\mathbf{C}_1),\Pi_{\mathbf{W}}(\mathbf{C}_2))$. In oder to transform the optimization (\ref{eq:op_THPD}) onto $\mathcal{M}_{\text{P}}$, the impact of the mapping $\mathbf{W}$ on $\mathcal{M}_{\text{P}}$ should be studies firstly. In the next contents, the relationship between $\mathcal{P}(\mathbf{C}_1,\mathbf{C}_2)$ and $\mathcal{P}(\Pi_{\mathbf{W}}(\mathbf{C}_1),\Pi_{\mathbf{W}}(\mathbf{C}_2))$ would be discussed.\par
For the matrices $\mathbf{C}_1,\mathbf{C}_2\in\mathcal{M}_{\mathcal{T}H_{++}}$, let $\mathbf{C}'_1=\mathbf{C}_2^{-\frac{1}{2}}\mathbf{C}_1\mathbf{C}_2^{-\frac{1}{2}}$, $\mathbf{C}'_2=\mathbf{I}$, $\mathbf{W}'=\mathbf{C}_2^{\frac{1}{2}}\mathbf{W}$ and suppose the QR decomposition of $\mathbf{W}'$ is $\mathbf{W}'=\mathbf{Q}\mathbf{R}\;(\mathbf{Q}^H\mathbf{Q}=\mathbf{I},\,\mathbf{R}\in GL(m,\mathbb{C}))$, the following equation is established,
\begin{equation}
	\begin{split}
		&\mathcal{D}_n(\mathbf{W}^H\mathbf{C}_1\mathbf{W},\mathbf{W}^H\mathbf{C}_2\mathbf{W})\\
		=&\mathcal{D}_n(\mathbf{W}'^H\mathbf{C}_2^{-\frac{1}{2}}\mathbf{C}_1\mathbf{C}_2^{-\frac{1}{2}}\mathbf{W}',\mathbf{W}'^H\mathbf{W}')\\
		=&\mathcal{D}_n(\mathbf{R}^H\mathbf{Q}^H\mathbf{C}_2^{-\frac{1}{2}}\mathbf{C}_1\mathbf{C}_2^{-\frac{1}{2}}\mathbf{Q}\mathbf{R},\mathbf{R}^H\mathbf{R})\\
		\overset{(a)}{=}&\mathcal{D}_n(\mathbf{Q}^H\mathbf{C}_2^{-\frac{1}{2}}\mathbf{C}_1\mathbf{C}_2^{-\frac{1}{2}}\mathbf{Q},\mathbf{I}),
	\end{split}
	\label{eq:W2Q}
\end{equation}
where $(a)$ is established as the affine invariant property of $\mathcal{D}_n$.
Overall, equation (\ref{eq:W2Q}) means that finding the enhanced mapping $\mathbf{W}$ is equivalent to figuring out the orthogonal mapping $\mathbf{Q}$, and the relationship between them is shown in Fig.\ref{fig:Enhan_Map}.\par
\begin{figure}[htp]
	\centering
	\includegraphics[width=0.4\textwidth]{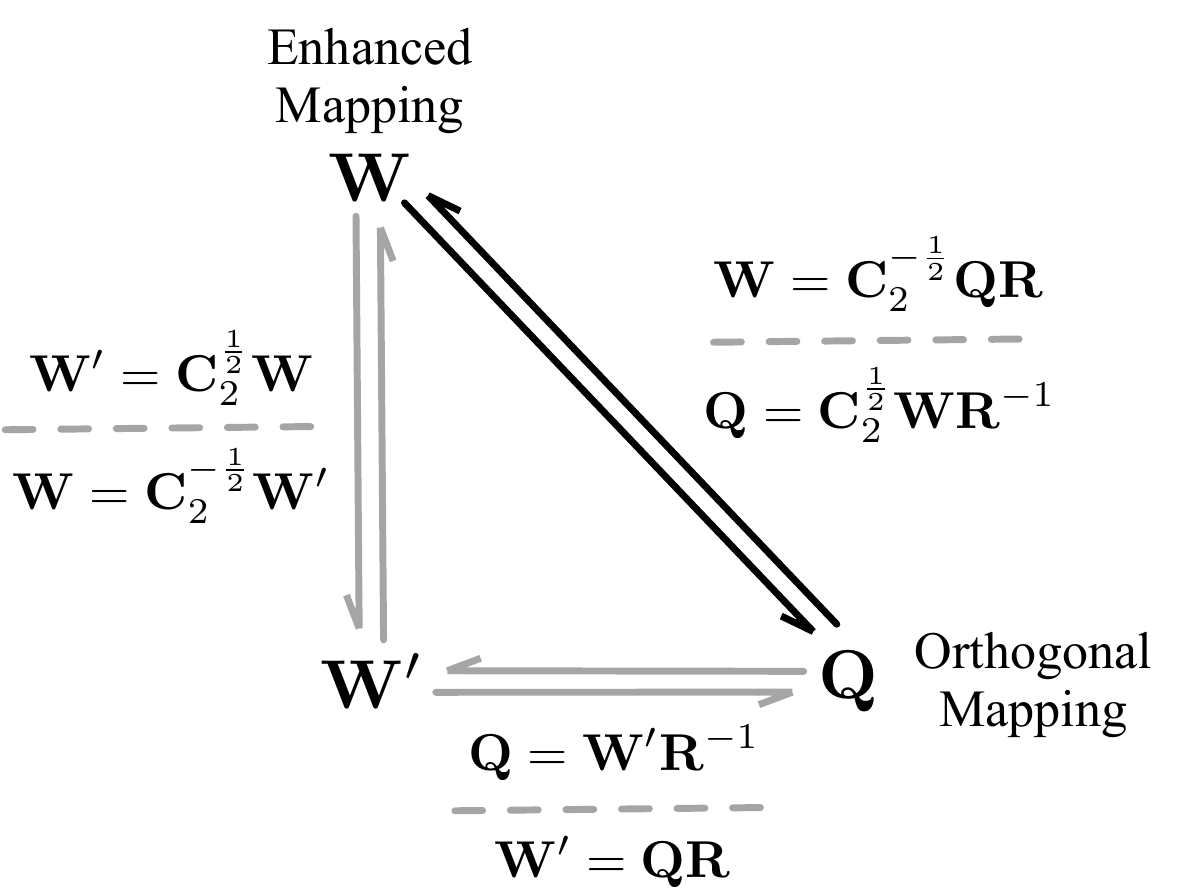}
	\caption{The relationship between the enhanced mapping $\mathbf{W}$ and the orthogonal mapping $\mathbf{Q}$.}
	\label{fig:Enhan_Map}
\end{figure}
Combining the (\ref{eq:ps_w_c1c2}) and (\ref{eq:W2Q}), the power spectrum $\mathcal{P}(\Pi_{\mathbf{W}}(\mathbf{C}_1),\Pi_{\mathbf{W}}(\mathbf{C}_2))$ can be formulated as
\begin{equation}
	\mathcal{P}(\Pi_{\mathbf{W}}(\mathbf{C}_1),\Pi_{\mathbf{W}}(\mathbf{C}_2))=\mathcal{P}(\mathbf{Q}^H\mathbf{C}'\mathbf{Q},\mathbf{I})=\Lambda_{\mathbf{Q}^H\mathbf{C}'\mathbf{Q}},
\end{equation}
where $\mathbf{C}'=\mathbf{C}_2^{-\frac{1}{2}}\mathbf{C}_1\mathbf{C}_2^{-\frac{1}{2}}$. 
%Consequently, as $\mathcal{P}(\mathbf{C}_1,\mathbf{C}_2)=\Lambda_{\mathbf{C}'}$, the power spectrum $\Lambda_{\mathbf{C}'}$ is transformed to $\Lambda_{\mathbf{Q}^H\mathbf{C}'\mathbf{Q}}$ by the mapping $\mathbf{W}$.
\par
%From the perspective of dual power spectrum manifold, by the mapping $\mathbf{W}$, the power spectrum $\Lambda_{\mathbf{C}'}$ is transformed to $\Lambda_{\mathbf{Q}^H\mathbf{C}'\mathbf{Q}}$, that means the enhancement of geometric measures by the enhanced mapping depends on the change of power spectrum, i.e., the eigenvalues of $\mathbf{C}'$ and $\mathbf{Q}^H\mathbf{C}'\mathbf{Q}$.\par
Therefore, by the mapping $\mathbf{W}$, the power spectrum is changed from the eigenvalues of $\mathbf{C}'$ to that of $\mathbf{Q}^H\mathbf{C}'\mathbf{Q}$. In order to investigate the more detailed impact of the mapping $\mathbf{W}$, the varying rule of the eigenvalues  is discussed in the following contents. Firstly, the Cauchy interlace theorem is introduced.
%Actually, the optimization problem (\ref{eq:op_THPD}), that the enhancement of geometric measures, can be transferred onto the dual power spectrum manifold. Before building the equivalent optimization problem on the dual power spectrum manifold, the varying rule of the eigenvalues should be studied first. And, about the eigenvalues of the submatrix, the Cauchy interlace theorem holds.
\begin{theorem}[Cauchy Interlace Theorem\cite{CauthyThe}]
	Let $\mathbf{A}$ be a Hermitian matrix of order $m$, and let $\mathbf{B}$ be a principle submatrix of $\mathbf{A}$ of order $m-1$. If $\lambda_{m-1}\le\lambda_{m-2}\le\cdots\le\lambda_0$ lists the eigenvalues of $\mathbf{A}$ and $\mu_{m-2}\le\mu_{m-3}\le\cdots\le\mu_0$ the eigenvalues of $\mathbf{B}$, then $\lambda_{m-1}\le\mu_{m-2}\le\lambda_{m-2}\le\cdots\le\lambda_1\le\mu_0\le\lambda_0$.
	\label{the:cauthy}
\end{theorem}
\begin{corollary}
	When the order of the principle submatrix $\mathbf{B}$ is $n$, then $\lambda_{i+m-n}\le\mu_i\le\lambda_i\,(0\le i\le n-1)$.
\end{corollary}
\begin{proof}
	By repeating the Cauchy interlace theorem $m-n$ times, this corollary can be derived.
\end{proof}
By the Cauchy interlace theorem and its corollary, the following lemma is established.
\begin{lemma}
	Given matrix $\mathbf{Q}\in\mathbb{C}^{m\times n}$ satisfying $\mathbf{Q}^H\mathbf{Q}=\mathbf{I}$ and suppose the matrix $\mathbf{C}\in\mathcal{M}_{\mathcal{T}H_{++}}$, if the eigenvalues of $\mathbf{C}$ are $\lambda_0,\dots,\lambda_{m-1}\,(\lambda_{m-1}\le\cdots\le\lambda_0)$, then the eigenvalues of $\mathbf{Q}^H\mathbf{C}\mathbf{Q}$, $\mu_0,\dots,\mu_{n-1}\,(\mu_{n-1}\le\cdots\le\mu_0)$, satisfy $\lambda_{i+m-n}\le\mu_i\le\lambda_i\,(0\le i\le n-1)$.
	\label{lem:eig_vary}
\end{lemma}
\begin{proof}
	Let $\widetilde{\mathbf{Q}}$ with order $m$ is the orthonormal expansion of $\mathbf{Q}$, i.e., $\widetilde{\mathbf{Q}}^H\widetilde{\mathbf{Q}}=\widetilde{\mathbf{Q}}\widetilde{\mathbf{Q}}^H=\mathbf{I}$ and
	\begin{equation}
		\widetilde{\mathbf{Q}}=
		\begin{bmatrix}
			\mathbf{Q} & \mathbf{Q}'
		\end{bmatrix}.
	\end{equation}
	Because the eigenvalues of $\widetilde{\mathbf{Q}}^H\mathbf{C}\widetilde{\mathbf{Q}}$ equal to the eigenvalues of $\mathbf{C}\widetilde{\mathbf{Q}}\widetilde{\mathbf{Q}}^H=\mathbf{C}$, the eigenvalues of $\widetilde{\mathbf{Q}}^H\mathbf{C}\widetilde{\mathbf{Q}}$ are $\lambda_0,\dots,\lambda_{m-1}$.\par
	And,
	\begin{equation}
		\widetilde{\mathbf{Q}}^H\mathbf{C}\widetilde{\mathbf{Q}}=
		\begin{bmatrix}
			\mathbf{Q}^H\\
			\mathbf{Q}'^H
		\end{bmatrix}
		\mathbf{C}
		\begin{bmatrix}
			\mathbf{Q} & \mathbf{Q}'
		\end{bmatrix}=
		\begin{bmatrix}
			\mathbf{Q}^H\mathbf{C}\mathbf{Q} & \mathbf{Q}^H\mathbf{C}\mathbf{Q}'\\
			\mathbf{Q}'^H\mathbf{C}\mathbf{Q} & \mathbf{Q}'^H\mathbf{C}\mathbf{Q}'\\
		\end{bmatrix},
	\end{equation}
	thus the matrix $\mathbf{Q}^H\mathbf{C}\mathbf{Q}$ is the principle submatrix of $\widetilde{\mathbf{Q}}^H\mathbf{C}\widetilde{\mathbf{Q}}$. According to the Cauchy interlace theorem, eigenvalue $\mu_i\,(0\le i\le n-1)$ satisfies $\lambda_{i+m-n}\le\mu_i\le\lambda_i\,(0\le i\le n-1)$.
\end{proof}
Lemma~\ref{lem:eig_vary} reveals that the eigenvalues of all the matrices shaped like $\mathbf{Q}^H\mathbf{C}\mathbf{Q}$ should satisfies the inequalities $\lambda_{i+m-n}\le\mu_i\le\lambda_i\,(0\le i\le n-1)$. The remaining question is whether there exists an orthogonal mapping $\mathbf{Q}$ for each eigenvalues $\mu_1,\mu_2,\dots,\mu_n$ satisfying inequalities $\lambda_{i+m-n}\le\mu_i\le\lambda_i\,(0\le i\le n-1)$, which makes the eigenvalues of $\mathbf{Q}^H\mathbf{C}\mathbf{Q}$ actually equal to $\mu_1,\mu_2,\dots,\mu_n$. The following lemma answers this problem.
\begin{lemma}
	Given the matrix $\mathbf{C}\in\mathcal{M}_{\mathcal{T}H_{++}}$ and suppose its eigenvalues are $\lambda_0,\dots,\lambda_{m-1}\,(\lambda_{m-1}\le\cdots\le\lambda_0)$. Then, for each $\mu_i\,(0\le i\le n-1)$ satisfying $\lambda_{i+m-n}\le\mu_i\le\lambda_i\,(0\le i\le n-1)$, there exists a matrix $\mathbf{Q}\in\mathbb{C}^{m\times n}$ satisfying $\mathbf{Q}^H\mathbf{Q}=\mathbf{I}$, which makes the eigenvalues of $\mathbf{Q}^H\mathbf{C}\mathbf{Q}$ are $\mu_0,\dots,\mu_{n-1}$.
	\label{lem:L2Q}
\end{lemma}
\begin{proof}
	Suppose the eigenvalue decomposition of matrix $\mathbf{C}$ is
	\begin{equation}
		\mathbf{C}=\mathbf{V}^H\text{diag}\,(\lambda_0,\dots,\lambda_{m-1})\mathbf{V},
	\end{equation}
	where $\mathbf{V}^H=[\bm{v}_0\;\bm{v}_1\;\cdots\;\bm{v}_{m-1}]$ is an orthogonal matrix.\par
	For each $\mu_i\,(0\le i\le n-1)$ satisfying $\lambda_{i+m-n}\le\mu_i\le\lambda_i\,(0\le i\le n-1)$, let
	\begin{equation}
		t_i=\frac{\mu_i-\lambda_{i+m-n}}{\lambda_i-\lambda_{i+m-n}},
		\label{eq:lm2t}
	\end{equation}
	then $0\le t_i\le 1$ and $\mu_i=t_i\lambda_i+(1-t_i)\lambda_{i+m-n}$.\par
	Let
	\begin{equation}
		\mathbf{Q}^H=
		\begin{bmatrix}
			\sqrt{t_0}\bm{v}^H_0+\sqrt{1-t_0}\bm{v}^H_{m-n}\\
			\sqrt{t_1}\bm{v}^H_1+\sqrt{1-t_1}\bm{v}^H_{1+m-n}\\
			\vdots\\
			\sqrt{t_{n-1}}\bm{v}^H_{n-1}+\sqrt{1-t_{n-1}}\bm{v}^H_{m-1}
		\end{bmatrix},
		\label{eq:q_construct}
	\end{equation}
	then we can get $\mathbf{Q}^H\mathbf{Q}=\mathbf{I}$ because
	\begin{equation}
		\begin{split}
			&(\sqrt{t_i}\bm{v}^H_i+\sqrt{1-t_i}\bm{v}^H_{i+m-n})(\sqrt{t_j}\bm{v}_j+\sqrt{1-t_j}\bm{v}_{j+m-n})\\
			=&\sqrt{t_it_j}\bm{v}^H_i\bm{v}_j+\sqrt{t_i(1-t_j)}\bm{v}^H_i\bm{v}_{j+m-n}\\
			&+\sqrt{(1-t_i)t_j}\bm{v}^H_{i+m-n}\bm{v}_j+\sqrt{(1-t_i)(1-t_j)}\bm{v}^H_{i+m-n}\bm{v}_{j+m-n}\\
			\overset{\text{(a)}}{=}&\sqrt{t_it_j}\bm{v}^H_i\bm{v}_j+\sqrt{(1-t_i)(1-t_j)}\bm{v}^H_{i+m-n}\bm{v}_{j+m-n}
			=\begin{cases}
				0 & i\ne j\\
				1 & i=j
			\end{cases}
		\end{split}
	\end{equation}
	where $(a)$ holds as $i\ne j+m-n$, $j\ne i+m-n$ ($n\le m/2$, so $0\le i,j\le n-1<m-n\le i+m-n,j+m-n\le m-1$).\par
	And, we can get
	\begin{equation}
		\mathbf{V}\mathbf{Q}=
		\begin{bmatrix}
			\sqrt{t_0} & & & \\
			& \sqrt{t_1} & & \\
			& & \ddots & \\
			& & & \sqrt{t_{n-1}}\\
			& & & \\
			& & & \\
			\sqrt{1-t_0} & & & \\
			& \sqrt{1-t_1} & & \\
			& & \ddots & \\
			& & & \sqrt{1-t_{n-1}}\\
		\end{bmatrix}
	\end{equation}
	and
	\begin{equation}
		\begin{split}
			\mathbf{Q}^H\mathbf{C}\mathbf{Q}=&\mathbf{Q}^H\mathbf{V}^H\text{diag}\,(\lambda_0,\dots,\lambda_{m-1})\mathbf{V}\mathbf{Q}\\
		=&\text{diag}\,(\mu_0,\dots,\mu_{n-1}).
		\end{split}
	\end{equation}
	That means the eigenvalues of $\mathbf{Q}^H\mathbf{C}\mathbf{Q}$ are $\mu_0,\dots,\mu_{n-1}$.
\end{proof}
\begin{remark}
	For each $\mu_i\,(0\le i\le n-1)$ satisfying $\lambda_{i+m-n}\le\mu_i\le\lambda_i\,(0\le i\le n-1)$, the above proof provides a construction approach for the orthogonal matrix $\mathbf{Q}$.
\end{remark}
Therefore, according to Theorem~\ref{the:afi2ipf}, Lemma~\ref{lem:L2Q} and above discussion, the optimization problem (\ref{eq:op_THPD}) can be transferred onto the dual power spectrum manifold as shown in the following theorem.
\begin{theorem}
	Given two optimization problems
	\begin{equation}
		\begin{split}
			\max_{\mathbf{W}}\;&f(\mathcal{D}_n(\mathbf{W}^H\mathbf{C}_1\mathbf{W},\mathbf{W}^H\mathbf{C}_2\mathbf{W}))\\
			\text{s.t.}&\quad\mathbf{W}\in\mathbb{C}^{m\times n},\,\text{rk}(\mathbf{W})=n,
		\end{split}
		\label{eq:op_THPD2}
	\end{equation}
	and
	\begin{equation}
		\begin{split}
			&\max_{P=\{\mu_0,\dots,\mu_{n-1}\}} f(\varphi_{\mathcal{D}_n}(P))\\
			s.t.\quad&\lambda_{i+m-n}\le\mu_i\le\lambda_i\quad(0\le i\le n-1),\\
		\end{split}
		\label{eq:op_PS}
	\end{equation}
	where $\lambda_{m-1}\le\cdots\le\lambda_0$ and $\mathcal{P}(\mathbf{C}_1,\mathbf{C}_2)=(\lambda_0,\dots,\lambda_{m-1})$. 
	Suppose the solution of the optimization (\ref{eq:op_PS}) is $P^*=(\mu^*_0,\dots,\mu^*_{n-1})$, then there exists a matrix $\mathbf{W}_*\in\mathbb{C}^{m\times n}$ satisfying $\mathcal{P}(\mathbf{W}_*^H\mathbf{C}_1\mathbf{W}_*,\mathbf{W}_*^H\mathbf{C}_2\mathbf{W}_*)=(\mu^*_0,\dots,\mu^*_{n-1})$ and $\text{rk}(\mathbf{W}_*)=n$. The matrix $\mathbf{W}_*$ is the solution of the optimization (\ref{eq:op_THPD2}).
%	
%	These two optimization have the same optimal value, and for the optimal solution $\mu^*_0,\dots,\mu^*_{n-1}$, the solution $\mathbf{W}_*$ is existed, that satisfies $\mathcal{P}(\mathbf{W}_*^H\mathbf{C}_1\mathbf{W}_*,\mathbf{W}_*^H\mathbf{C}_2\mathbf{W}_*)=(\mu^*_0,\dots,\mu^*_{n-1})$.
	\label{the:op_eq}
\end{theorem}
\begin{proof}
%	Firstly, we will show that the optimal values of these two optimizations are equivalent.\par
	Suppose the solution of the optimization (\ref{eq:op_PS}) is $P^*=(\mu^*_0,\dots,\mu^*_{n-1})$ and the solution of the optimization (\ref{eq:op_THPD2}) is $\mathbf{W}'$.\par
	According to Lemma~\ref{lem:L2Q}, there exists a matrix $\mathbf{Q}_*$ satisfying $\mathbf{Q}_*^H\mathbf{Q}_*=\mathbf{I}$ and that the eigenvalues of $\mathbf{Q}_*^H\mathbf{C}_2^{-\frac{1}{2}}\mathbf{C}_1\mathbf{C}_2^{-\frac{1}{2}}\mathbf{Q}_*$ are $\mu_0,\dots,\mu_{n-1}$. Let $\mathbf{W}_*=\mathbf{C}_2^{-\frac{1}{2}}\mathbf{Q}_*$, then we can get $\mathcal{D}_n(\mathbf{W}_*^H\mathbf{C}_1\mathbf{W}_*,\mathbf{W}_*^H\mathbf{C}_2\mathbf{W}_*)=\mathcal{D}_n(\mathbf{Q}_*^H\mathbf{C}_2^{-\frac{1}{2}}\mathbf{C}_1\mathbf{C}_2^{-\frac{1}{2}}\mathbf{Q}_*,\mathbf{I})$. As $\text{rk}(\mathbf{W}_*)=n$, we can get $f(\mathcal{D}_n(\mathbf{W}'^H\mathbf{C}_1\mathbf{W}',\mathbf{W}'^H\mathbf{C}_2\mathbf{W}'))\ge f(\mathcal{D}_n(\mathbf{W}_*^H\mathbf{C}_1\mathbf{W}_*,\mathbf{W}_*^H\mathbf{C}_2\mathbf{W}_*))=\mathcal{D}_n(\mathbf{Q}_*^H\mathbf{C}_2^{-\frac{1}{2}}\mathbf{C}_1\mathbf{C}_2^{-\frac{1}{2}}\mathbf{Q}_*,\mathbf{I})=f(\varphi_{\mathcal{D}_n}(P^*))$.\par
	Besides, according to (\ref{eq:W2Q}), there exists a matrix $\mathbf{Q}$ satisfying $\mathbf{Q}^H\mathbf{Q}=\mathbf{I}$ and $\mathcal{D}_n(\mathbf{W}'^H\mathbf{C}_1\mathbf{W}',\mathbf{W}'^H\mathbf{C}_2\mathbf{W}')=\mathcal{D}_n(\mathbf{Q}^H\mathbf{C}_2^{-\frac{1}{2}}\mathbf{C}_1\mathbf{C}_2^{-\frac{1}{2}}\mathbf{Q},\mathbf{I})$. As Lemma~\ref{lem:eig_vary}, the eigenvalues $\mu'_0,\dots,\mu'_{n-1}$ of $\mathbf{Q}^H\mathbf{C}_2^{-\frac{1}{2}}\mathbf{C}_1\mathbf{C}_2^{-\frac{1}{2}}\mathbf{Q}$ should satisfy the inequality $\lambda_{i+m-n}\le\mu'_i\le\lambda_i\;(0\le i\le n-1)$. Therefore, we can get $f(\mathcal{D}_n(\mathbf{W}'^H\mathbf{C}_1\mathbf{W}',\mathbf{W}'^H\mathbf{C}_2\mathbf{W}'))=f(\varphi_{\mathcal{D}_n}(\mu'_1,\dots,\mu'_n))\le f(\varphi_{\mathcal{D}_n}(P^*))$.\par 
	Totally, we can get $f(\mathcal{D}_n(\mathbf{W}'^H\mathbf{C}_1\mathbf{W}',\mathbf{W}'^H\mathbf{C}_2\mathbf{W}'))=f(\mathcal{D}_n(\mathbf{W}_*^H\mathbf{C}_1\mathbf{W}_*,\mathbf{W}_*^H\mathbf{C}_2\mathbf{W}_*))=f(\varphi_{\mathcal{D}_n}(P^*))$, i.e., $\mathbf{W}_*$ is also the solution of optimization (\ref{eq:op_THPD2}).\par
\end{proof}
%\begin{remark}
%	The decision variable of optimization problem (\ref{eq:op_THPD2}) is from $\mathbb{C}^{m\times n}$ and should satisfy the full column rank constraint, and the decision variable of optimization problem (\ref{eq:op_PS}) is from $\mathbb{R}^n$ and the constraint of each component $\mu_k$ is independent, $\lambda_{k+m-n}\le\mu_k\le\lambda_k$.
%\end{remark}
Suppose the solution of (\ref{eq:op_PS}) is $(\mu^*_0,\dots,\mu^*_{n-1})$, then the solution $\mathbf{W}_*$ of (\ref{eq:op_THPD2}) can be obtained by the optimal solution $(\mu^*_0,\dots,\mu^*_{n-1})$ and matrices $\mathbf{C}_1,\mathbf{C}_2$, which is shown in algorithm~\ref{alg:P2W}.
%Actually, in the last step, the original equation is $\mathbf{W}=\mathbf{C}_2^{-\frac{1}{2}}\mathbf{Q}\mathbf{R}$, and $\mathbf{R}$ can be any full rank matrix of order $n$. So, we set $\mathbf{R}=\mathbf{I}$ for simplicity.
\par
\begin{algorithm}[htb]
	\setstretch{1.1}
	\caption{Enhanced mapping calculation by $\mu^*_0,\dots,\mu^*_{n-1}$}
	\label{alg:P2W}
	\renewcommand{\algorithmicrequire}{\textbf{Input:}}
	\renewcommand{\algorithmicensure}{\textbf{Output:}}
	\begin{algorithmic}[1]
		\REQUIRE ~~\\
		The optimal power spectrum $P=(\mu^*_0,\dots,\mu^*_{n-1})$.\\
		The matrices $\mathbf{C}_1,\mathbf{C}_2$.\\
		\ENSURE ~~\\
		The enhanced mapping $\mathbf{W}_*$.\\
		\STATE Calculating $\mathbf{C}=\mathbf{C}_2^{-\frac{1}{2}}\mathbf{C}_1\mathbf{C}_2^{-\frac{1}{2}}$.
		\STATE Performing the eigenvalue decomposition of $\mathbf{C}$ that $\mathbf{C}=\mathbf{V}^H\text{diag}\,(\lambda_0,\dots,\lambda_{m-1})\mathbf{V},\;\mathbf{V}^H=[\bm{v}_0\;\bm{v}_1\;\cdots\;\bm{v}_{m-1}]$.\\
		\STATE Calculating $t_0,\dots,t_{n-1}$ with $\lambda_0,\dots,\lambda_{m-1}$ and $\mu^*_0,\dots,\mu^*_{n-1}$ by (\ref{eq:lm2t}).\\
		\STATE Calculating $\mathbf{Q}$ with $t_0,\dots,t_{n-1}$ and $\bm{v}_0,\dots,\bm{v}_{m-1}$ by (\ref{eq:q_construct}).\\
		\STATE Calculating enhanced mapping by $\mathbf{W}_*=\mathbf{C}_2^{-\frac{1}{2}}\mathbf{Q}$.\\
	\end{algorithmic}
\end{algorithm}
Compared with the optimization (\ref{eq:op_THPD2}) and (\ref{eq:op_PS}), the decision variable of (\ref{eq:op_THPD2}) is a matrix $\mathbf{W}$ with dimension $m\times n$, and it is $(\mu_0,\dots,\mu_{n-1})$ with dimension $n$ for (\ref{eq:op_PS}), that means the dimension of gradient is $m\times n$ for (\ref{eq:op_THPD2}) and $n$ for (\ref{eq:op_PS}) in the gradient descent method. Moreover, the constraint condition of (\ref{eq:op_THPD2}) is $\text{rk}(\mathbf{W})=n$ which is independent for each element of $\mathbf{W}$, and that of (\ref{eq:op_PS}) is just $\lambda_{i+m-n}\le\mu_i\le\lambda_i$ that the feasible region is just a $n$-dimensional cube. 
%As a result, the enhancement of the geometric measure on $\mathcal{M}_{\mathcal{T}H_{++}}$ as (\ref{eq:op_THPD2}) can be transformed to the optimization problem on $\mathcal{M}_{\text{P}}$, which is formulated as (\ref{eq:op_PS}) and means the enhancement of induced potential function. Moreover, for the optimal power spectrum $P=(\mu^*_0,\dots,\mu^*_{n-1})$, the corresponding enhanced mapping $\mathbf{W}$ can be obtained by algorithm~\ref{alg:P2W}.
\subsection{Numerical Example}
%Compared with the enhancement on Toeplitz HPD matrix manifold, the enhancement of the induced potential function on $\mathcal{M}_{\text{P}}$ is easier to solve. 
Compared with the optimization (\ref{eq:op_THPD2}), the optimization (\ref{eq:op_PS}) is easier to solve. 
In the following example, the closed-form solution of the optimization (\ref{eq:op_PS}) is figured out.
\begin{figure*}[htpb]
	\centering
	\subfigure[Riemannian distance]{\includegraphics[width=0.32\textwidth]{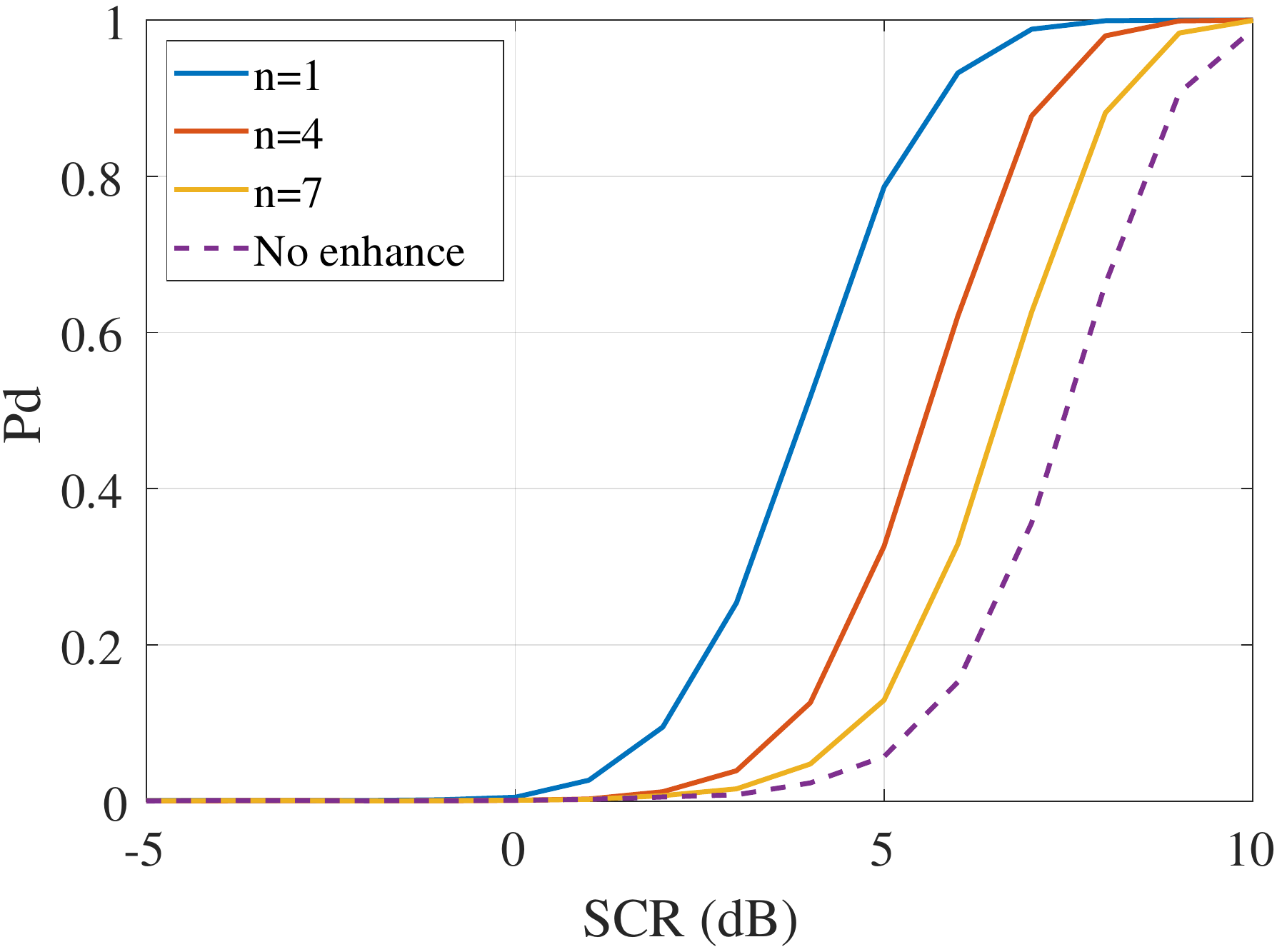}}
	\subfigure[Kullback-Leibler divergence]{\includegraphics[width=0.32\textwidth]{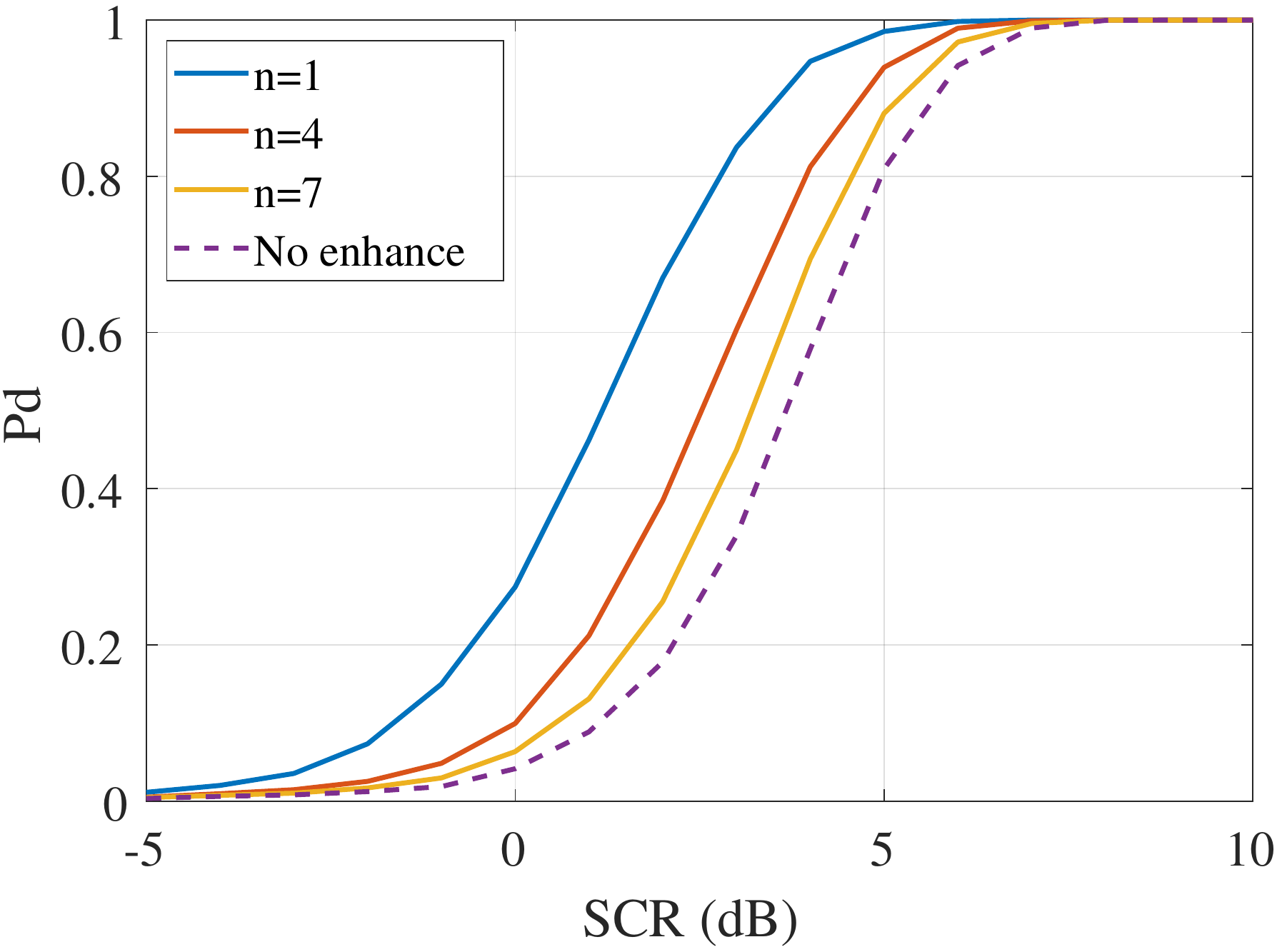}}
	\subfigure[log-determinant divergence]{\includegraphics[width=0.32\textwidth]{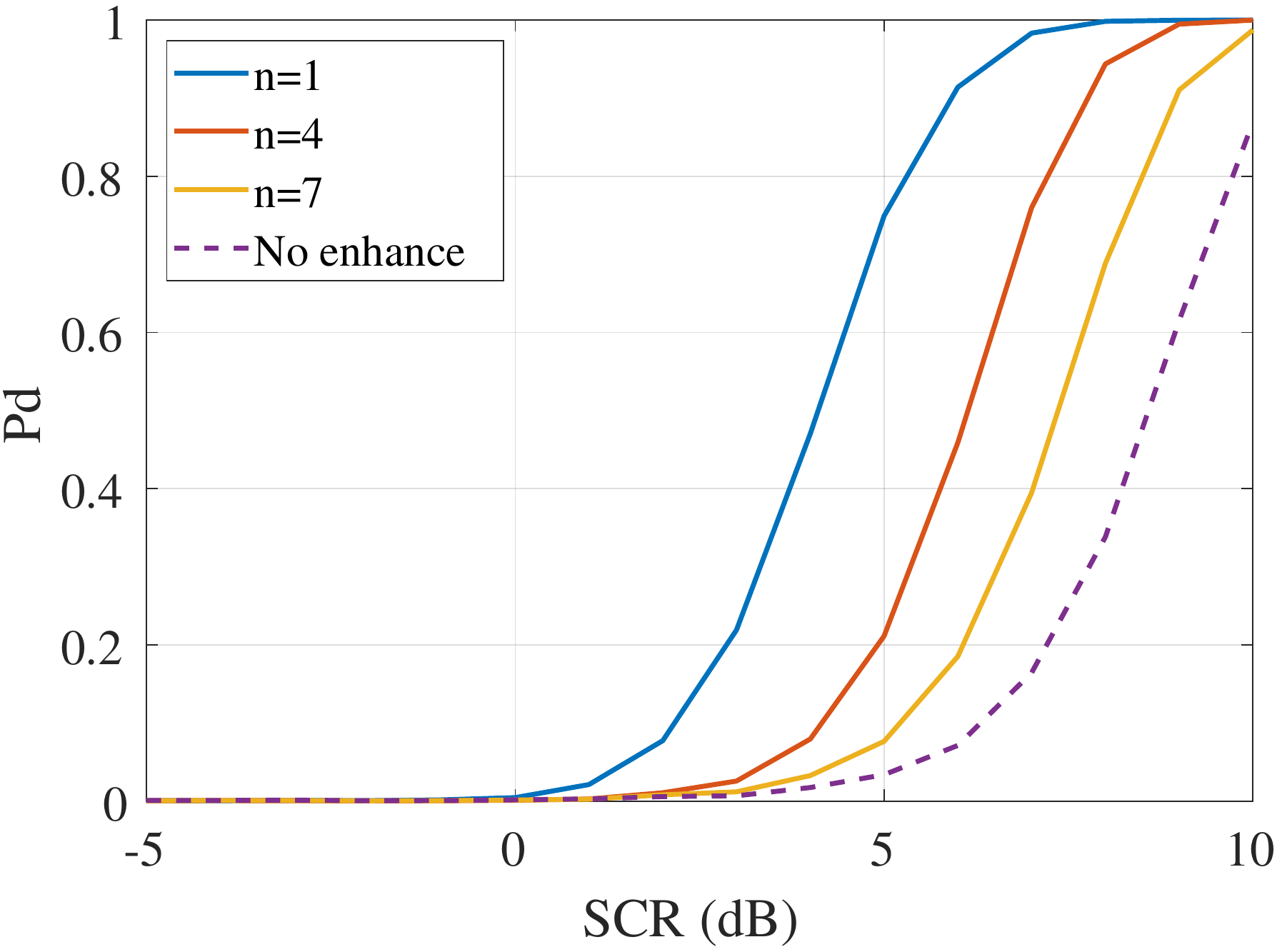}}
%	\subfigure[Enhanced Riemannian distance of $\bm{s}_2$]{\includegraphics[width=0.32\textwidth]{per_rm_s2.pdf}}
%	\subfigure[Enhanced Kullback-Leibler divergence of $\bm{s}_2$]{\includegraphics[width=0.32\textwidth]{per_kld_s2.pdf}}
%	\subfigure[Enhanced log-determinant divergence of $\bm{s}_2$]{\includegraphics[width=0.32\textwidth]{per_ldd_s2.pdf}}
%	\subfigure[Enhanced Riemannian distance of $\bm{s}_3$]{\includegraphics[width=0.32\textwidth]{per_rm_s3.pdf}}
%	\subfigure[Enhanced Kullback-Leibler divergence of $\bm{s}_3$]{\includegraphics[width=0.32\textwidth]{per_kld_s3.pdf}}
%	\subfigure[Enhanced log-determinant divergence of $\bm{s}_3$]{\includegraphics[width=0.32\textwidth]{per_ldd_s3.pdf}}
	\caption{The detection performance of enhanced detector with SCR from -5 dB to 10 dB. In these figures, the detection probabilities are calculated based on $10^4$ Monte Carlo runs.}
	\label{fig:curve_pd_enhance}
\end{figure*}
\begin{example}
	Consider the radar target detection scenario, the details of the simulation parameters are listed in the Table~\ref{Tab:para}. The primary data is from the cell under test, and the secondary data are from the reference cells.\par
	\begin{table}[hbt]
	  \caption{Settings of the Simulation.}
	  \label{Tab:para}
	  \centering
	  \begin{tabular}{|l|c|}
	  	\hline
	   	Parameters or Variables & Setting\\
	   	\hline
	   	\hline
	    Number of range cells & 17 \\
	    Number of pulses & 15\\
	    False-alarm Probability & $P_f=10^{-3}$\\
	    Target cell & $9^{\text{th}}$\\
	    Pulse repetition frequency & $f_r=1000$ Hz\\
	    \hline
	    \multirow{2}{*}{K-distribution clutter} & Shape parameter 1\\
	     & Scalar parameter 0.5\\
	    \hline
	  \end{tabular}
	\end{table}
%	The setting of this example is as same as Example~\ref{exp:exp}, but only the $\bm{s}_1,\bm{s}_2,\bm{s}_3$ are considered.
	The target echo steering vector is
	\begin{equation}
		\bm{s}=A[1\;\text{exp}(j2\pi f_df_r^{-1})\;\cdots\;\text{exp}(j2\pi (M-1)f_df_r^{-1})],
	\end{equation}
	where $A$ is the power of target echo, and $f_d$ is the Doppler frequency which is 135 Hz in this example.\par
	Let the objective function in (\ref{eq:op_THPD2}) be $f(x)=x$, the enhancements of $\mathcal{D}_{\text{RD}},\mathcal{D}_{\text{KL}},\mathcal{D}_{\text{LD}}$ are discussed.
	\label{exp:exp}
\end{example}
The enhancement of the geometric measures $\mathcal{D}_{\text{RD}},\mathcal{D}_{\text{KL}},\mathcal{D}_{\text{LD}}$ are formulated as follows,
\begin{equation}
	\begin{split}
		&\max_{\mathbf{W}}\mathcal{D}_{\text{RD}}(\mathbf{W}^H\mathbf{C}_1\mathbf{W},\mathbf{W}^H\mathbf{C}_2\mathbf{W})\\
		=&\max_{\mathbf{W}}\Arrowvert \log((\mathbf{W}^H\mathbf{C}_1\mathbf{W})^{-\frac{1}{2}}\mathbf{W}^H\mathbf{C}_2\mathbf{W}(\mathbf{W}^H\mathbf{C}_1\mathbf{W})^{-\frac{1}{2}})\Arrowvert_F^2,
	\end{split}
	\label{eq:THPD_rem_op}
\end{equation}
\begin{equation}
	\begin{split}
		\max_{\mathbf{W}}\mathcal{D}_{\text{KL}}(\mathbf{W}^H\mathbf{C}_1\mathbf{W},\mathbf{W}^H\mathbf{C}&_2\mathbf{W})\\
		=\max_{\mathbf{W}}\text{tr}(\mathbf{W}^H\mathbf{C}_1\mathbf{W}(\mathbf{W}\mathbf{C}_2\mathbf{W}^H&)^{-1}-\mathbf{I})\\
	-\log&\left|\mathbf{W}^H\mathbf{C}_1\mathbf{W}(\mathbf{W}\mathbf{C}_2\mathbf{W}^H)^{-1}\right|,
	\end{split}
	\label{eq:THPD_KLD_op}
\end{equation}
\begin{equation}
	\begin{split}
		&\max_{\mathbf{W}}\mathcal{D}_{\text{RD}}(\mathbf{W}^H\mathbf{C}_1\mathbf{W},\mathbf{W}^H\mathbf{C}_2\mathbf{W})\\
		=&\max_{\mathbf{W}}\log\left|\frac{\mathbf{W}^H\mathbf{C}_1\mathbf{W}+\mathbf{W}^H\mathbf{C}_2\mathbf{W}}{2}\right|-\log\sqrt{\left|\mathbf{W}^H\mathbf{C}_1\mathbf{W}\mathbf{W}^H\mathbf{C}_2\mathbf{W}\right|}.
	\end{split}
	\label{eq:THPD_LDD_op}
\end{equation}
By Theorem~\ref{the:op_eq}, the equivalent enhancements on $\mathcal{M}_{\text{P}}$ are 
\begin{equation}
	\begin{split}
		&\max_{P=\{\mu_0,\dots,\mu_{n-1}\}} \varphi_{\mathcal{D}_{\text{RD}}}(P)=\sum^{n-1}_{k=0} \log^2\mu_k\\
		s.t.&\quad\lambda_{i+m-n}\le\mu_i\le\lambda_i\quad(0\le i\le n-1),\\ 
	\end{split}
	\label{eq:PSD_rem_op}
\end{equation}
\begin{equation}
	\begin{split}
		&\max_{P=\{\mu_0,\dots,\mu_{n-1}\}} \varphi_{\mathcal{D}_{\text{KL}}}(P)=\sum^{n-1}_{k=0} \mu_k-1-\log\mu_k\\
		&s.t.\quad\lambda_{i+m-n}\le\mu_i\le\lambda_i\quad(0\le i\le n-1),\\ 
	\end{split}
	\label{eq:PSD_KLD_op}
\end{equation}
\begin{equation}
	\begin{split}
		&\max_{P=\{\mu_0,\dots,\mu_{n-1}\}} \varphi_{\mathcal{D}_{\text{LD}}}(P)=\sum^{n-1}_{k=0} \log\frac{\mu_k+1}{2\sqrt{\mu_k}}\\
		&s.t.\quad\lambda_{i+m-n}\le\mu_i\le\lambda_i\quad(0\le i\le n-1),\\ 
	\end{split}
	\label{eq:PSD_LDD_op}
\end{equation}
where $\lambda_{m-1}\le\cdots\le\lambda_0$ and $\mathcal{P}(\mathbf{C}_1,\mathbf{C}_2)=\{\lambda_0,\dots,\lambda_{m-1}\}$.\par
The optimization problems of (\ref{eq:PSD_rem_op}-\ref{eq:PSD_LDD_op}) are much easier than problems of (\ref{eq:THPD_rem_op}-\ref{eq:THPD_LDD_op}), because the variables $\mu_i\,(i=0,\dots,n-1)$ are independent in the objective function and constraints. Therefore, these optimization problems can be divided into the summation of subproblem as
\begin{equation}
	\sum^{n-1}_{k=0}\max_{\lambda_{k+m-n}\le\mu_k\le\lambda_k}\log^2\mu_k,
	\label{eq:PSD_rem_op_sum}
\end{equation}
\begin{equation}
	\sum^{n-1}_{k=0}\max_{\lambda_{k+m-n}\le\mu_k\le\lambda_k}\mu_k-1-\log\mu_k,
	\label{eq:PSD_KLD_op_sum}
\end{equation}
\begin{equation}
	\sum^{n-1}_{k=0} \max_{\lambda_{k+m-n}\le\mu_k\le\lambda_k}\log\frac{\mu_k+1}{2\sqrt{\mu_k}}.
	\label{eq:PSD_LDD_op_sum}
\end{equation}
Thus, they can be solved by the derivate of the objective function
\begin{equation}
	\frac{\text{d}\log^2\mu_k}{\text{d} \mu_k}=2\frac{\log \mu_k}{\mu_k},
	\label{eq:dev_rem_op}
\end{equation}
\begin{equation}
	\frac{\text{d}\mu_k-1-\log\mu_k}{\text{d} \mu_k}=1-\frac{1}{\mu_k},
	\label{eq:dev_kl_op}
\end{equation}
\begin{equation}
	\frac{\text{d}\log\frac{\mu_k+1}{2\sqrt{\mu_k}}}{\text{d} \mu_k}=\frac{\mu_k-1}{2(1+\mu_k)\mu_k}.
	\label{eq:dev_ld_op}
\end{equation}
According to their derivates, these three sorts of objective functions are monotonically decreasing in $(0,1]$ and monotonically increasing in $[1,\infty)$. Therefore, the optimal $\mu^*_i$ is chosen from the two ends of the interval $[\lambda_{i+m-n},\lambda_i]$, i.e.,
\begin{equation}
	\mu^*_i|_{\varphi_{\mathcal{D}_{\text{RD}}}}=\arg\max_{\mu_i\in\{\lambda_i,\lambda_{i+m-n}\}}\log^2\mu_i,
\end{equation}
\begin{equation}
	\mu^*_i|_{\varphi_{\mathcal{D}_{\text{KL}}}}=\arg\max_{\mu_i\in\{\lambda_i,\lambda_{i+m-n}\}}\mu_i-1-\log\mu_i,
\end{equation}
\begin{equation}
	\mu^*_i|_{\varphi_{\mathcal{D}_{\text{LD}}}}=\arg\max_{\mu_i\in\{\lambda_i,\lambda_{i+m-n}\}}\log\frac{\mu_i+1}{2\sqrt{\mu_i}}.
\end{equation}
Then the enhanced mapping can be obtained by the algorithm~\ref{alg:P2W}, which is
\begin{equation}
	\mathbf{W}_*=\mathbf{C}_2^{-\frac{1}{2}}
	\begin{bmatrix}
		\sqrt{\frac{\mu^*_0-\lambda_{m-n}}{\lambda_0-\lambda_{m-n}}}\bm{v}^H_0+\sqrt{\frac{\lambda_0-\mu^*_0}{\lambda_0-\lambda_{m-n}}}\bm{v}^H_{m-n}\\
		\sqrt{\frac{\mu^*_1-\lambda_{m-n+1}}{\lambda_1-\lambda_{m-n+1}}}\bm{v}^H_1+\sqrt{\frac{\lambda_1-\mu^*_1}{\lambda_1-\lambda_{m-n+1}}}\bm{v}^H_{m-n+1}\\
		\vdots\\
		\sqrt{\frac{\mu^*_{n-1}-\lambda_{m-1}}{\lambda_{n-1}-\lambda_{m-1}}}\bm{v}^H_{n-1}+\sqrt{\frac{\lambda_{n-1}-\mu^*_{n-1}}{\lambda_{n-1}-\lambda_{m-1}}}\bm{v}^H_{m-1}
	\end{bmatrix},
\end{equation}
where $\bm{v}_0,\bm{v}_1,\cdots,\bm{v}_{m-1}$ are the corresponding eigenvectors of $\lambda_0,\dots,\lambda_{m-1}$.\par
\emph{Computational efficiency}: The computation of solving the optimization (\ref{eq:PSD_rem_op}-\ref{eq:PSD_LDD_op}) is without iterations, and the main computational burden is from working out each $\mu^*_k$ by a closed analytic expression. So, the computational complexity is $\mathcal{O}(m)$ for solving (\ref{eq:PSD_rem_op}-\ref{eq:PSD_LDD_op}). The previous methods\cite{Yang2020} need the gradient descent-based method to solve the optimization problem, and the SVD (singular value decomposition) is required in each iteration. That means the computational complexity is $\mathcal{O}(km^3)$ ($k$ indicates the number of iteration) for directly solving (\ref{eq:THPD_rem_op}-\ref{eq:THPD_LDD_op}) on Grassmannian manifold.
In addition, the transformations from (\ref{eq:THPD_rem_op}-\ref{eq:THPD_LDD_op}) to (\ref{eq:PSD_rem_op}-\ref{eq:PSD_LDD_op}) depends on the computation of eigenvalues, so the computation complexity of the transformations is $\mathcal{O}(m^3)$.
Totally, the computational complexity of solving the optimization (\ref{eq:THPD_rem_op}-\ref{eq:THPD_LDD_op}) descends from $\mathcal{O}(km^3)$ to $\mathcal{O}(m^3)$ by transforming them to (\ref{eq:PSD_rem_op}-\ref{eq:PSD_LDD_op}).\par
The detection performances of the enhanced detectors are shown in Fig.\ref{fig:curve_pd_enhance}. According to this experiment, the following conclusions can be drawn:
\begin{itemize}
	\item The closed-form solution of the enhanced mapping is figured out by transforming (\ref{eq:THPD_rem_op}-\ref{eq:THPD_LDD_op}) to (\ref{eq:PSD_rem_op}-\ref{eq:PSD_LDD_op}).
	\item The complexity of computing the enhanced mapping descends from $\mathcal{O}(km^3)$ to $\mathcal{O}(m^3)$.
	\item KL divergence is the best among these three sorts of geometric measures.
	\item The improvement of the enhancement of KLD is less than others.
	\item The smaller dimension $n$ results in the better performance of the enhanced detector.
\end{itemize}
%\begin{itemize}
%	\item Among the three sorts of geometric measures, the optimal $n$ of the signal $\bm{s}_k$ is $k$, i.e., the optimal dimension of enhanced mapping equals to the discrete bandwidth of target echo, that corresponds the theoretical analysis in section~\ref{sec:enhance_detector}.
%	\item The closer dimension of the enhanced mapping to the discrete bandwidth encourages the better detection performance.
%	\item The enhancement of detection performance is obvious, while the discrete bandwidth of the target echo is small, i.e., the spectrum is narrow.
%	\item KL divergence is the best among these three sorts of geometric measures, in terms of the detection performance. And, the improvement of the enhancement of KLD is less than others.
%\end{itemize}

\section{Analytic Method Based on Dual Power Spectrum Manifold for Detection Performance}\label{sec:perform}
As discussed in section~\ref{sec:THPD_SP}, the affine invariant geometric measures on $\mathcal{M}_{\mathcal{T}H_{++}}$ is equivalent to a corresponding induced potential function on $\mathcal{M}_{\text{P}}$. So, the geometric detector can be expressed as a simpler form by the induced potential function than the original form, and the simpler form can more obviously reveal the relation between the test statistics and the characteristic of the signal, which would benefit the analysis of the detection performance. Moreover, the remaining problem in the enhanced detection, that the selection criterion of the dimension of the enhanced mapping, is also possible to be solved from the view of the induced potential function. This section would introduce the analytic method of the detection performance and discuss the optimal dimension of the enhanced mapping on $\mathcal{M}_{\text{P}}$.\par
%geometric detectors on the $\mathcal{M}_{\mathcal{T}H_{++}}$ and study their performance on the $\mathcal{M}_{\text{P}}$.\par
%The brief contents are presented as follows. In the first part, the detection scheme of geometric methods is provided. The second part introduces an analysis approach based on the induced potential function to analyze the advantages of the selected geometric measures. Moreover, as examples, three typical geometric measures is analyzed. In the last part, the numeric experiments are presented to prove the effectiveness of the analysis approach, and a series of conclusions, in terms of detection performance, about three mentioned typical geometric measures is provided.
\subsection{Performance Analysis Based on the Induced Potential Function}
%the Partial Derivative of 
%Besides the SCR, the characteristics of signals are also a significant factor to the detection performance. 
%According to Theorem~\ref{the:afi2ipf}, the affine invariant geometric measure $\mathcal{D}$ can be substituted by the induced potential function $\varphi_{\mathcal{D}}$, so geometric detectors using the affine invariant geometric measure can be reformulated as induced potential function-forms.
As the decision rule (\ref{eq:dec_rul}), the large geometric measure means that the signal has a great chance to be detected, and the small one is the opposite.
In addition, according to Theorem~\ref{the:afi2ipf}, the geometric detector using the affine invariant geometric measure can be reformulated as the form based on the induced potential function, so the detection performance is related to the value of the induced potential function. As illustrated by Fig.\ref{fig:analytic_method}, the maximal point of the induced potential function means the best detection performance, and the advantageous characteristics, which benefit the detection performance, can be deduced by analyzing the corresponding power spectrum to the maximal point. 
Similarly, the disadvantageous characteristics can be analyzed by the corresponding power spectrum to the minimal point.
Without loss of generality, only the advantageous characteristics for the detection performance is considered in this part. In the next contents, we use the partial derivative to obtain the maximal point and deduce the advantageous characteristics.\par
%This part mainly focuses on the advantageous characteristics for the detection performance. We use the partial derivative to obtain the maximal point and deduce the advantageous characteristics for the detection performance.\par
%In this form, the test statistic relies in the power spectrum of the whitened signal, so this power spectrum is a decisive factor in terms of the detection performance.\par
\begin{figure}[htp]
	\centering
	\includegraphics[width=0.49\textwidth]{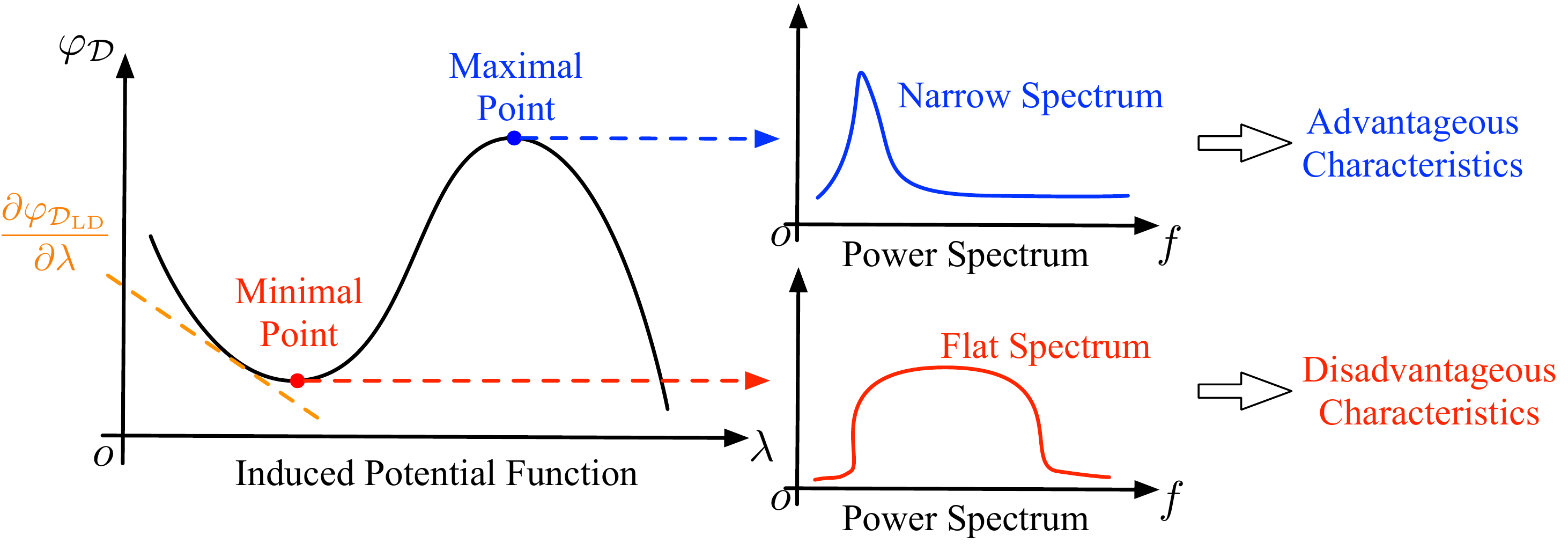}
	\caption{The variable of the induced potential function corresponds to the power spectrum of the observed data. The advantageous and disadvantageous characteristics for detection performance are embedded in the corresponding power spectrums to the minimal point and the maximal point.}
	\label{fig:analytic_method}
\end{figure}
Firstly, the background settings are introduced as follows. Suppose that the estimate $\hat{\mathbf{C}}$ precisely equals to the Toeplitz covariance matrix of the clutter in the primary data. That means the Toeplitz HPD matrix $\mathbf{C}_{\bm{x}}$ and matrix $\hat{\mathbf{C}}$ can be expressed as
\begin{equation}
	\begin{cases}
		\mathbf{C}_{\bm{x}}=\mathbf{C}_{\bm{s}}+\mathbf{C}_{\bm{w}},\\
		\hat{\mathbf{C}}=\mathbf{C}_{\bm{w}},\\
	\end{cases}
\end{equation}
where $\mathbf{C}_{\bm{w}}$ is the Toeplitz covariance matrix of the clutter in the primary data.\par
According to Lemma~\ref{lem:lam2p} and the attached remark, the matrix can be approximated to the following form,
\begin{equation}
	\begin{cases}
		\mathbf{C}_{\bm{x}}\approx\mathbf{F}\text{diag}\,(\lambda_0^s+\lambda'_0,\dots,\lambda^s_{m-1}+\lambda'_{m-1})\mathbf{F}^H,\\
		\hat{\mathbf{C}}\approx\mathbf{F}\text{diag}\,(\lambda'_0,\dots,\lambda'_{m-1})\mathbf{F}^H,\\
	\end{cases}
\end{equation}
where $\lambda^s_k\,(k=0,\dots,m-1)$ are the eigenvalues of $\mathbf{C}_{\bm{s}}$, and $\lambda'_k\,(k=0,\dots,m-1)$ are the eigenvalues of $\mathbf{C}_{\bm{w}}$. Let $\lambda^*_k=\lambda^s_k/\lambda'_k$, the $\lambda^*_k$ can be regarded as the power of the $k^{\text{th}}$ frequency component of the target echo whitened by the clutter according to Lemma~\ref{lem:lam2p} and Definition~\ref{def:map}.
Then $\mathcal{P}(\mathbf{C}_{\bm{x}},\hat{\mathbf{C}})=(\lambda^*_0+1,\dots,\lambda^*_{m-1}+1)$ is established. 
%Meanwhile, according to Lemma~\ref{lem:lam2p} and Definition~\ref{def:map}, $\lambda^*_k$ also can be regarded as the power of the $k^{\text{th}}$ frequency component of the target echo whitened by the clutter.\par
Because the summation of eigenvalues equals to the trace, the following equation is established,
\begin{equation}
	\sum^{m-1}_{k=0}\lambda^*_k=mc_0|_{\mathcal{H}_0}=m\sigma^2\;\Leftrightarrow\; \lambda^*_0=m\sigma^2-\sum^{m-1}_{k=1}\lambda^*_k,
\end{equation}
where $\sigma^2$ indicates the SCR. Then, under fixed SCR $\sigma^2$, the induced potential function can be reformulated as the following adjusted form,
\begin{equation}
	\begin{split}
		&\varphi'_{\mathcal{D}}(\lambda^*_1,\dots,\lambda^*_{m-1})\\
		=&\varphi_{\mathcal{D}}\left(1+m\sigma^2-\sum^{m-1}_{k=1}\lambda^*_k,1+\lambda^*_1,\dots,1+\lambda^*_{m-1}\right).
	\end{split}
\end{equation}\par
As instances, the typical geometric measures are analyzed by this approach. The adjusted induced potential functions with respect to $\lambda^*_1,\dots,\lambda^*_{m-1}$ are as follows,
\begin{equation}
	\begin{split}
		\varphi'_{\mathcal{D}_{\text{RD}}}&(\lambda^*_1,\dots,\lambda^*_{m-1})\\
		&=\sum^{m-1}_{k=1}\log^2(1+\lambda^*_k)+\log^2(1+m\sigma^2-\sum^{m-1}_{k=1}\lambda^*_k),
	\end{split}
	\label{eq:fun_rem}
\end{equation}
\begin{equation}
	\begin{split}
		\varphi'_{\mathcal{D}_{\text{KL}}}&(\lambda^*_1,\dots,\lambda^*_{m-1})\\
		&=m\sigma^2-\sum^{m-1}_{k=0}\log(1+\lambda^*_k)-\log(1+m\sigma^2-\sum^{m-1}_{k=1}\lambda^*_k),
	\end{split}
	\label{eq:fun_kl}
\end{equation}
%\begin{equation}
%	\begin{split}
%		\varphi'_{\mathcal{D}_{\text{JS}}}&(\lambda^*_1,\dots,\lambda^*_{m-1})\\
%		&=\sum^{m-1}_{k=1}\log\frac{2+\lambda^*_k}{2\sqrt{\lambda^*_k+1}}+\log\frac{2+m\sigma^2-\sum^{m-1}_{k=1}\lambda^*_k}{2\sqrt{1+m\sigma^2-\sum^{m-1}_{k=1}\lambda^*_k}},
%	\end{split}
%	\label{eq:fun_js}
%\end{equation}
\begin{equation}
	\begin{split}
		\varphi'_{\mathcal{D}_{\text{LD}}}&(\lambda^*_1,\dots,\lambda^*_{m-1})\\
		&=\sum^{m-1}_{k=1}\log\frac{2+\lambda^*_k}{2\sqrt{\lambda^*_k+1}}+\log\frac{2+m\sigma^2-\sum^{m-1}_{k=1}\lambda^*_k}{2\sqrt{1+m\sigma^2-\sum^{m-1}_{k=1}\lambda^*_k}},
	\end{split}
	\label{eq:fun_ld}
\end{equation}
and the domain of definition is
\begin{equation}
	\left\{(\lambda^*_1,\dots,\lambda^*_{m-1})\;\bigg|\;\lambda^*_k\ge 0,\,\sum^{m-1}_{k=1}\lambda^*_k\le m\sigma^2\right\}.
\end{equation}
\begin{figure*}[thp]
	\centering
	\subfigure[Riemannian distance]{\includegraphics[width=0.32\textwidth]{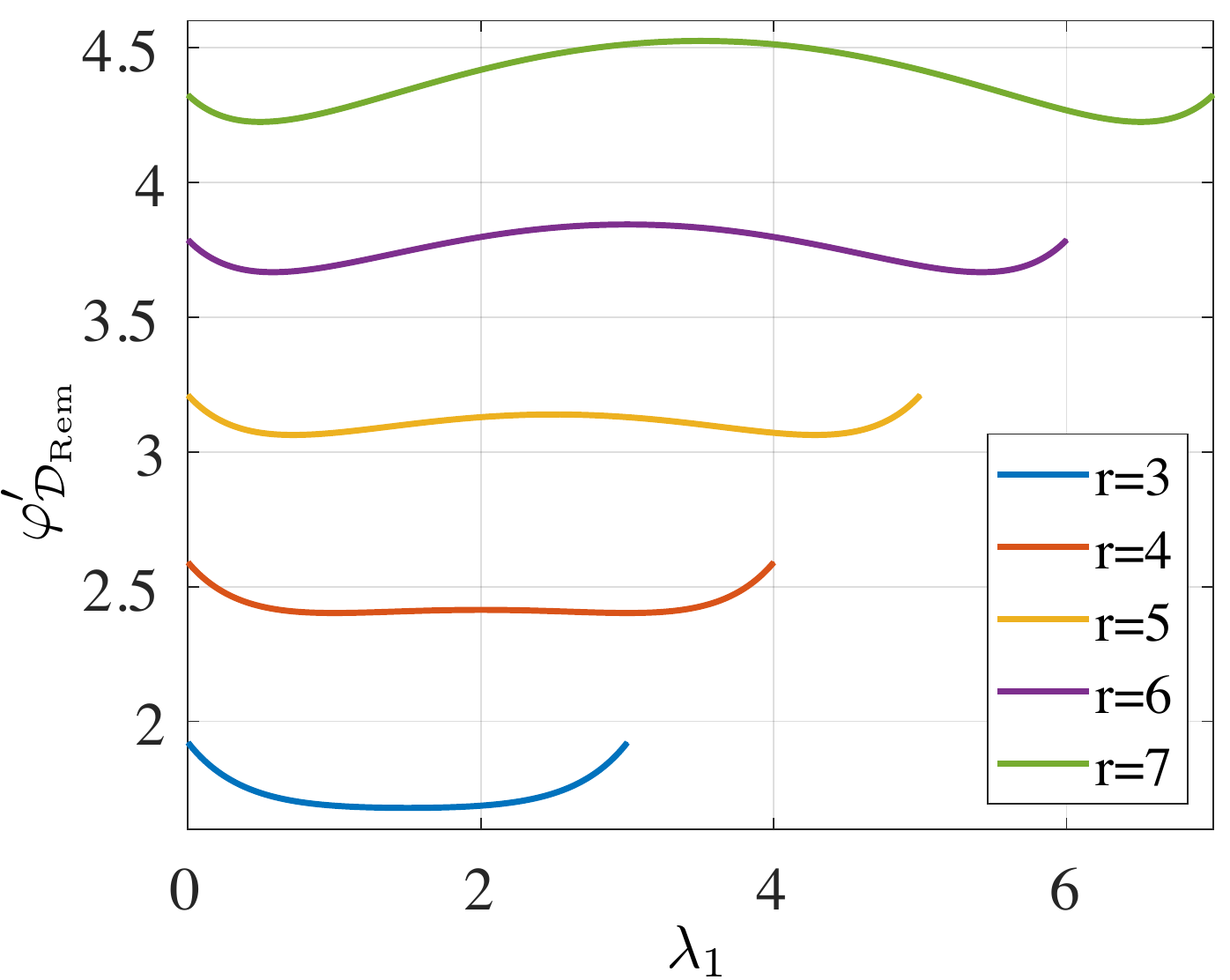}}
	\subfigure[Kullback-Leibler Divergence]{\includegraphics[width=0.32\textwidth]{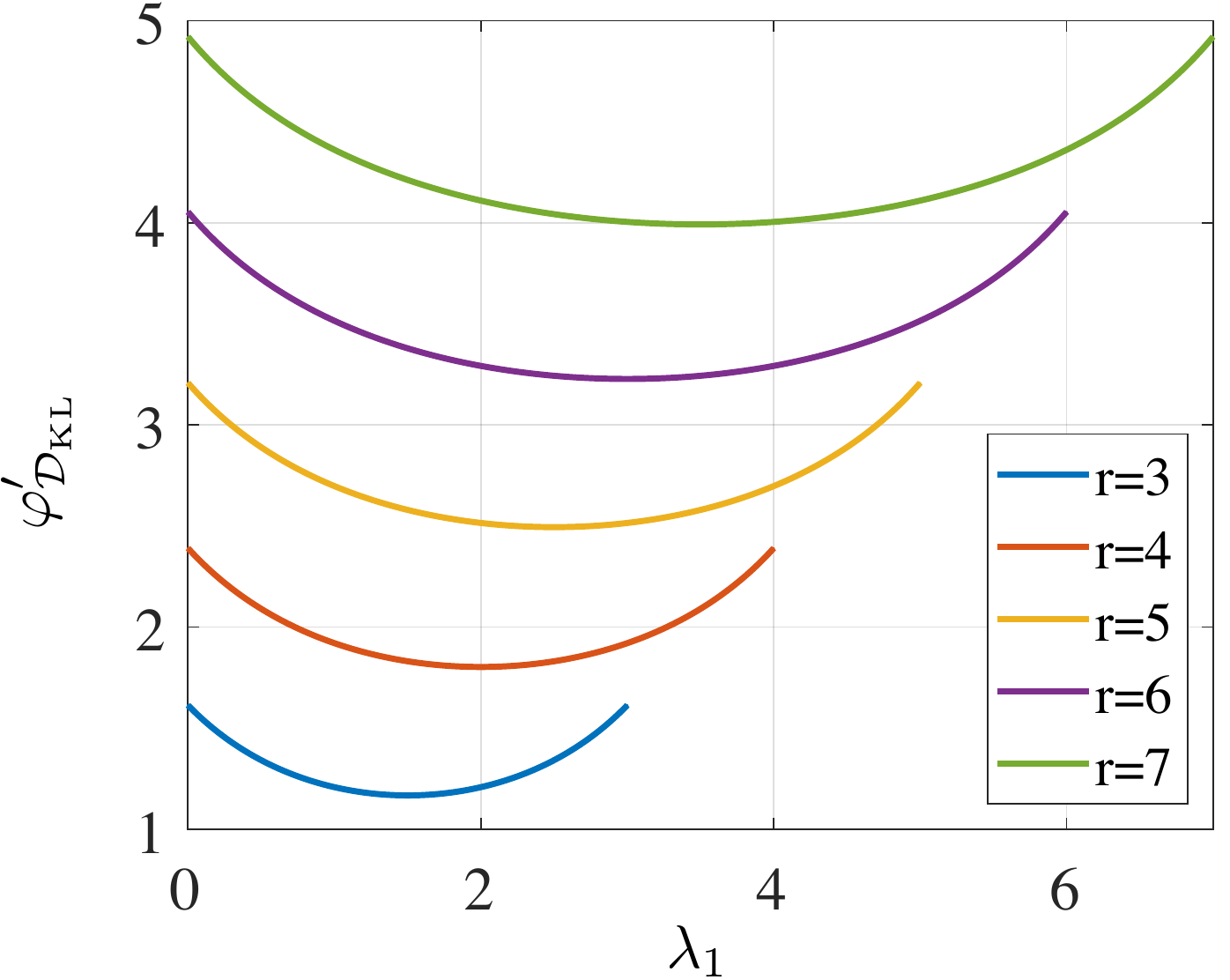}}
%	\subfigure[Jensen-Shannon Divergence]{\includegraphics[width=0.3\textwidth]{js_curve.pdf}}
	\subfigure[log-determinant Divergence]{\includegraphics[width=0.32\textwidth]{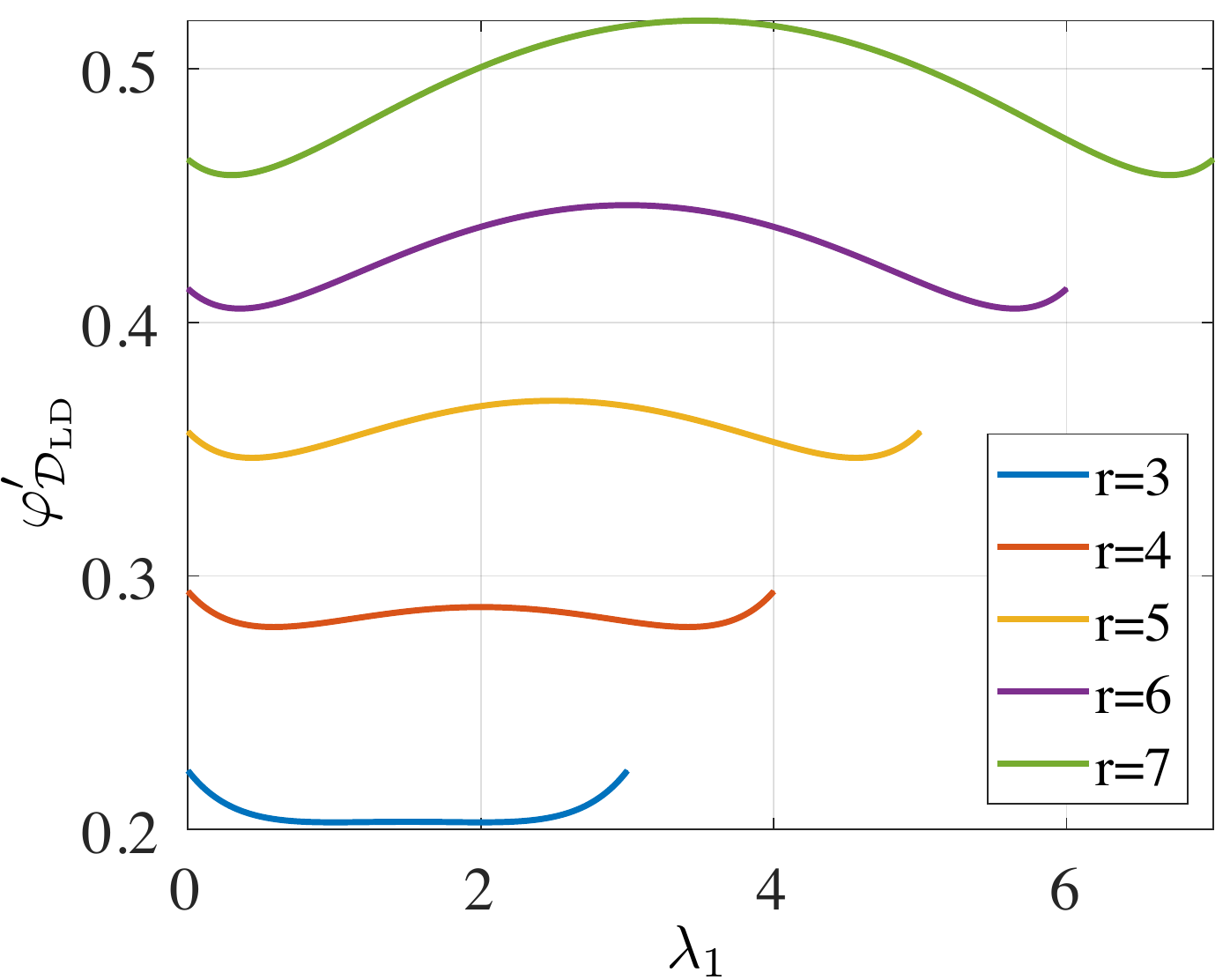}}
	\caption{The curve of induced potential functions with respect to $\lambda^*_1$, where other eigenvalues $\lambda^*_2,\dots,\lambda^*_m$ are fixed. The $r$ in the legends is $r=m\sigma^2-\sum^{m-1}_{k=2}\lambda^*_k=\lambda^*_0+\lambda^*_1$.}
	\label{fig:curve_pf}
\end{figure*}
The curves of induced potential functions (\ref{eq:fun_rem}-\ref{eq:fun_ld}) with respect to $\lambda^*_1$ (suppose $\lambda^*_2,\dots,\lambda^*_{m-1}$ is fixed) are shown in Fig.\ref{fig:curve_pf}. 
%In these figures, the maximum is reached when $\lambda^*_1$ is in the two ends or the medium. 
For $\mathcal{D}_{\text{RD}}$ and $\mathcal{D}_{\text{LD}}$, the ratio of the medium to the two ends is increasing when $r$ ($r=\lambda^*_0+\lambda^*_1$) is ascending. The maximal point is in the medium when $r$ is large, and the maximal point is in the two ends when $r$ is small. Moreover, for $\mathcal{D}_{\text{KL}}$, the maximal point is always in the two ends regardless of the value of $r$, and the KLD reaches the minimum when $\lambda^*_1$ is in the medium.\par
%Actually, according to the Jensen inequality, (\ref{eq:fun_kl}) reaches the minimal, when $\lambda^*_k=\sigma^2\,(k=1,\dots,m-1)$.\par
Moreover, the partial derivate of these functions are as follows,
\begin{equation}
	\frac{\partial \varphi'_{\mathcal{D}_{\text{RD}}}}{\partial \lambda^*_i}=2\left(\frac{\log (1+\lambda^*_i)}{1+\lambda^*_i}-\frac{\log (1+m\sigma^2-\sum\limits^{m-1}_{k=1}\lambda^*_k)}{1+m\sigma^2-\sum\limits^{m-1}_{k=1}\lambda^*_k}\right),
	\label{eq:dev_rem}
\end{equation}
\begin{equation}
	\frac{\partial \varphi'_{\mathcal{D}_{\text{KL}}}}{\partial \lambda^*_i}=-\frac{1}{1+\lambda^*_i}+\frac{1}{1+m\sigma^2-\sum\limits^{m-1}_{k=1}\lambda^*_k},
	\label{eq:dev_kl}
\end{equation}
%\begin{equation}
%	\begin{split}
%		\frac{\partial \varphi'_{\mathcal{D}_{\text{JS}}}}{\partial \lambda^*_i}=&\frac{\lambda^*_i}{2(2+\lambda^*_i)(1+\lambda^*_i)}\\
%	&-\frac{m\sigma^2-\sum\limits^{m-1}_{k=1}\lambda^*_k}{2(2+m\sigma^2-\sum\limits^{m-1}_{k=1}\lambda^*_k)(1+m\sigma^2-\sum\limits^{m-1}_{k=1}\lambda^*_k)},
%	\end{split}
%	\label{eq:dev_js}
%\end{equation}
\begin{equation}
	\begin{split}
		\frac{\partial \varphi'_{\mathcal{D}_{\text{LD}}}}{\partial \lambda^*_i}=&\frac{\lambda^*_i}{2(2+\lambda^*_i)(1+\lambda^*_i)}\\
	&-\frac{m\sigma^2-\sum\limits^{m-1}_{k=1}\lambda^*_k}{2(2+m\sigma^2-\sum\limits^{m-1}_{k=1}\lambda^*_k)(1+m\sigma^2-\sum\limits^{m-1}_{k=1}\lambda^*_k)}.
	\end{split}	
	\label{eq:dev_ld}
\end{equation}
By the partial derivative of the induced potential function, the maximal points and the minimal points can be calculated. In the following contents, we discuss the maximal points of these induced potential function to conclude the factors which benefit the detection performance.\par
%By analyzing the corresponding power spectrum to these points, the characteristic, which directly impacts the detection performance, can be deduced (the maximal points correspond to the advantageous characteristics, and the minimal points correspond to the disadvantageous characteristics). In the following contents, we discuss the maximal points of these induced potential function to conclude the factors which benefit the detection performance.\par
Analyzing the monotonic intervals of (\ref{eq:fun_rem}-\ref{eq:fun_ld}) by their partial derivates (\ref{eq:dev_rem}-\ref{eq:dev_ld}), the maximal value is reached over the following situations.
\begin{itemize}
	\item Riemannian Distance and log-determinant Divergence:
	\begin{equation}
		\lambda^*_i\in\left\{0,m\sigma^2-\sum\limits_{1\le k\le m-1\atop k\ne i}\lambda^*_k,\frac{m\sigma^2}{2}-\sum\limits_{1\le k\le m-1\atop k\ne i}\frac{\lambda^*_k}{2}\right\};
		\label{eq:ep_rem}
	\end{equation}
	\item Kullback-Leibler Divergence:
	\begin{equation}
		\lambda^*_i\in\left\{0,m\sigma^2-\sum\limits_{1\le k\le m-1\atop k\ne i}\lambda^*_k\right\}.
		\label{eq:ep_kl}
	\end{equation}
\end{itemize}
\begin{remark}
	The equation (\ref{eq:ep_rem}) corresponds to the Fig.\ref{fig:curve_pf}, each component of the maximal point is chosen from two ends or the medium, and which situation is maximal depends on the value of $m\sigma^2$, i.e., SCR. The equation (\ref{eq:ep_kl}) also reveals same conclusions as Fig.\ref{fig:curve_pf}, that the maximal points are in the two ends.
\end{remark}\par
By using the adjustment method to the (\ref{eq:ep_rem}) and (\ref{eq:ep_kl}), the maximal values of these considered induced potential functions are from the following power spectrum, respectively.
\begin{itemize}
	\item Riemannian Distance and log-determinant Divergence:
	\begin{equation}
		P^*_k=\bigg(0,\dots,0,\underbrace{\frac{m\sigma^2}{k},\dots,\frac{m\sigma^2}{k}}_{k}\bigg)\quad(k=1,\dots,m),
		\label{eq:mp_rem}
	\end{equation}
	$P^*_k$ is the candidate of the maximal power spectrum, i.e., $P^*_{\max}\in\{P^*_1,\dots,P^*_m\}$, and which one is maximal depends on the value of $m\sigma^2$.
	\item Kullback-Leibler Divergence:
	\begin{equation}
		P^*_{\max}=\left(0,\dots,0,m\sigma^2\right).
		\label{eq:mp_kl}
	\end{equation}
\end{itemize}
%x(\ref{eq:mp_rem}) and (\ref{eq:mp_kl}) are the solutions when all eigenvalues satisfy (\ref{eq:ep_rem}) and (\ref{eq:ep_kl}), respectively. 
Equation (\ref{eq:mp_kl}) represents the power spectrum with only one nonzero component. It means that the narrower spectrum encourages larger induced potential function when the geometric measure is KL divergence. Moreover, equation (\ref{eq:mp_rem}) represents the flat spectrums with different bandwidths, and which bandwidth is maximal depends on the SCR ($m\sigma^2$). So, the maximal test statistic is selected from the flat spectrum when the geometric measure is one of other twos. Besides, according to the curves in Fig.\ref{fig:curve_pf}, the optimal $k$ is incremental by the ascending SCR.
\begin{figure*}[thp]
	\centering
	\subfigure[Riemannian distance]{\includegraphics[width=0.32\textwidth]{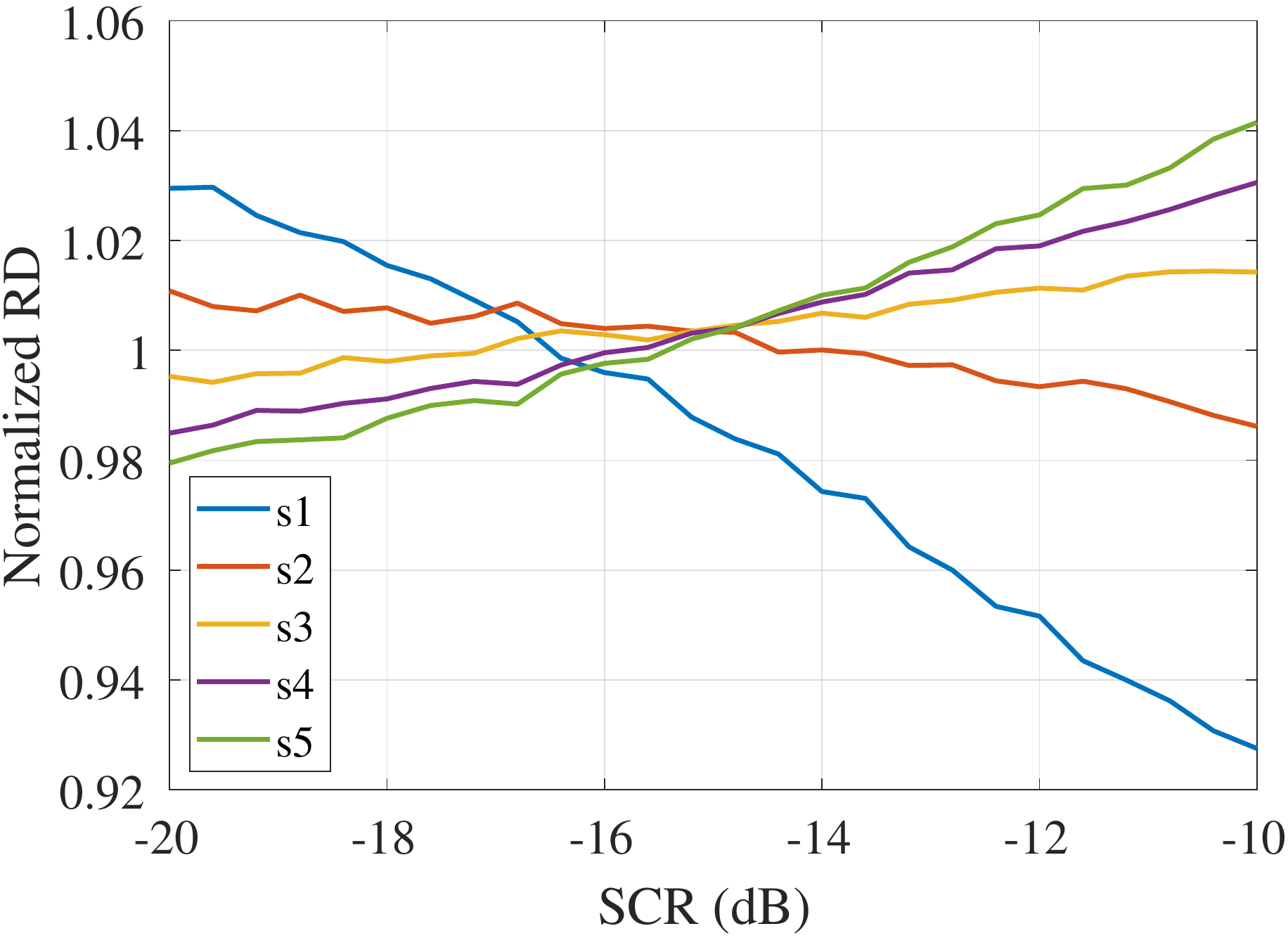}}
	\subfigure[Kullback-Leibler divergence]{\includegraphics[width=0.32\textwidth]{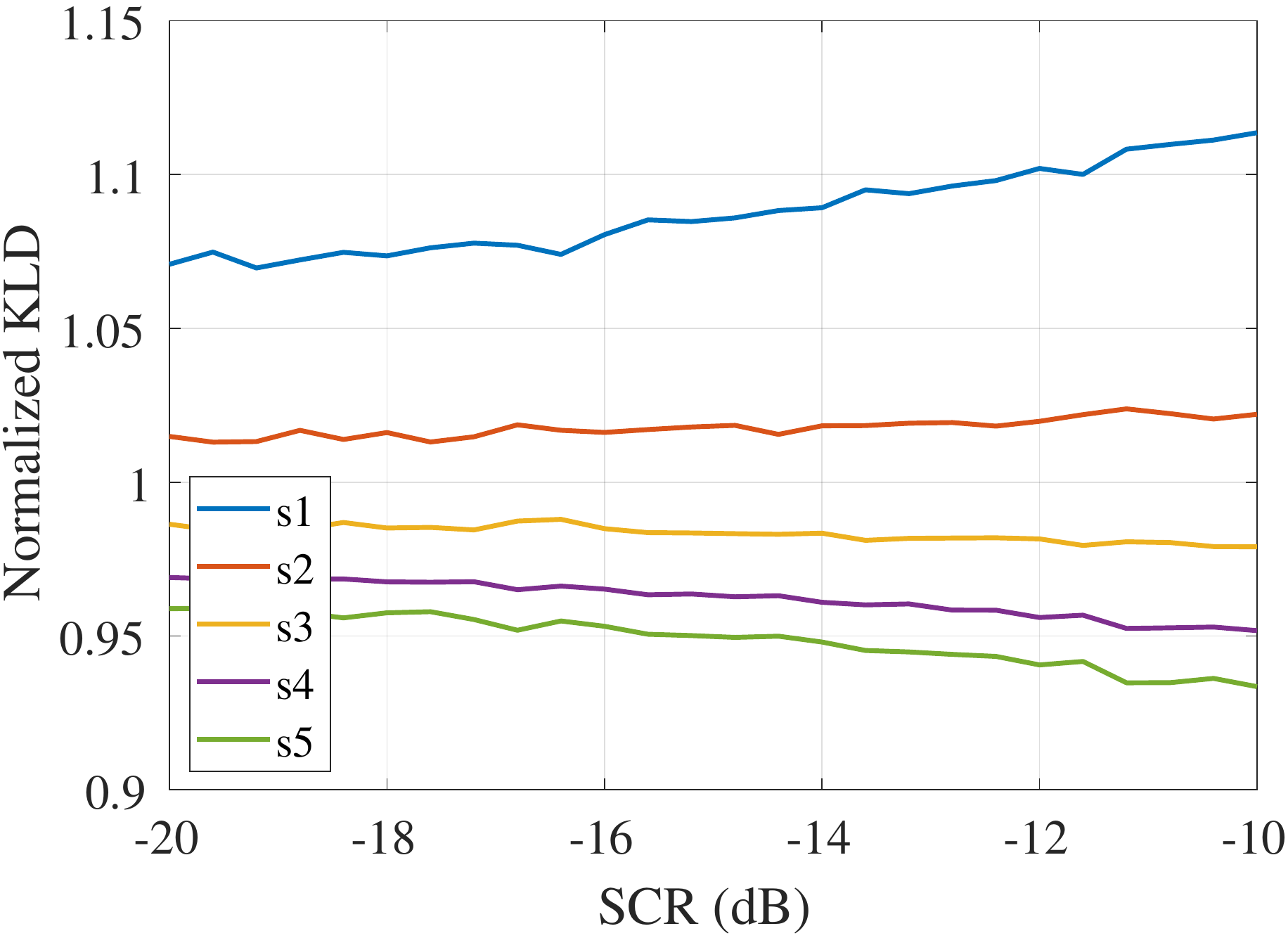}}
%	\subfigure[Jensen-Shannon divergence]{\includegraphics[width=0.4\textwidth]{JSD_TS.pdf}}\qquad
	\subfigure[log-determinant divergence]{\includegraphics[width=0.32\textwidth]{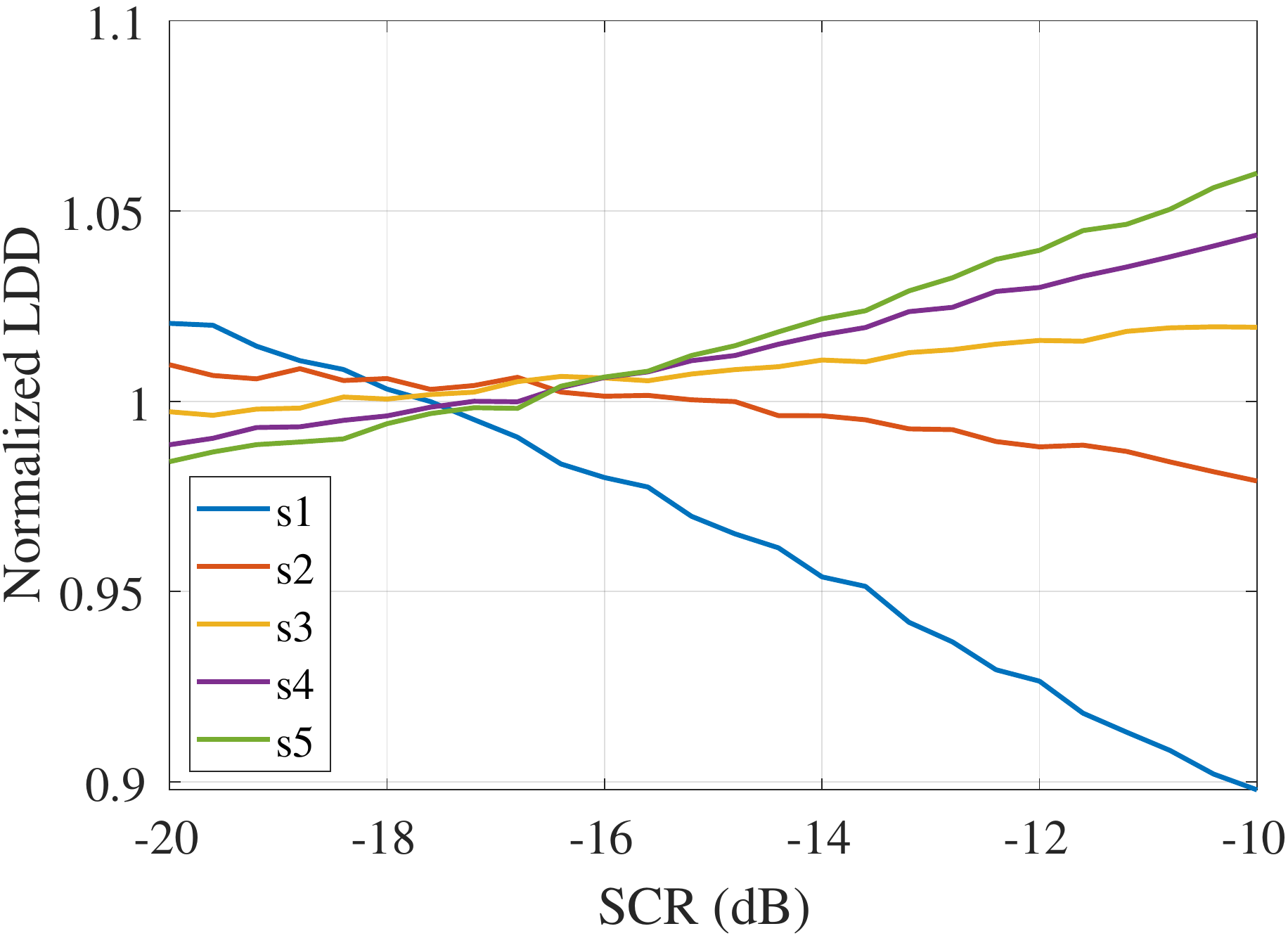}}
	\caption{The normalized geometric measures versus SCR with $\bm{s}_1,\dots,\bm{s}_5$. For each SCR, the geometric measure of $\bm{s}_k$ is normalized by the mean of $\bm{s}_1,\dots,\bm{s}_5$ to clearly show the relative value of them. Moreover, each normalized geometric measure is the average of $10^5$ Monte Carlo runs.}
	\label{fig:curve_ts}
\end{figure*}
\begin{figure*}[htpb]
	\centering
	\subfigure[Riemannian distance]{\includegraphics[width=0.32\textwidth]{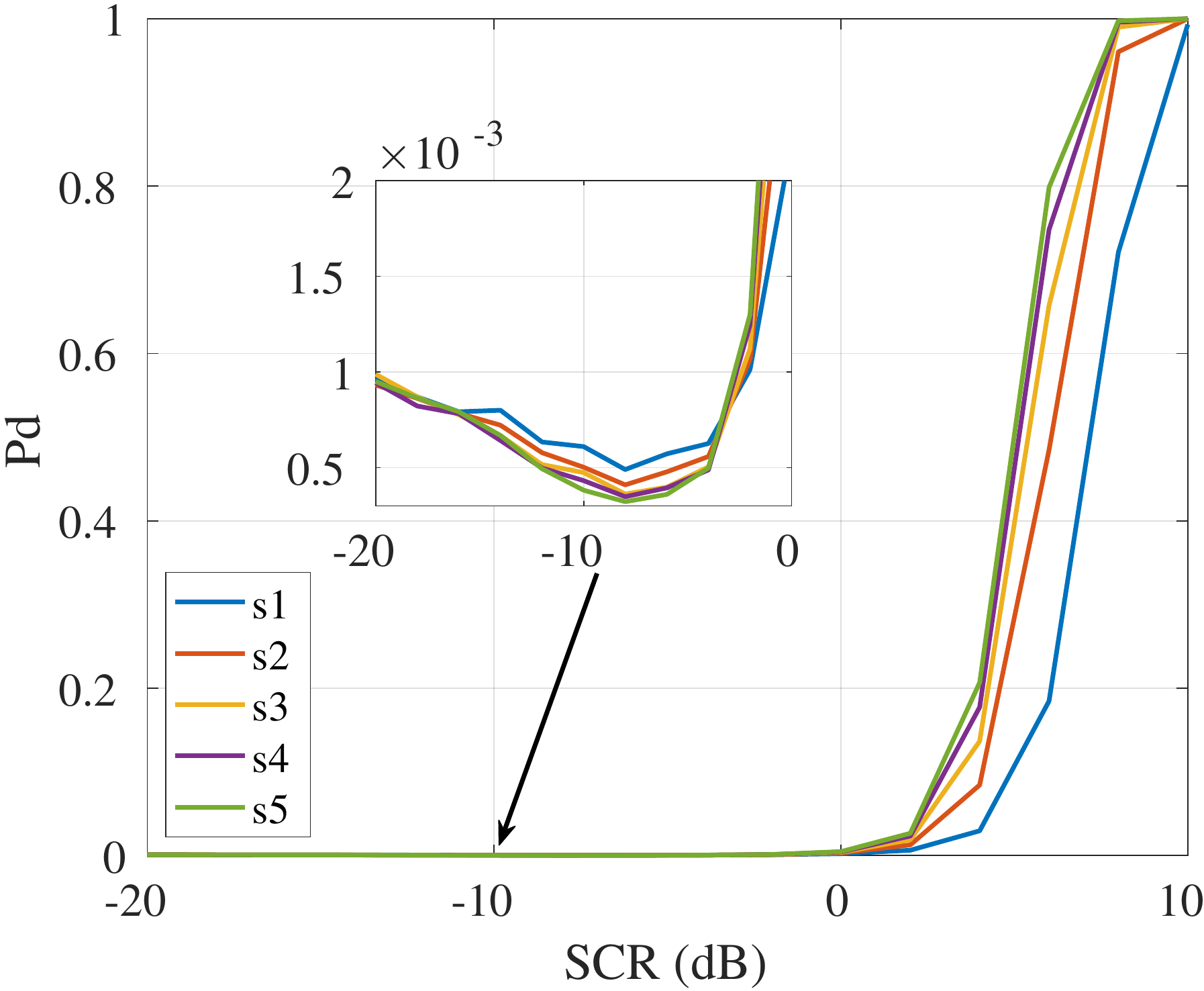}}
	\subfigure[Kullback-Leibler divergence]{\includegraphics[width=0.32\textwidth]{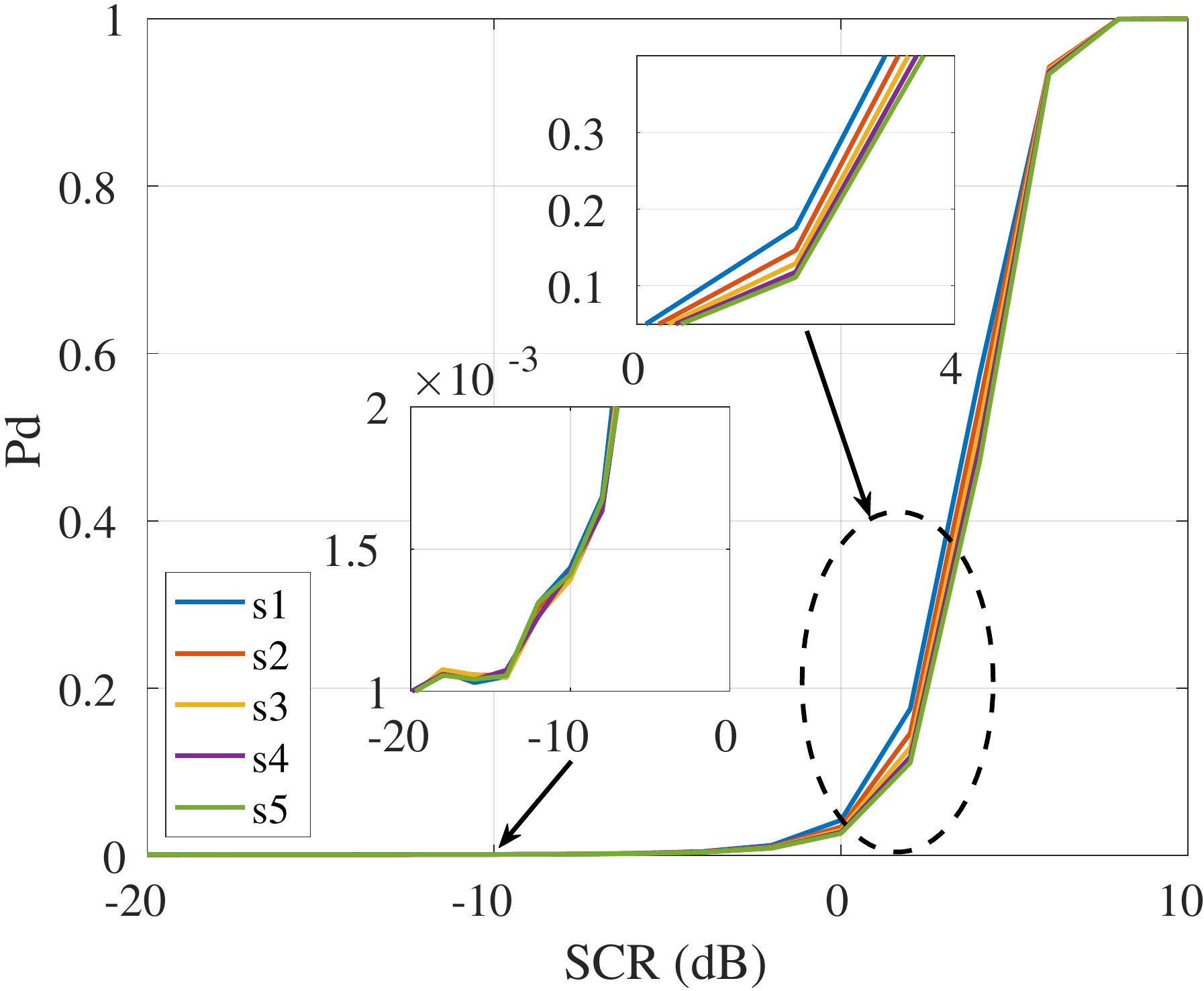}}
%	\subfigure[Jensen-Shannon divergence]{\includegraphics[width=0.4\textwidth]{Pd_JSD.pdf}}\qquad
	\subfigure[log-determinant divergence]{\includegraphics[width=0.32\textwidth]{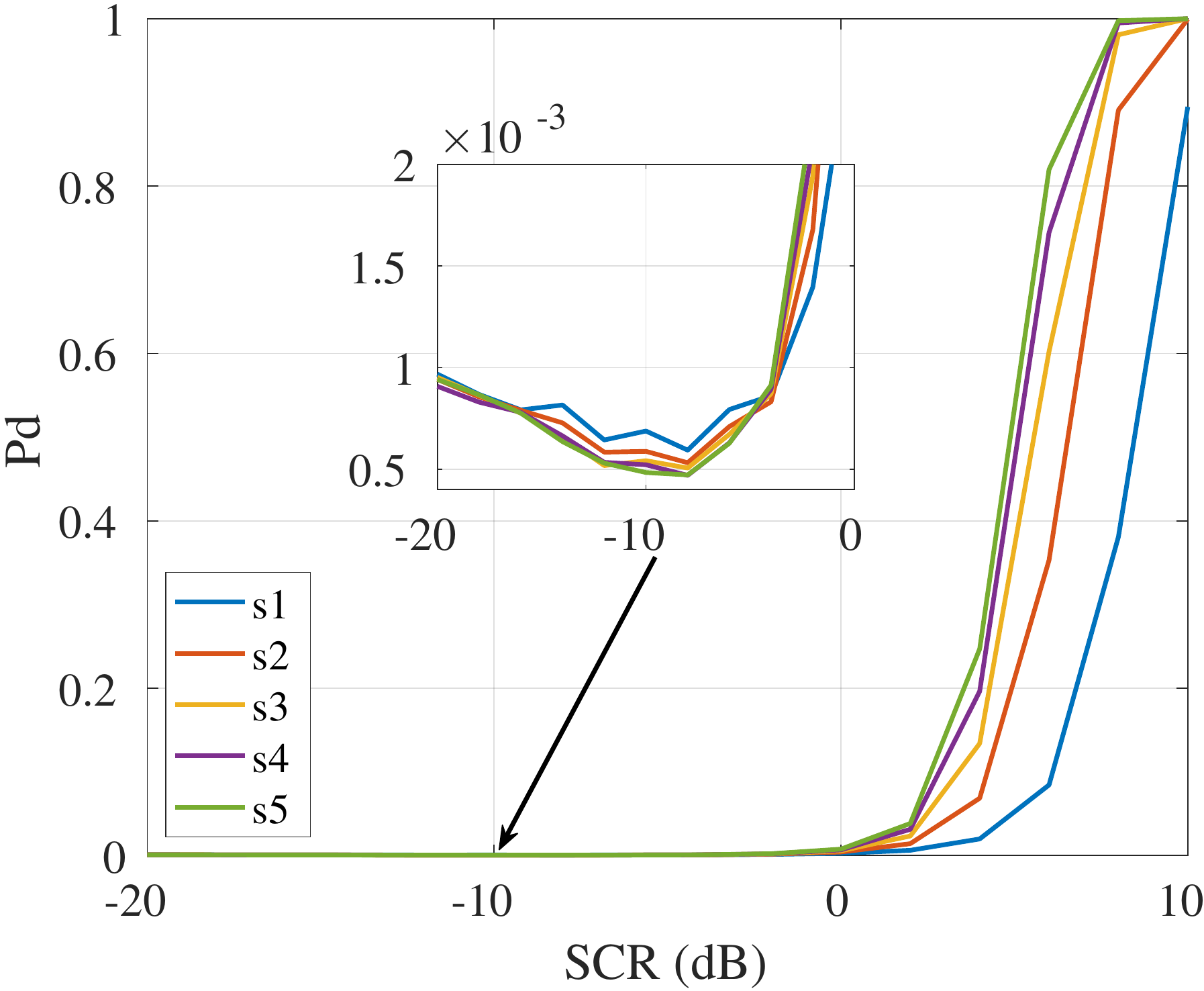}}
	\caption{The detection performance of $\bm{s}_k\,(k=1,\dots,5)$, under SCR from -20 dB to 10 dB. In this figure, the detection probability is calculated based on $5\times10^5$ Monte Carlo runs. And, the curves of the detection probability under low SCR (-20 dB to 0 dB) are attached in the center of figures.}
	\label{fig:curve_pd}
\end{figure*}
\subsubsection*{Numerical Example}
The above conclusions are verified by the following sample. 
\begin{example}
	The majority of settings are as same as example~\ref{exp:exp}, which are shown in Table~\ref{Tab:para}.
%\begin{table}[hbt]
%  \caption{Settings of the Simulation.}
%  \label{Tab:para}
%  \centering
%  \begin{tabular}{|l|c|}
%  	\hline
%   	Parameters or Variables & Setting\\
%   	\hline
%   	\hline
%    Number of range cells & 17 \\
%    Number of pulses & 15\\
%    False-alarm Probability & $P_f=10^{-3}$\\
%    Target cell & $9^{\text{th}}$\\
%    Pulse repetition frequency & $f_r=1000$ Hz\\
%    \hline
%    \multirow{2}{*}{K-distribution clutter} & Shape parameter 1\\
%     & Scalar parameter 0.5\\
%    \hline
%  \end{tabular}
%\end{table}
Differently, there are five sorts of target echo steering vectors,
\begin{equation}
	\bm{s}_k=A\mathcal{F}^{-1}\bigg(0,\dots,0,\underbrace{\frac{1}{\sqrt{k}},\dots,\frac{1}{\sqrt{k}}}_k\bigg),\quad k=1,\dots,5,
\end{equation}
where $A$ is the amplitude of the target echo and $\mathcal{F}^{-1}$ denotes the inverse discrete Fourier transform. For convenience, we define the discrete bandwidth $B$ of $\bm{s}_k$ is $k$, then the discrete bandwidths of these selected target echo are $1,\dots,5$, respectively.
% It is worth mentioning that $\bm{s}_k$ indicates the bandwidth of the discrete doppler spectrum.
%It is worth mentioning that $k$ indicates the bandwidth of the discrete doppler spectrum.
\end{example}
To verify the proposed analytic method, the test statistics of geometric detectors with $\bm{s}_1,\dots,\bm{s}_5$ would be discussed firstly. To coincide with the above analysis, we suppose the reference data is actually equivalent to the component of the clutter in the primary data. For showing the order changing, the test statistics are normalized by the mean of $\bm{s}_1,\dots,\bm{s}_5$ for each SCR. The normalized test statistics of geometric detectors with $\bm{s}_1,\dots,\bm{s}_5$ are presented in Fig.\ref{fig:curve_ts}.
%For clearly showing the order of value among the test statistics of $\bm{s}_k\,(k=1,\dots,5)$ and eliminating the influence by the mean matrix estimation, the normalized geometric measure between the received signal and its own noise component of target cell of the example is presented in Fig.\ref{fig:curve_ts}. 
The results correspond to the conclusion that the discrete bandwidth $k$ of the maximal test statistic depends on the SCR for RD and LDD, and the maximal $k$ equals to 1 for KLD. Moreover, the results prove the discussion that the optimal $k$ of the RD and the LDD ascends when SCR is increasing.\par
Furthermore, the detection performance of this example with the geometry-based detection scheme (as Fig.\ref{fig:CFAR} depicted) is shown in Fig.\ref{fig:curve_pd}. For KLD, the detection probabilities of the five signals $\bm{s}_1,\dots,\bm{s}_5$ are similar. In the magnified figure, we can see that $\bm{s}_1$ has the largest detection probability and is followed by $\bm{s}_2,\bm{s}_3,\bm{s}_4,\bm{s}_5$, successively. For the RD and LDD, the order of the detection probabilities in descending is $\bm{s}_5,\bm{s}_4,\bm{s}_3,\bm{s}_2,\bm{s}_1$ under high SCR, and it is $\bm{s}_1,\bm{s}_2,\bm{s}_3,\bm{s}_4,\bm{s}_5$ under low SCR.\par
\begin{figure*}[htpb]
	\centering
%	\subfigure[Riemannian distance with $\bm{s}_4$]{\includegraphics[width=0.32\textwidth]{RD_4.pdf}}
%	\subfigure[Kullback-Leibler divergence with $\bm{s}_4$]{\includegraphics[width=0.32\textwidth]{KLD_4.pdf}}
%	\subfigure[log-determinant divergence with $\bm{s}_4$]{\includegraphics[width=0.32\textwidth]{LDD_4.pdf}}
%	\subfigure[Riemannian distance with $\bm{s}_7$]{\includegraphics[width=0.32\textwidth]{RD_7.pdf}}
%	\subfigure[Kullback-Leibler divergence with $\bm{s}_7$]{\includegraphics[width=0.32\textwidth]{KLD_7.pdf}}
%	\subfigure[log-determinant divergence with $\bm{s}_7$]{\includegraphics[width=0.32\textwidth]{LDD_7.pdf}}\
	\subfigure[Enhanced Riemannian distance of $\bm{s}_1$]{\includegraphics[width=0.32\textwidth]{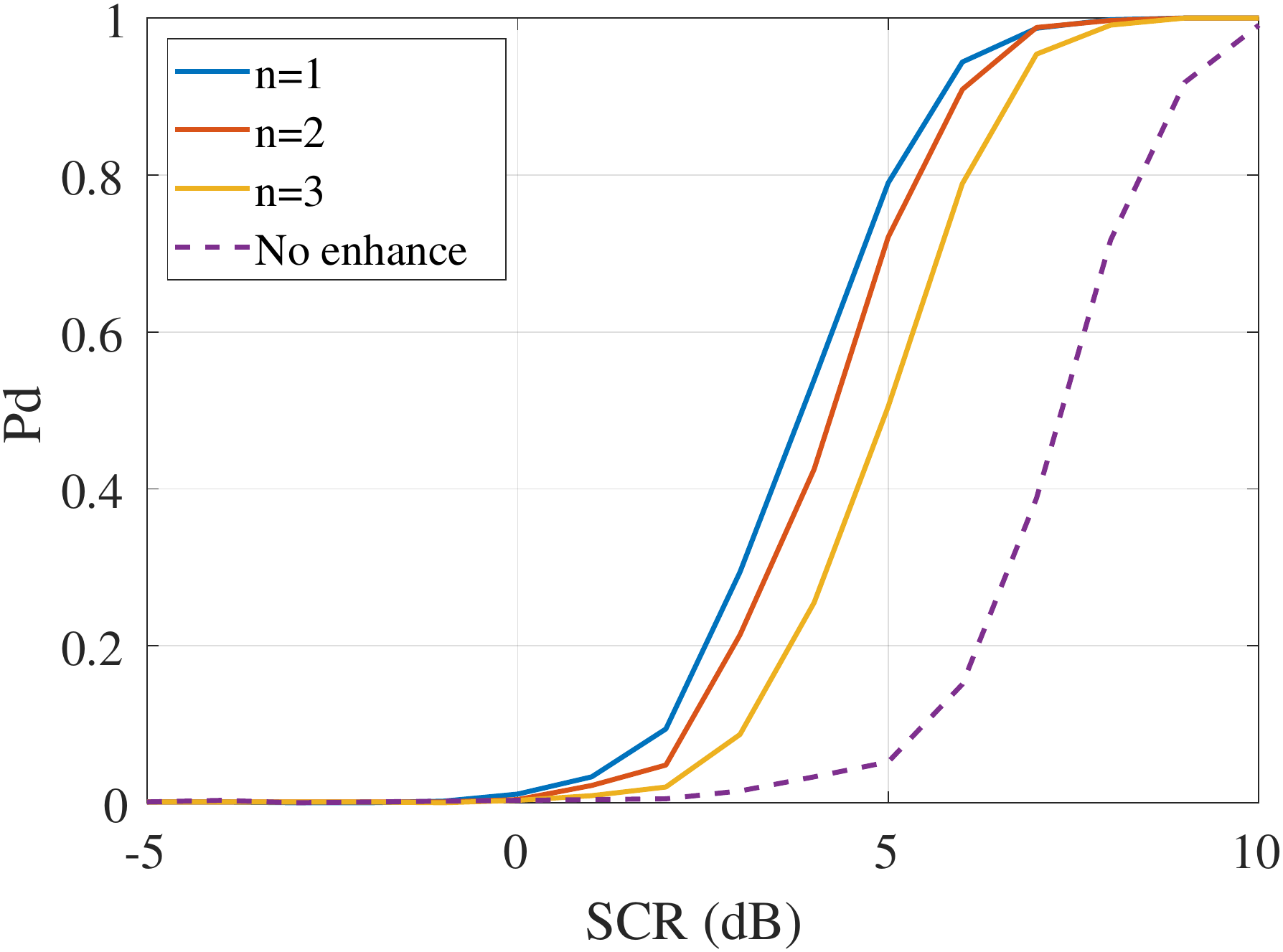}}
	\subfigure[Enhanced Kullback-Leibler divergence of $\bm{s}_1$]{\includegraphics[width=0.32\textwidth]{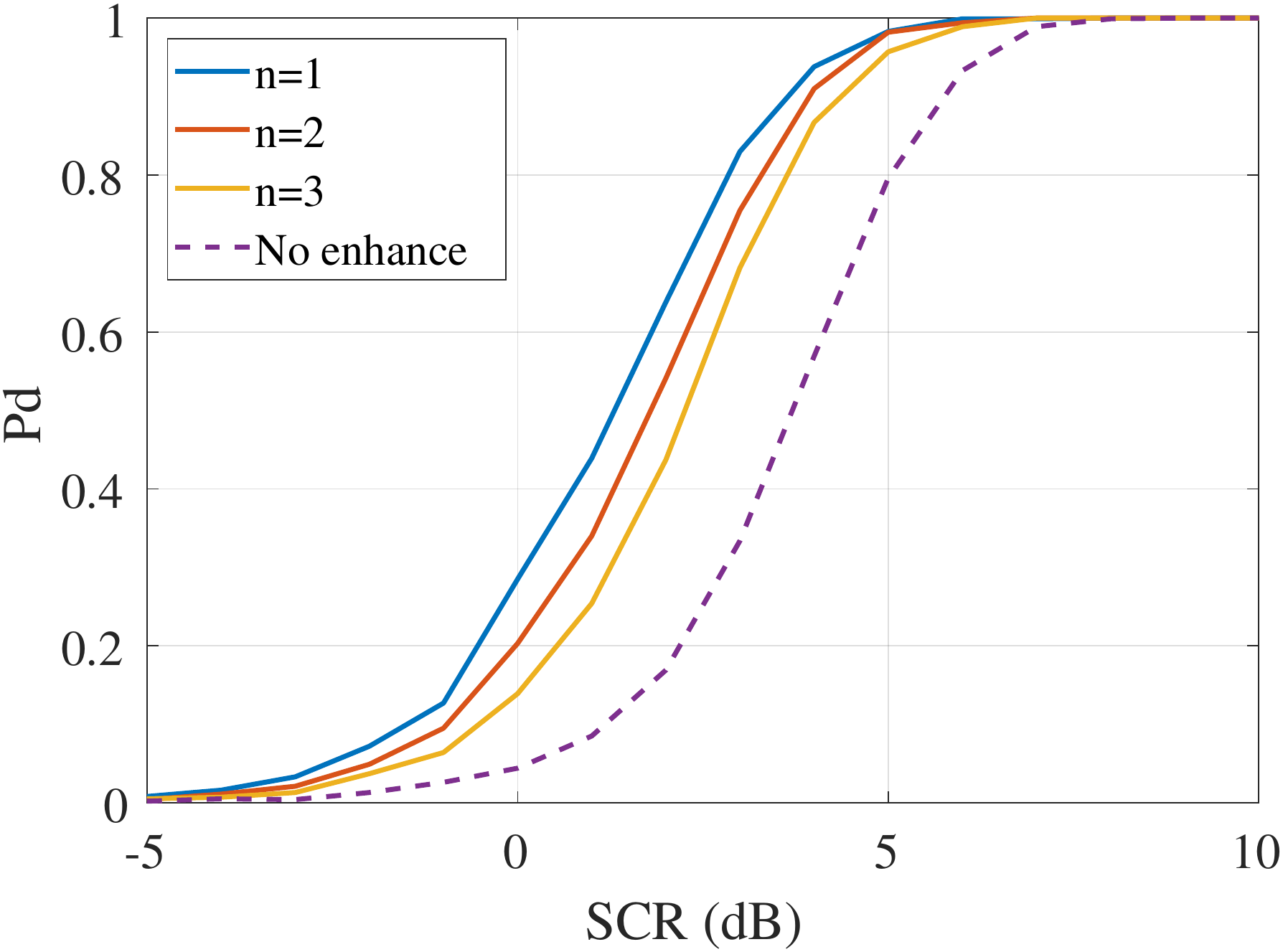}}
	\subfigure[Enhanced log-determinant divergence of $\bm{s}_1$]{\includegraphics[width=0.32\textwidth]{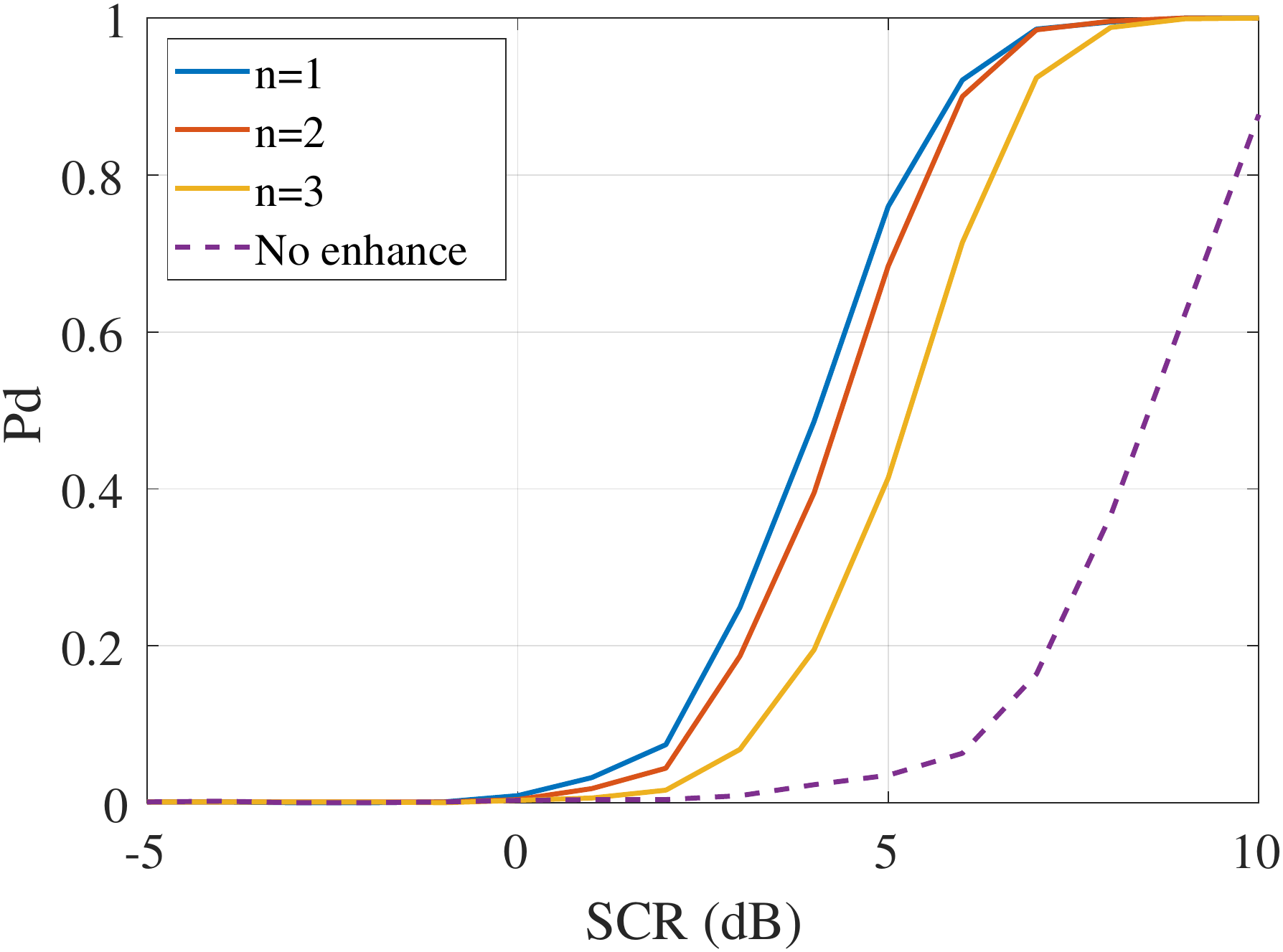}}
	\subfigure[Enhanced Riemannian distance of $\bm{s}_1$]{\includegraphics[width=0.32\textwidth]{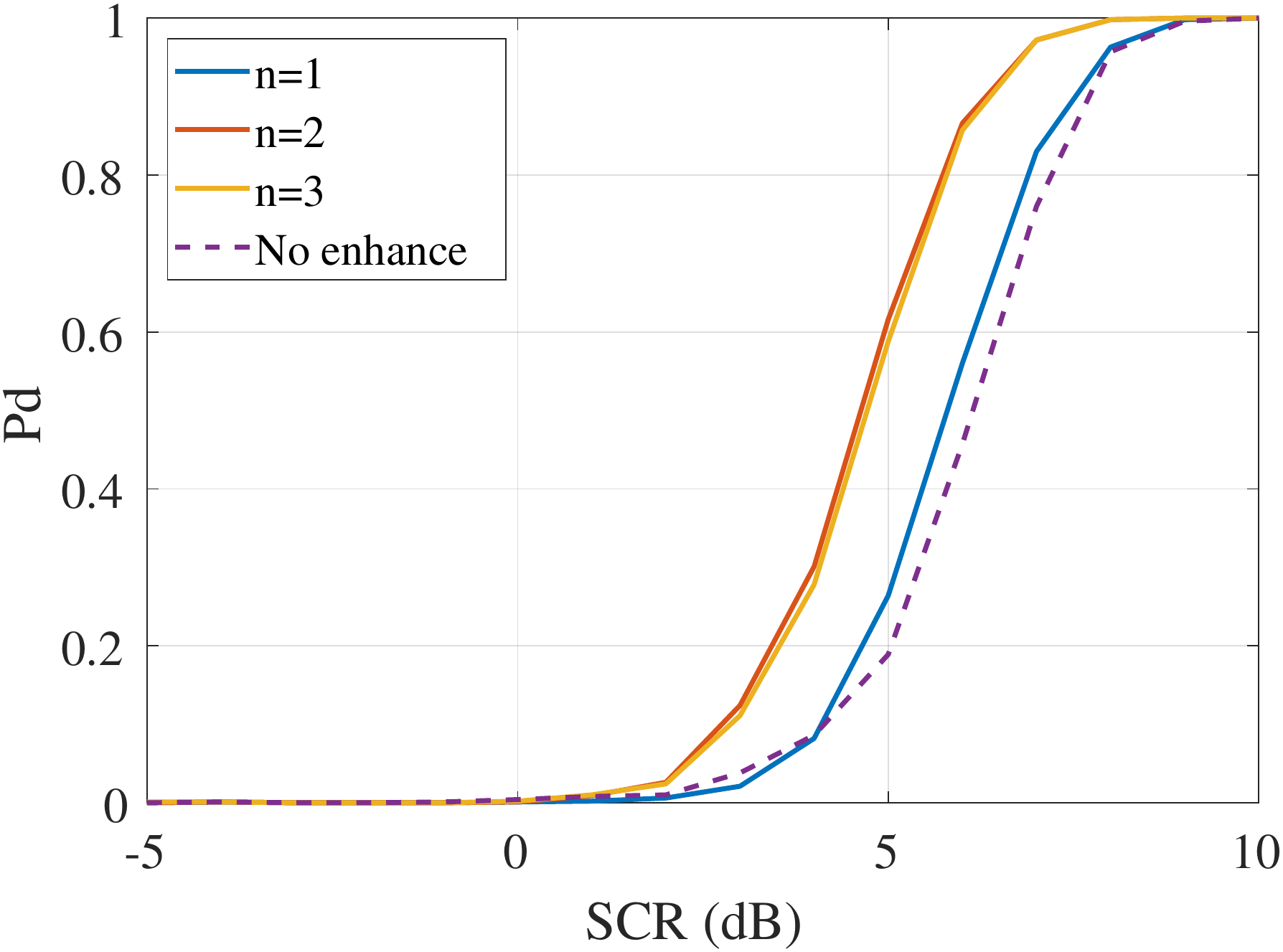}}
	\subfigure[Enhanced Kullback-Leibler divergence of $\bm{s}_2$]{\includegraphics[width=0.32\textwidth]{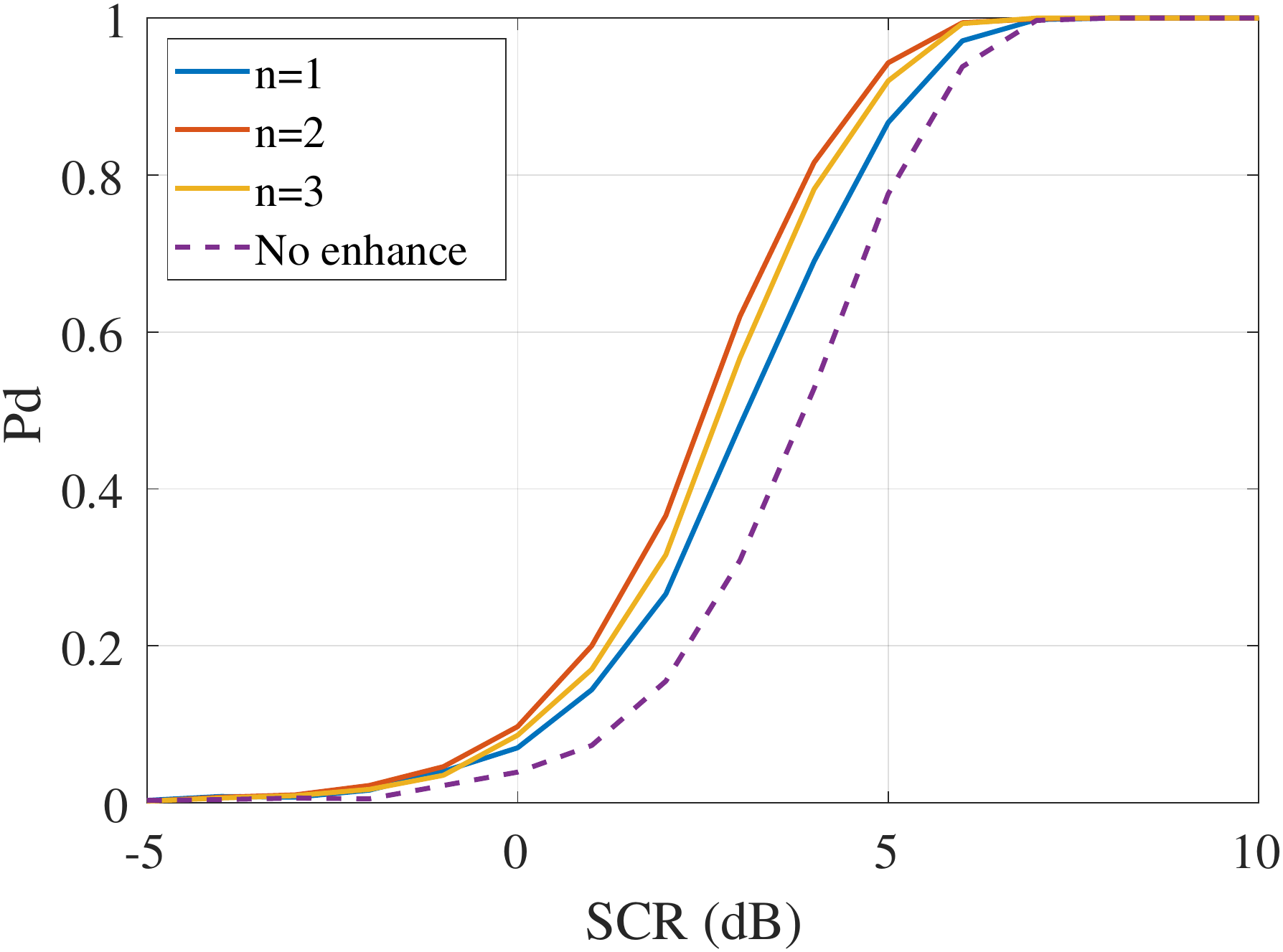}}
	\subfigure[Enhanced log-determinant divergence of $\bm{s}_2$]{\includegraphics[width=0.32\textwidth]{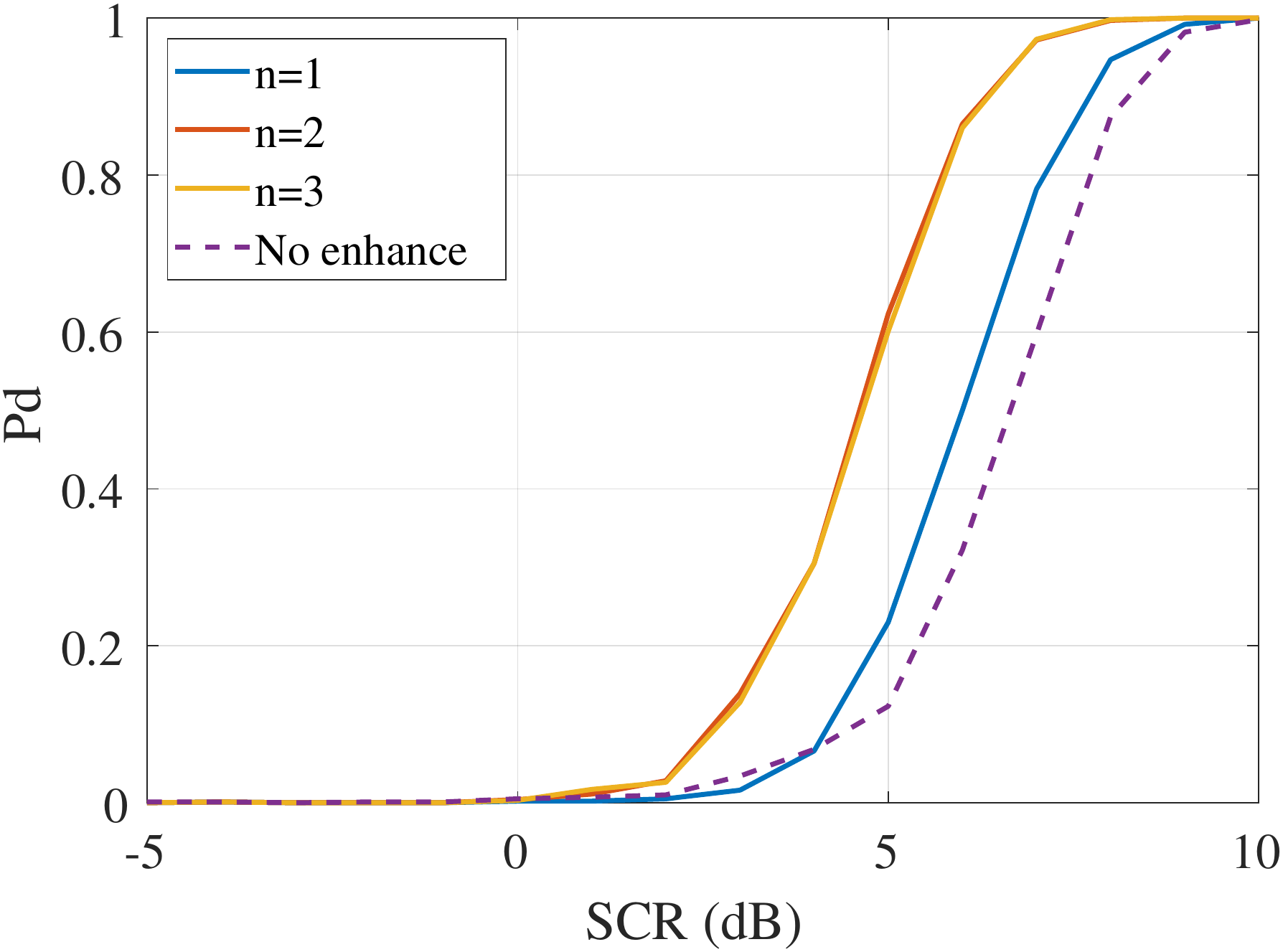}}
	\subfigure[Enhanced Riemannian distance of $\bm{s}_3$]{\includegraphics[width=0.32\textwidth]{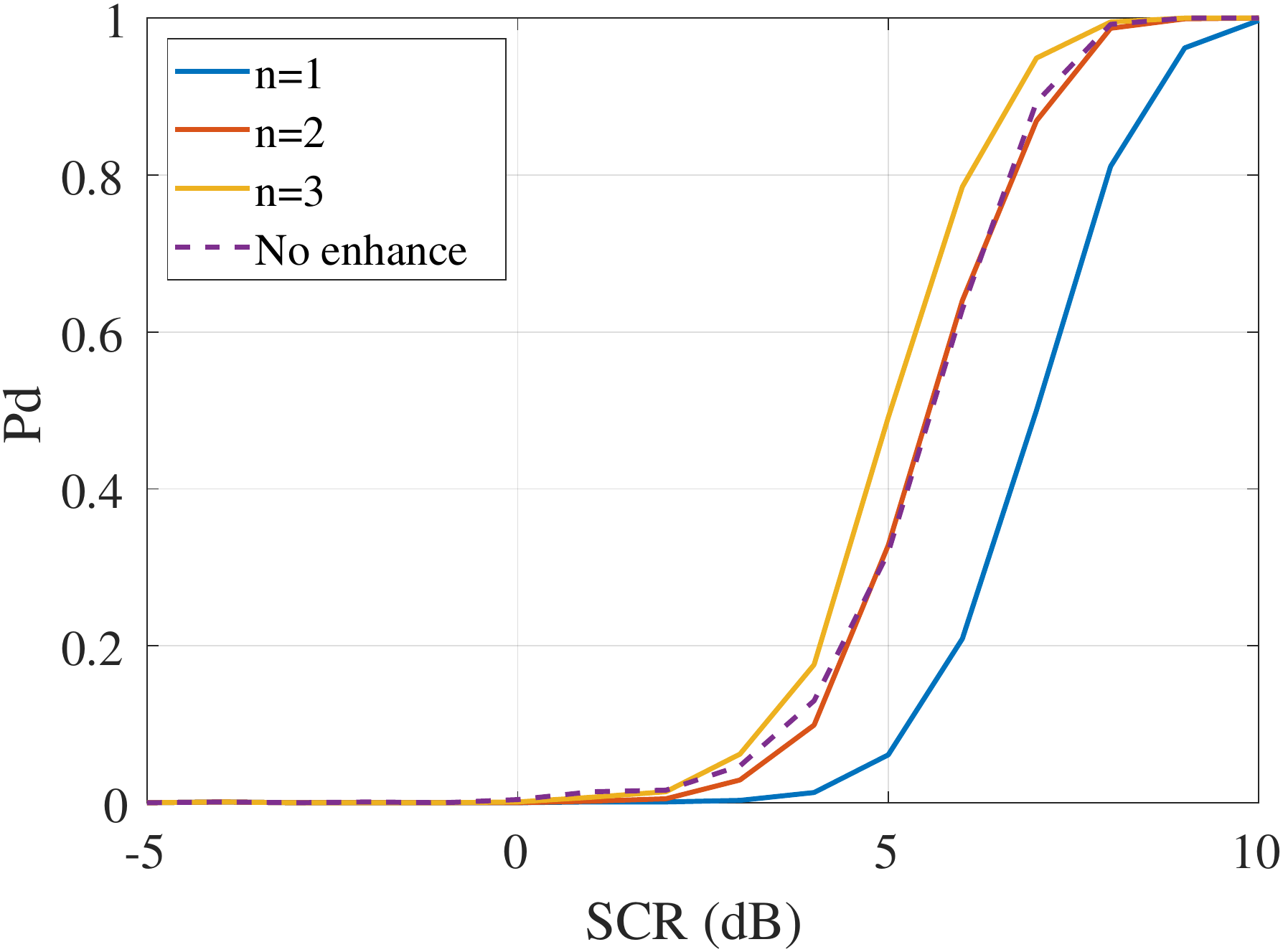}}
	\subfigure[Enhanced Kullback-Leibler divergence of $\bm{s}_3$]{\includegraphics[width=0.32\textwidth]{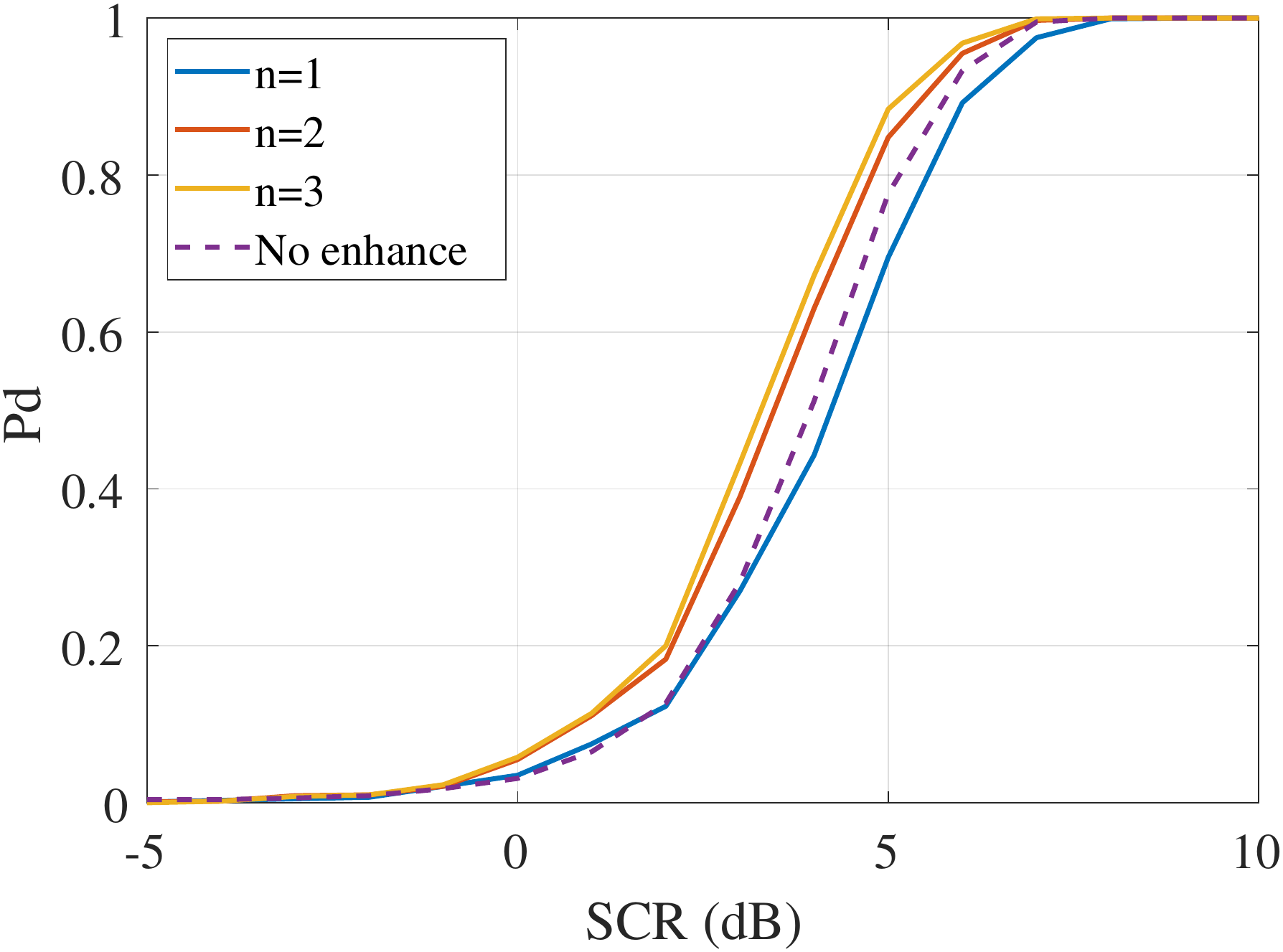}}
	\subfigure[Enhanced log-determinant divergence of $\bm{s}_3$]{\includegraphics[width=0.32\textwidth]{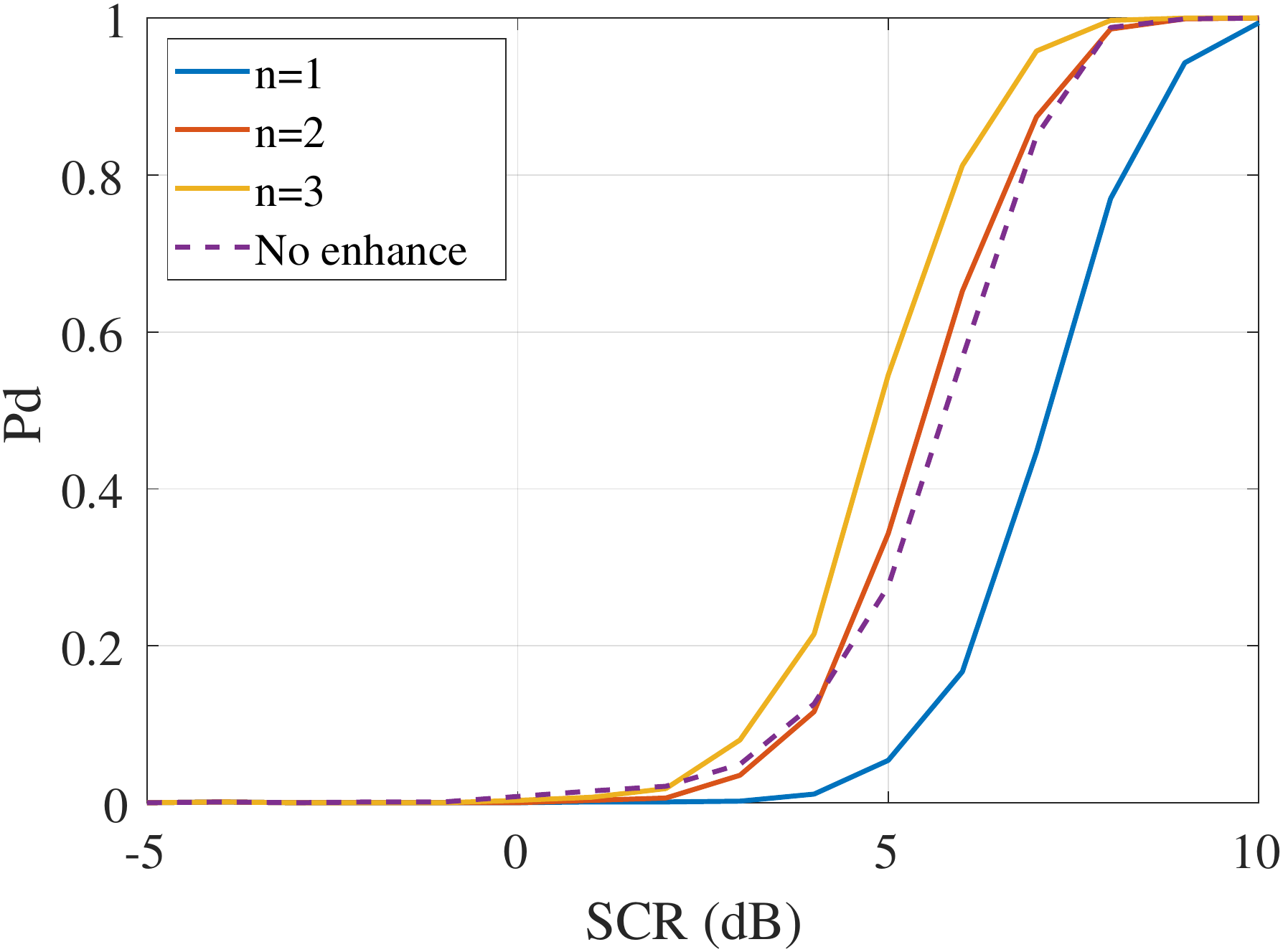}}
	\caption{The detection performance of enhanced detector with SCR from -5 dB to 10 dB. In these figures, the detection probabilities are calculated based on $10^4$ Monte Carlo runs.}
	\label{fig:curve_pd_enhance2}
\end{figure*}
In addition, two problems about the experimental results should be discussed: 1) The cross points (the order changed points) of detection performance curves are inconsistent with the normalized geometric measures in the SCR axis for the RD and LDD; 2) The detection probability descends when SCR increases from -20 dB to -8 dB (KLD is from -18 dB to -16 dB). This two problems both result from the accuracy of the mean matrix estimates. For the first problem, the normalized geometric measures in Fig.\ref{fig:curve_ts} are calculated without the estimate of the mean matrix, but this estimate is considered in the detection. Therefore, under same SCR, the test statistics with the estimate constitute more clutter components than another. So that, the critical points of detection probability are larger than the normalized geometric measures in terms of SCR. As for the second problem, the difference between the estimate and the ground truth of the clutter covariance matrix results in that the eigenvalues in $\mathcal{P}(\mathbf{C}_{\bm{x}},\hat{\mathbf{C}})$ might be lower than 1, i.e., $\lambda^*_k<0\,(1+\lambda^*_k<1)$ might be encountered. But, $\varphi_{\mathcal{D}_{\text{RD}}}$, $\varphi_{\mathcal{D}_{\text{KL}}}$ and $\varphi_{\mathcal{D}_{\text{LDD}}}$ is monotonously decreasing when $\lambda^*_k<0$. So, the test statistics descend when SCR increases from -10 dB to -5 dB, then the detection probability also descends.\par
Totally, the proposed analytic method is effective that the detection performance is consistent with the analysis. By the proposed analytic method, we deduce the maximal points of the induced potential functions, and then the advantageous characteristics in terms of detection performance are concluded by the maximal points.
%Overall, the following conclusions can be drawn from the experiments:
%\begin{itemize}
%	\item The geometric detector based on RD and LDD have the advantages in the detection of wide spectrum signal.
%	\item KLD based geometric detector has the advantages in the narrow spectrum signal and is not good at wide spectrum signal, but the difference is little.
%\end{itemize}
%Actually, in terms of stochastic signal, the wide and narrow power spectrum correspond to the weak and strong temporal correlation, respectively, because of the Wiener-Khinchin theorem and the time-frequency relationship. So, above conclusions also can be described as follows: the KLD is suitable for the detection of signal with strong temporal correlation, and the RD and LDD is suggested to be selected when signal is with weak temporal correlation.
\subsection{Optimal Dimension for Enhanced Matrix CFAR Detection}
As depicted in Fig.\ref{fig:curve_pd_enhance}, the different dimension $n$ of the enhanced mapping results in the distinct detection performance. This part would discuss the relation between the optimal $n$ and the signal from the view of dual power spectrum manifold.\par
Consider the geometric measures $\mathcal{D}_{\text{RD}}$, $\mathcal{D}_{\text{KL}}$ and $\mathcal{D}_{\text{LD}}$. Actually, the induced potential functions of them can be commonly expressed as the following form
\begin{equation}
	\varphi_{\mathcal{D}}(\lambda_1,\dots,\lambda_m)=\sum^m_{k=0} \varphi_{\mathcal{D}}^*(\lambda_k),
	\label{eq:common_pf}
\end{equation}
where $\varphi_{\mathcal{D}}^*(\lambda)=\log^2 \lambda$, $\varphi_{\mathcal{D}}^*(\lambda)=\lambda-1-\log \lambda$ and $\varphi_{\mathcal{D}}^*(\lambda)=\log\frac{\lambda+1}{2\sqrt{\lambda}}$ correspond to the $\varphi_{\mathcal{D}_{\text{RD}}}$, $\varphi_{\mathcal{D}_{\text{KLD}}}$ and $\varphi_{\mathcal{D}_{\text{LD}}}$, respectively. According to the relation between the eigenvalues and the power spectrum, the eigenvalue $\lambda_k$ corresponds to the $k$-th component of the power spectrum. Suppose the discrete bandwidth of the target echo $\bm{s}$ is $B$, which is defined as the bandwidth in the DFT. In the summation (\ref{eq:common_pf}), only the $B$ of the $m$ addends ($\varphi_{\mathcal{D}}^*(\lambda_k),\;k=1,\dots,m$) contain the components of $\bm{s}$, and other $m-B$ are pure clutter. So, the dimension of $\mathcal{M}_{\text{P}}$ is suggested to be reduced to $B$. That means the optimal $n$ equals to the discrete bandwidth of the target echo.\par
This conclusion would be verified by the following sample.
\begin{example}
	The majority of settings are same as example~\ref{exp:exp}, which are shown in Table~\ref{Tab:para}. Differently, there are three target echo steering vectors,
	\begin{equation}
	\bm{s}_k=A\mathcal{F}^{-1}\bigg(0,\dots,0,\underbrace{\frac{1}{\sqrt{k}},\dots,\frac{1}{\sqrt{k}}}_k\bigg),\quad k=1,\dots,3,
\end{equation}
where $A$ is the power of target echo and $\mathcal{F}^{-1}$ denotes the inverse discrete Fourier transform. The discrete bandwidth $B$ of $\bm{s}_1$, $\bm{s}_2$ and $\bm{s}_3$ are 1, 2 and 3, respectively.
%	\begin{equation}
%		\begin{split}
%			&\bm{s}_k=A[\text{exp}(j2\pi\psi_k(1))\;\cdots\;\text{exp}(j2\pi \psi_k(M))],\\
%			&\psi_k(x)=f_df_r^{-1}x+\frac{k-1}{2m^2} x^2,\quad (k=1,4,7)
%		\end{split}
%	\end{equation}
%	\begin{equation}
%		\bm{s}_k=A[\text{exp}(j2\pi\psi_k(1))\;\cdots\;\text{exp}(j2\pi \psi_k(M))],\; (k=1,4,7)
%	\end{equation}
%	\begin{equation}
%		\psi_k(x)=f_df_r^{-1}x+\frac{k-1}{2m^2} x^2,
%	\end{equation}
%	where $A$ is the power of target echo, $f_r=1000 Hz$ is the pulse repetition frequency and $f_d=135 Hz$ is the Doppler frequency in this example. The discrete bandwidth $B$ of $\bm{s}_1$, $\bm{s}_4$ and $\bm{s}_7$ are 1, 4 and 7, respectively.
\end{example}
%The detection performance with $\bm{s}_1$ is depicted in Fig.\ref{fig:curve_pd_enhance}. For $\bm{s}_4$ and $\bm{s}_7$, the detection performance is shown in Fig.\ref{fig:curve_pd_enhance2}. 
The experimental results illustrated in Fig.\ref{fig:curve_pd_enhance2} correspond to the conclusions discussed above, that the optimal $n$ equals to the discrete bandwidth $B$. 
%The exceptions are that there are other dimensions having similar detection performance with $n=B$ for KLD.
%For $\bm{s}_4$, the enhanced KLD with $n=1$ and $n=4$ have similar performance, and the enhanced KLD with $n=4$ and $n=7$ have similar performance for $\bm{s}_7$. Generally, $n=B$ also can be approximately regarded as the optimal for KLD. 
In addition, it is worth pointing that the detection performance of enhanced detector with $n=1$ is worser than the original detector without enhancement when the target echo is $\bm{s}_3$.
Moreover, the following conclusions can be drawn:
\begin{itemize}
	\item The optimal dimension of the enhanced mapping equals to the discrete bandwidth of target echo for enhanced detector.
	\item The closer dimension of the enhanced mapping to the discrete bandwidth encourages the better detection performance.
	\item The improvement of detection performance is obvious, while the discrete bandwidth of the target echo is small, i.e., the spectrum is narrow.
	\item When the bandwidth is large, the enhanced mapping with low dimension would degrade the detection performance of the geometric detector.
\end{itemize}
\section{Conclusions}\label{sec:con}
In this paper, the dual power spectrum manifold $\mathcal{M}_{\text{P}}$, as the dual manifold of Toeplitz Hermitian positive definite (HPD) manifold $\mathcal{M}_{\mathcal{T}H_{++}}$, have been studied. 
%The dual power spectrum manifold has lower dimension and possesses great potentials to deal with the existing challenges in the matrix CFAR detection scheme. 
By the duality of these two manifolds, we have shown that there exists an equivalent induced potential function on $\mathcal{M}_{\text{P}}$ for every affine invariant geometric measure on $\mathcal{M}_{\mathcal{T}H_{++}}$. Based on that, the geometric detectors can be reformulated as the form related to the power spectrum, and two problems have been circumvented: 
1) The enhancement of the geometric detector is reformulated as a simpler problem and is easier to solve than it used to be. 
In the presented examples, 
%the enhancement of the geometric detectors using RD, KLD and LDD are implemented by the proposed method. In addition, 
the analytic expression of the enhanced mapping is provide by the proposed method, that the enhanced mapping was obtained by the gradient descent-based methods. Moreover, the effectiveness of the enhanced detector has been verified by the experimental results. 
2) The performance of geometric detectors have been analyzed by the induced potential functions. 
%For RD, KLD and LDD, we have deduced the relation between the detection performance and the characteristics of signal by the partial derivative of the induced potential functions. 
Besides, the optimal dimension of the enhance mapping also have been deduced by the induced potential function. Experimentally, these deduced conclusions have been validated.

\appendices
\section{The Affine Invariant Properties of the Geometric Measures $\mathcal{D}_{\text{RD}},\mathcal{D}_{\text{KL}},\mathcal{D}_{\text{JS}},\mathcal{D}_{\text{LD}}$}\label{app:aff_invar_mea}
\subsection{Riemannian Distance $\mathcal{D}_{\text{RD}}$}
The Riemannian distance can be transformed to
\begin{equation}
	\begin{split}
		\mathcal{D}_{\text{RD}}(\mathbf{C}_1,\mathbf{C}_2)=&\Arrowvert \log(\mathbf{C}_1^{-\frac{1}{2}}\mathbf{C}_2\mathbf{C}_1^{-\frac{1}{2}})\Arrowvert_F^2\\
		=&\text{tr}\left(\log^2(\mathbf{C}_1^{-\frac{1}{2}}\mathbf{C}_2\mathbf{C}_1^{-\frac{1}{2}})\right)\\
		=&\sum_{k=1}^m\lambda_k\left(\log^2(\mathbf{C}_1^{-\frac{1}{2}}\mathbf{C}_2\mathbf{C}_1^{-\frac{1}{2}})\right)\\
		\overset{\text{(a)}}{=}&\sum_{k=1}^m\log^2\left(\lambda_k(\mathbf{C}_1^{-\frac{1}{2}}\mathbf{C}_2\mathbf{C}_1^{-\frac{1}{2}})\right)\\
		\overset{\text{(b)}}{=}&\sum_{k=1}^m\log^2\left(\lambda_k(\mathbf{C}_2\mathbf{C}_1^{-1})\right),\\
	\end{split}
\end{equation}
where $\lambda_k(\mathbf{A})$ indicates the $k$-th eigenvalue of matrix $\mathbf{A}$, (a) holds by the property of matrix logarithm function $\lambda(\log(\mathbf{C}))=\log(\lambda(\mathbf{C}))$ and (b) holds by the well-known property of matrix eigenvalue that matrices $\mathbf{C}_1\mathbf{C}_2$ and $\mathbf{C}_2\mathbf{C}_1$ have the same eigenvalues.
Therefore,
\begin{equation}
	\begin{split}
		&\mathcal{D}_{\text{RD}}(\mathbf{W}^H\mathbf{C}_1\mathbf{W},\mathbf{W}^H\mathbf{C}_2\mathbf{W})\\
		=&\sum_{k=1}^m\log^2\left(\lambda_k(\mathbf{W}^H\mathbf{C}_2\mathbf{W}\mathbf{W}^{-1}\mathbf{C}_1^{-1}(\mathbf{W}^H)^{-1})\right)\\
		=&\sum_{k=1}^m\log^2\left(\lambda_k(\mathbf{W}^H\mathbf{C}_2\mathbf{C}_1^{-1}(\mathbf{W}^H)^{-1})\right)\\
		=&\sum_{k=1}^m\log^2\left(\lambda_k(\mathbf{C}_2\mathbf{C}_1^{-1})\right)=\mathcal{D}_{\text{RD}}(\mathbf{C}_1,\mathbf{C}_2).
	\end{split}
\end{equation}
\subsection{Kullback-Leibler Divergence $\mathcal{D}_{\text{KL}}$}
\begin{equation}
	\begin{split}
		&\mathcal{D}_{\text{KL}}(\mathbf{W}^H\mathbf{C}_1\mathbf{W},\mathbf{W}^H\mathbf{C}_2\mathbf{W})\\
	=&\text{tr}(\mathbf{W}^H\mathbf{C}_1\mathbf{W}\mathbf{W}^{-1}\mathbf{C}_2^{-1}(\mathbf{W}^H)^{-1}-\mathbf{I})\\
	&-\log\left|\mathbf{W}^H\mathbf{C}_1\mathbf{W}\mathbf{W}^{-1}\mathbf{C}_2^{-1}(\mathbf{W}^H)^{-1}\right|\\
	=&\text{tr}(\mathbf{C}_1\mathbf{C}_2^{-1}(\mathbf{W}^H)^{-1}\mathbf{W}^H-\mathbf{I})\\
	&-\log\left(\left|\mathbf{W}^H\right|\left|\mathbf{C}_1\mathbf{C}_2^{-1}\right|\left|(\mathbf{W}^{H})^{-1}\right|\right)\\
	=&\text{tr}(\mathbf{C}_1\mathbf{C}_2^{-1}-\mathbf{I})-\log\left|\mathbf{C}_1\mathbf{C}_2^{-1}\right|=\mathcal{D}_{\text{KL}}(\mathbf{C}_1,\mathbf{C}_2).\\
	\end{split}
\end{equation}
\subsection{Jensen-Shannon Divergence $\mathcal{D}_{\text{JS}}$}
\begin{equation}
	\begin{split}
		&\mathcal{D}_{\text{JS}}(\mathbf{W}^H\mathbf{C}_1\mathbf{W},\mathbf{W}^H\mathbf{C}_2\mathbf{W})\\
		=&\frac{1}{2}\left[\mathcal{D}_{\text{KL}}\left(\mathbf{W}^H\mathbf{C}_1\mathbf{W},\frac{\mathbf{W}^H\mathbf{C}_1\mathbf{W}+\mathbf{W}^H\mathbf{C}_2\mathbf{W}}{2}\right)\right.\\
		&+\left.\mathcal{D}_{\text{KL}}\left(\mathbf{W}^H\mathbf{C}_2\mathbf{W},\frac{\mathbf{W}^H\mathbf{C}_1\mathbf{W}+\mathbf{W}^H\mathbf{C}_2\mathbf{W}}{2}\right)\right]\\
		=&\frac{1}{2}\left[\mathcal{D}_{\text{KL}}\left(\mathbf{C}_1,\frac{\mathbf{C}_1+\mathbf{C}_2}{2}\right)+\mathcal{D}_{\text{KL}}\left(\mathbf{C}_2,\frac{\mathbf{C}_1+\mathbf{C}_2}{2}\right)\right]\\
		=&\mathcal{D}_{\text{JS}}(\mathbf{C}_1,\mathbf{C}_2).
	\end{split}
\end{equation}
\subsection{log-determinant Divergence $\mathcal{D}_{\text{LD}}$}
\begin{equation}
	\begin{split}
		&\mathcal{D}_{\text{LD}}(\mathbf{W}^H\mathbf{C}_1\mathbf{W},\mathbf{W}^H\mathbf{C}_2\mathbf{W})\\
	=&\log\left|\frac{\mathbf{W}^H\mathbf{C}_1\mathbf{W}+\mathbf{W}^H\mathbf{C}_2\mathbf{W}}{2}\right|-\log\sqrt{\left|\mathbf{W}^H\mathbf{C}_1\mathbf{W}\mathbf{W}^H\mathbf{C}_2\mathbf{W}\right|}\\
	=&\log\frac{\left|\mathbf{W}^H\right|\left|\frac{1}{2}\left(\mathbf{C}_1+\mathbf{C}_2\right)\right|\left|\mathbf{W}\right|}{\sqrt{\left|\mathbf{C}_1\mathbf{C}_2\right|}|\mathbf{W}|\left|\mathbf{W}^H\right|}\\
	=&\log\left|\frac{\mathbf{C}_1+\mathbf{C}_2}{2}\right|-\log\sqrt{\left|\mathbf{C}_1\mathbf{C}_2\right|}=\mathcal{D}_{\text{LD}}(\mathbf{C}_1,\mathbf{C}_2).
	\end{split}
\end{equation}
\section{Induced Potential Function of the Geometric Measures $\mathcal{D}_{\text{RD}},\mathcal{D}_{\text{KL}},\mathcal{D}_{\text{JS}},\mathcal{D}_{\text{LD}}$}\label{app:potential_function}
Let $\mathbf{C}=\mathbf{C}_2^{-\frac{1}{2}}\mathbf{C}_1\mathbf{C}_2^{-\frac{1}{2}}$, the induced potential functions of $\mathcal{D}_{\text{RD}},\mathcal{D}_{\text{KL}},\mathcal{D}_{\text{JS}},\mathcal{D}_{\text{LD}}$ are shown as follows.
\subsection{Riemannian Distance $\mathcal{D}_{\text{RD}}$}
\begin{equation}
	\begin{split}
		&\varphi_{\mathcal{D}_{\text{RD}}}(\mathcal{P}(\mathbf{C}_1,\mathbf{C}_2))=\mathcal{D}_{\text{RD}}(\mathbf{C}_1,\mathbf{C}_2)=\mathcal{D}_{\text{RD}}(\mathbf{C},\mathbf{I})\\
		=&\parallel \log\mathbf{C}^{-1}\parallel^2_F=\text{tr}(\log^2\mathbf{C}^{-1})=\sum^{m-1}_{k=0}\log^2\lambda_k.
	\end{split}
\end{equation}
\subsection{Kullback-Leibler Divergence $\mathcal{D}_{\text{KL}}$}
\begin{equation}
	\begin{split}
		&\varphi_{\mathcal{D}_{\text{KL}}}(\mathcal{P}(\mathbf{C}_1,\mathbf{C}_2))=\mathcal{D}_{\text{KL}}(\mathbf{C}_1,\mathbf{C}_2)=\mathcal{D}_{\text{KL}}(\mathbf{C},\mathbf{I})\\
		=&\text{tr}(\mathbf{C}-\mathbf{I})-\log|\mathbf{C}|=\sum^{m-1}_{k=0} \lambda_k-1-\log\lambda_k.
	\end{split}
\end{equation}\subsection{Jensen-Shannon Divergence $\mathcal{D}_{\text{JS}}$}
\begin{equation}
	\begin{split}
		&\varphi_{\mathcal{D}_{\text{JS}}}(\mathcal{P}(\mathbf{C}_1,\mathbf{C}_2))=\mathcal{D}_{\text{JS}}(\mathbf{C}_1,\mathbf{C}_2)=\mathcal{D}_{\text{JS}}(\mathbf{C},\mathbf{I})\\
		=&\frac{1}{2}\left\{\text{tr}\left[\mathbf{C}\left(\frac{\mathbf{C}+\mathbf{I}}{2}\right)^{-1}-\mathbf{I}\right]-\log\left|\mathbf{C}\left(\frac{\mathbf{C}+\mathbf{I}}{2}\right)^{-1}\right|\right.\\
		&+\left.\text{tr}\left[\left(\frac{\mathbf{C}+\mathbf{I}}{2}\right)^{-1}-\mathbf{I}\right]-\log\left|\left(\frac{\mathbf{C}+\mathbf{I}}{2}\right)^{-1}\right|\right\}\\
		=&\text{tr}\left[\left(\frac{\mathbf{C}+\mathbf{I}}{2}\right)\left(\frac{\mathbf{C}+\mathbf{I}}{2}\right)^{-1}-\mathbf{I}\right]-\frac{1}{2}\log\left|\mathbf{C}\left(\frac{\mathbf{C}+\mathbf{I}}{2}\right)^{-2}\right|\\
		=&\sum^{m-1}_{k=0} \log\frac{\lambda_k+1}{2\sqrt{\lambda_k}}.
	\end{split}
\end{equation}
\subsection{log-determinant Divergence $\mathcal{D}_{\text{LD}}$}
\begin{equation}
	\begin{split}
		&\varphi_{\mathcal{D}_{\text{LD}}}(\mathcal{P}(\mathbf{C}_1,\mathbf{C}_2))=\mathcal{D}_{\text{LD}}(\mathbf{C}_1,\mathbf{C}_2)=\mathcal{D}_{\text{LD}}(\mathbf{C},\mathbf{I})\\
		=&\log\left|\frac{\mathbf{C}+\mathbf{I}}{2}\right|-\log\sqrt{\left|\mathbf{C}\right|}=\sum^{m-1}_{k=0}\log\frac{\lambda_k+1}{2\sqrt{\lambda_k}}.
	\end{split}
\end{equation}

\ifCLASSOPTIONcaptionsoff
  \newpage
\fi
\bibliographystyle{IEEEtranM}
\bibliography{Toeplitz}

\begin{thebibliography}{10}
\providecommand{\url}[1]{#1}
\csname url@samestyle\endcsname
\providecommand{\newblock}{\relax}
\providecommand{\bibinfo}[2]{#2}
\providecommand{\BIBentrySTDinterwordspacing}{\spaceskip=0pt\relax}
\providecommand{\BIBentryALTinterwordstretchfactor}{4}
\providecommand{\BIBentryALTinterwordspacing}{\spaceskip=\fontdimen2\font plus
\BIBentryALTinterwordstretchfactor\fontdimen3\font minus
  \fontdimen4\font\relax}
\providecommand{\BIBforeignlanguage}[2]{{%
\expandafter\ifx\csname l@#1\endcsname\relax
\typeout{** WARNING: IEEEtran.bst: No hyphenation pattern has been}%
\typeout{** loaded for the language `#1'. Using the pattern for}%
\typeout{** the default language instead.}%
\else
\language=\csname l@#1\endcsname
\fi
#2}}
\providecommand{\BIBdecl}{\relax}
\BIBdecl

\bibitem{Zeitouni1992}
O.~{Zeitouni}, J.~{Ziv}, and N.~{Merhav}, ``When is the generalized likelihood
  ratio test optimal?'' \emph{IEEE Transactions on Information Theory},
  vol.~38, no.~5, pp. 1597--1602, Sep. 1992.

\bibitem{Stoica2004}
P.~{Stoica}, {Jianhua Liu}, and {Jian Li}, ``Maximum-likelihood double
  differential detection clarified,'' \emph{IEEE Transactions on Information
  Theory}, vol.~50, no.~3, pp. 572--576, 2004.

\bibitem{Chen2009}
H.~{Chen}, B.~{Chen}, and P.~K. {Varshney}, ``Further results on the optimality
  of the likelihood-ratio test for local sensor decision rules in the presence
  of nonideal channels,'' \emph{IEEE Transactions on Information Theory},
  vol.~55, no.~2, pp. 828--832, 2009.

\bibitem{Moustakides2012}
G.~V. {Moustakides}, G.~H. {Jajamovich}, A.~{Tajer}, and X.~{Wang}, ``Joint
  detection and estimation: Optimum tests and applications,'' \emph{IEEE
  Transactions on Information Theory}, vol.~58, no.~7, pp. 4215--4229, 2012.

\bibitem{Pascal2008}
F.~{Pascal}, P.~{Forster}, J.~{Ovarlez}, and P.~{Larzabal}, ``Performance
  analysis of covariance matrix estimates in impulsive noise,'' \emph{IEEE
  Transactions on Signal Processing}, vol.~56, no.~6, pp. 2206--2217, 2008.

\bibitem{Pascal2010}
C.~Y. {Chong}, F.~{Pascal}, J.~{Ovarlez}, and M.~{Lesturgie}, ``Mimo radar
  detection in non-gaussian and heterogeneous clutter,'' \emph{IEEE Journal of
  Selected Topics in Signal Processing}, vol.~4, no.~1, pp. 115--126, 2010.

\bibitem{Pascal2011}
G.~{Pailloux}, P.~{Forster}, J.~{Ovarlez}, and F.~{Pascal}, ``Persymmetric
  adaptive radar detectors,'' \emph{IEEE Transactions on Aerospace and
  Electronic Systems}, vol.~47, no.~4, pp. 2376--2390, 2011.

\bibitem{Rong2021}
Y.~{Rong}, A.~{Aubry}, A.~{De Maio}, and M.~{Tang}, ``Adaptive radar detection
  in low-rank heterogeneous clutter via invariance theory,'' \emph{IEEE
  Transactions on Signal Processing}, vol.~69, pp. 1492--1506, 2021.

\bibitem{Amari2016}
S.-i. Amari, \emph{Information Geometry and Its Applications}, 1st~ed.\hskip
  1em plus 0.5em minus 0.4em\relax Springer Publishing Company, Incorporated,
  2016.

\bibitem{Arnaudon2013Riemannian}
M.~Arnaudon, F.~Barbaresco, and L.~Yang, ``Riemannian medians and means with
  applications to radar signal processing,'' \emph{IEEE Journal of Selected
  Topics in Signal Processing}, vol.~7, no.~4, pp. 595--604, 2013.

\bibitem{zhao2019}
W.~{Zhao}, W.~{Liu}, and M.~{Jin}, ``Spectral norm based mean matrix estimation
  and its application to radar target cfar detection,'' \emph{IEEE Transactions
  on Signal Processing}, vol.~67, no.~22, pp. 5746--5760, 2019.

\bibitem{Zhang2019}
J.~Zhang, X.~Zhang, W.~Deng, L.~Ye, and Q.~Yang, ``A geometric barycenter-based
  clutter suppression method for ship detection in hf mixed-mode surface wave
  radar,'' \emph{Remote Sensing}, vol.~11, no.~9, 2019.

\bibitem{HUA2021}
X.~Hua and L.~Peng, ``Mig median detectors with manifold filter,'' \emph{Signal
  Processing}, p. 108176, 2021.

\bibitem{Said2017}
S.~{Said}, L.~{Bombrun}, Y.~{Berthoumieu}, and J.~H. {Manton}, ``Riemannian
  gaussian distributions on the space of symmetric positive definite
  matrices,'' \emph{IEEE Transactions on Information Theory}, vol.~63, no.~4,
  pp. 2153--2170, 2017.

\bibitem{Said2018}
S.~{Said}, H.~{Hajri}, L.~{Bombrun}, and B.~C. {Vemuri}, ``Gaussian
  distributions on riemannian symmetric spaces: Statistical learning with
  structured covariance matrices,'' \emph{IEEE Transactions on Information
  Theory}, vol.~64, no.~2, pp. 752--772, 2018.

\bibitem{Lapuyade2008Radar}
J.~Lapuyade-Lahorgue and F.~Barbaresco, ``Radar detection using siegel distance
  between autoregressive processes, application to {HF and X-band} radar,'' in
  \emph{Radar Conference, 2008. RADAR '08. IEEE}, 2008, pp. 1--6.

\bibitem{Burgestimation}
A.~Decurninge and F.~Barbaresco, ``Robust burg estimation of radar scatter
  matrix for autoregressive structured sirv based on fr\'{e}chet medians,''
  \emph{IET Radar, Sonar \& Navigation}, vol.~11, no.~1, pp. 78--89, 2017.

\bibitem{Barbaresco2010}
Fr{\'e}d{\'e}ric, ``Radar monitoring of a wake vortex: Electromagnetic
  reflection of wake turbulence in clear air,'' \emph{Comptes Rendus Physique},
  vol.~11, no.~1, pp. 54--67, 2010, propagation and remote sensing.

\bibitem{Liu2013}
\BIBentryALTinterwordspacing
Z.~Liu and F.~Barbaresco, \emph{Doppler Information Geometry for Wake
  Turbulence Monitoring}.\hskip 1em plus 0.5em minus 0.4em\relax Berlin,
  Heidelberg: Springer Berlin Heidelberg, 2013, pp. 277--290. [Online].
  Available: \url{https://doi.org/10.1007/978-3-642-30232-9_11}
\BIBentrySTDinterwordspacing

\bibitem{hua2017matrix}
X.~Hua, Y.~Cheng, H.~Wang, Y.~Qin, Y.~Li, and W.~Zhang, ``Matrix {CFAR}
  detectors based on symmetrized {Kullback-Leibler and total Kullback-Leibler
  divergences},'' \emph{Digital Signal Processing}, vol.~69, pp. 106--116,
  2017.

\bibitem{hua2018information2}
X.~Hua, H.~Fan, Y.~Cheng, H.~Wang, and Y.~Qin, ``Information geometry for radar
  target detection with total {Jensen-Bregman} divergence,'' \emph{Entropy},
  vol.~20, no.~4, p. 256, 2018.

\bibitem{Cherian2013}
A.~{Cherian}, S.~{Sra}, A.~{Banerjee}, and N.~{Papanikolopoulos},
  ``Jensen-bregman logdet divergence with application to efficient similarity
  search for covariance matrices,'' \emph{IEEE Transactions on Pattern Analysis
  and Machine Intelligence}, vol.~35, no.~9, pp. 2161--2174, 2013.

\bibitem{hua2018geometric}
X.~Hua, Y.~Cheng, H.~Wang, Y.~Qin, and D.~Chen, ``Geometric target detection
  based on total {Bregman} divergence,'' \emph{Digital Signal Processing},
  vol.~75, pp. 232--241, 2018.

\bibitem{hua2018information}
X.~Hua, Y.~Cheng, H.~Wang, and Y.~Qin, ``Information geometry for covariance
  estimation in heterogeneous clutter with total {Bregman} divergence,''
  \emph{Entropy}, vol.~20, no.~4, p. 258, 2018.

\bibitem{hua2017geometric}
X.~Hua, Y.~Cheng, H.~Wang, Y.~Qin, and Y.~Li, ``Geometric means and medians
  with applications to target detection,'' \emph{Iet Signal Processing},
  vol.~11, no.~6, pp. 711--720, 2017.

\bibitem{Zhao2018}
W.~Zhao, C.~Liu, W.~Liu, and M.~Jin, ``Maximum eigenvalue based target
  detection for the k distributed clutter environment,'' \emph{IET Radar, Sonar
  and Navigation}, vol.~12, 08 2018.

\bibitem{Hua_TBD}
X.~Hua, Y.~Ono, L.~Peng, Y.~Cheng, and H.~Wang, ``Target detection within
  nonhomogeneous clutter via total bregman divergence-based matrix information
  geometry detectors,'' \emph{IEEE Transactions on Signal Processing}, vol.~69,
  pp. 4326--4340, 2021.

\bibitem{Wong2017}
K.~M. Wong, J.-K. Zhang, J.~Liang, and H.~Jiang, ``Mean and median of psd
  matrices on a riemannian manifold: Application to detection of narrow-band
  sonar signals,'' \emph{IEEE Transactions on Signal Processing}, vol.~65,
  no.~24, pp. 6536--6550, 2017.

\bibitem{Barbaresco2019}
F.~Barbaresco, ``Coding statistical characterization of radar signal
  fluctuation for lie group machine learning,'' in \emph{2019 International
  Radar Conference (RADAR)}, 2019, pp. 1--6.

\bibitem{Barbaresco2008}
------, ``Innovative tools for radar signal processing based on cartan’s
  geometry of spd matrices information geometry,'' in \emph{2008 IEEE Radar
  Conference}, 2008, pp. 1--6.

\bibitem{Pennec2006}
X.~Pennec, ``Intrinsic statistics on riemannian manifolds: Basic tools for
  geometric measurements,'' \emph{Journal of Mathematical Imaging and Vision},
  vol.~25, pp. 127--154, 07 2006.

\bibitem{gray2005toeplitz}
R.~M. Gray, ``Toeplitz and circulant matrices: a review,'' \emph{Foundations
  and Trends in Communications and Information Theory}, vol.~2, no.~3, pp.
  155--239, 2005.

\bibitem{Harandi2018}
M.~Harandi, M.~Salzmann, and R.~Hartley, ``Dimensionality reduction on spd
  manifolds: The emergence of geometry-aware methods,'' \emph{IEEE Transactions
  on Pattern Analysis and Machine Intelligence}, vol.~40, no.~1, pp. 48--62,
  2018.

\bibitem{Yang2020}
Z.~{Yang}, Y.~{Cheng}, H.~{Wu}, and H.~{Wang}, ``Enhanced matrix cfar detection
  with dimensionality reduction of riemannian manifold,'' \emph{IEEE Signal
  Processing Letters}, vol.~27, pp. 2084--2088, 2020.

\bibitem{CauthyThe}
S.-G. Hwang, ``Cauchy's interlace theorem for eigenvalues of hermitian
  matrices,'' \emph{American Mathematical Monthly}, vol. 111, 02 2004.

\end{thebibliography}
\end{document}